\documentclass[showpacs,amsmath,amssymb,twocolumn,superscriptaddress,notitlepage,preprintnumbers,pra]{revtex4-1}

\usepackage[dvips]{graphicx}
\usepackage{amsmath,amssymb,amsthm,mathrsfs,amsfonts,dsfont}
\usepackage{subfigure, epsfig}
\usepackage{braket}
\usepackage{bm}
\usepackage{enumerate}
\usepackage{color}
\usepackage{xcolor}
\usepackage{qcircuit}

\usepackage{graphicx}
\usepackage{algorithm}
\usepackage{algpseudocode}
\usepackage{hyperref}
\hypersetup{colorlinks=true, linkcolor=blue, citecolor=blue, urlcolor=black }

\newcommand{\tr}{\mathrm{Tr}}

\newcommand{\mc}[1]{\mathcal{#1}}

\newcommand{\revise}[1]{\textcolor{black}{#1}}

\usepackage{varioref}
\labelformat{equation}{(#1)}

\labelformat{section}{#1}

\labelformat{figure}{#1}

\newtheorem{theorem}{Theorem}
\newtheorem{proposition}{Proposition}

\newtheorem{condition}{Condition}
\newtheorem{corollary}{Corollary}

\usepackage{lineno}
\usepackage{graphicx}
\usepackage{algorithm}
\usepackage{algpseudocode}

\begin{document}

\title{Perturbative quantum simulation}

\date{\today}
% ------------  AUTHORS AND AFFILIATIONS ----------
\author{Jinzhao Sun}
\email{jinzhao.sun@physics.ox.ac.uk}
\affiliation{Center on Frontiers of Computing Studies, Peking University, Beijing 100871, China}
\affiliation{Clarendon Laboratory, University of Oxford, Parks Road, Oxford OX1 3PU, United Kingdom}
\affiliation{Quantum Advantage Research, Beijing 100080, China}

\author{Suguru Endo }
\affiliation{NTT Computer \& Data Science Laboratories, NTT corporation, Musashino, Tokyo 180-8585, Japan}

\author{Huiping Lin}
\affiliation{Center on Frontiers of Computing Studies, Peking University, Beijing 100871, China}
\affiliation{School of Computer Science, Peking University, Beijing 100871, China}

\author{Patrick Hayden}
\affiliation{Stanford Institute for Theoretical Physics, Stanford University, Stanford California 94305, USA}

\author{Vlatko Vedral}
\affiliation{Clarendon Laboratory, University of Oxford, Parks Road, Oxford OX1 3PU, United Kingdom}
\affiliation{Centre for Quantum Technologies, National University of Singapore, Singapore 117543,Singapore}

\author{Xiao Yuan}
\email{xiaoyuan@pku.edu.cn}
\affiliation{Center on Frontiers of Computing Studies, Peking University, Beijing 100871, China}
\affiliation{School of Computer Science, Peking University, Beijing 100871, China}

\affiliation{Stanford Institute for Theoretical Physics, Stanford University, Stanford California 94305, USA}

% --------------------  ABSTRACT  --------------------
\begin{abstract}
Approximation based on perturbation theory is the foundation for most of the quantitative predictions of quantum mechanics, whether in quantum many-body physics, chemistry, quantum field theory or other domains. Quantum computing provides an alternative to the perturbation paradigm, yet state-of-the-art quantum processors with tens of noisy qubits are of limited practical utility. Here, we introduce perturbative quantum simulation, which combines the complementary strengths of the two approaches, enabling the solution of large practical quantum problems using limited noisy intermediate-scale quantum hardware. The use of a quantum processor eliminates the need to identify a solvable unperturbed Hamiltonian, while the introduction of perturbative coupling permits the quantum processor to simulate systems larger than the available number of physical qubits.
We present an explicit perturbative expansion that mimics the Dyson series expansion and involves only local unitary operations, and show its optimality over other expansions under certain conditions. We numerically benchmark the method for interacting bosons, fermions, and quantum spins in different topologies, and study different physical phenomena, such as information propagation, charge-spin separation, and magnetism, on systems of up to $48$ qubits only using an $8+1$ qubit quantum hardware. We experimentally demonstrate our scheme on the IBM quantum cloud, verifying its noise robustness and illustrating its potential for benchmarking large quantum processors with smaller ones.
% To demonstrate the feasibility of our scheme, we use $5$ physical qubits on the IBM quantum cloud to experimentally simulate the $8$-qubit Ising model, verifying the noise robustness of our method and illustrating its potential for benchmarking large quantum processors with smaller ones.
\end{abstract}

\maketitle

A universal quantum computer can naturally simulate the real-time dynamics of any closed finite dimensional quantum system~\cite{RevModPhys.86.153}, a challenging task for classical computers.
While there has been tremendous progress in quantum computing hardware development, including the landmark quantum supremacy/advantage experiments with superconducting and optical systems~\cite{arute2019quantum,ZhongBosonSampling20, PhysRevLett.127.180502, PhysRevLett.127.180501}, state-of-the-art quantum hardware can still only control tens of noisy qubits~\cite{arute2019quantum,Gong_2021,PhysRevLett.127.180501,mi2021information}. That is insufficient for the implementation of fault-tolerant universal quantum computing, which {requires $10^3$ or more physical qubits per logical qubit to suppress the physical error}~\cite{campbell2017roads}. It is more pragmatic in the near term to focus on the noisy intermediate-scale quantum (NISQ) regime and utilise hybrid methods, which run a shallow circuit without implementing full error correction~\cite{preskill2018quantum}. Nevertheless, most quantum simulation algorithms, whether targeting NISQ or universal quantum computers, generally entails a number of physical or logical qubits no smaller than the problem size~\cite{mcardle2018quantum,cerezo2020variationalreview,bharti2021noisy}. 
Given that large-scale fault-tolerant quantum computers do not yet exist and there will be significant size constraints even on NISQ devices for the foreseeable future, a pressing question is how to solve large practical problems with limited quantum devices~\cite{altman2019quantum,bauer2020quantum}.

One possibility is to leverage the classical methods that have been developed to solve quantum many-body problems, wherein the most successful one is perturbation theory.  This method divides the Hamiltonian into a major but easily solved component and a weak but potentially complicated counterpart, in which case the full dynamics can be expressed as a series expansion~\cite{abrikosov2012methods,fradkin2013field,rigol2007numerical,rigol2006numerical,wilson1985diagrammatic,rohringer2018diagrammatic}. 
However, the ability to solve the major component and compute the higher-order expansions limits the use of perturbation theory in classical simulation of general many-body problems.

Here, we propose perturbative quantum simulation (PQS),
which directly simulates the major component on a quantum computer while perturbatively approximates the weak interaction component.
Since there is no assumption on the size or interaction of the major component, PQS potentially goes beyond the conventional perturbative approach, and it could simulate classically challenging systems, such as large systems with weak inter-subsystem interactions or intermediate systems with general interactions. 
Compared to universal quantum computing, PQS has limited power for arbitrary problems; yet, the perturbative treatment of the weak component greatly reduces quantum resources compared with conventional quantum simulation. Notably, PQS runs a shallower circuit with fewer qubits, making it more noise-robust and thus useful in benchmarking large quantum devices with smaller ones. 
Our experiments on the IBM quantum devices demonstrate a significant improvement of the simulation accuracy over direct simulation. 

For eigenstate problems, there are considerable hybrid schemes that combine different classical methods, such as density matrix embedding theory~\cite{DMET2012,DMET2013,Wouters16,rubin2016hybrid}, dynamical mean field theory~\cite{Troyer16,rungger2020dynamical,Chen_2021}, 
density functional theory embedding~\cite{Ivano21}, quantum defect embedding theory~\cite{Ma20, sheng21},
tensor networks~\cite{barratt2021parallel,yuan2020quantum}, entanglement forging~\cite{PRXQuantum.3.010309,huembeli2022entanglement}, virtual orbital approximation~\cite{PhysRevX.10.011004}, quantum Monte Carlo~\cite{Li_2019,Huggins_2022,2022arXiv220514903X,2022arXiv220509203M,2021arXiv211202190P,2022arXiv220610431Z}, etc. Our work instead focuses on the different but meaningful dynamics problem, which is based on perturbation theory and fundamentally from existing ones with different assumptions, limitations, or applications~\footnote{The assumptions of the embedding methods and comparisons to our method are discussed in \cite{NoteX}}.\\

\begin{figure*}[htb!]
\centering
\includegraphics[width =1.0\textwidth]{ 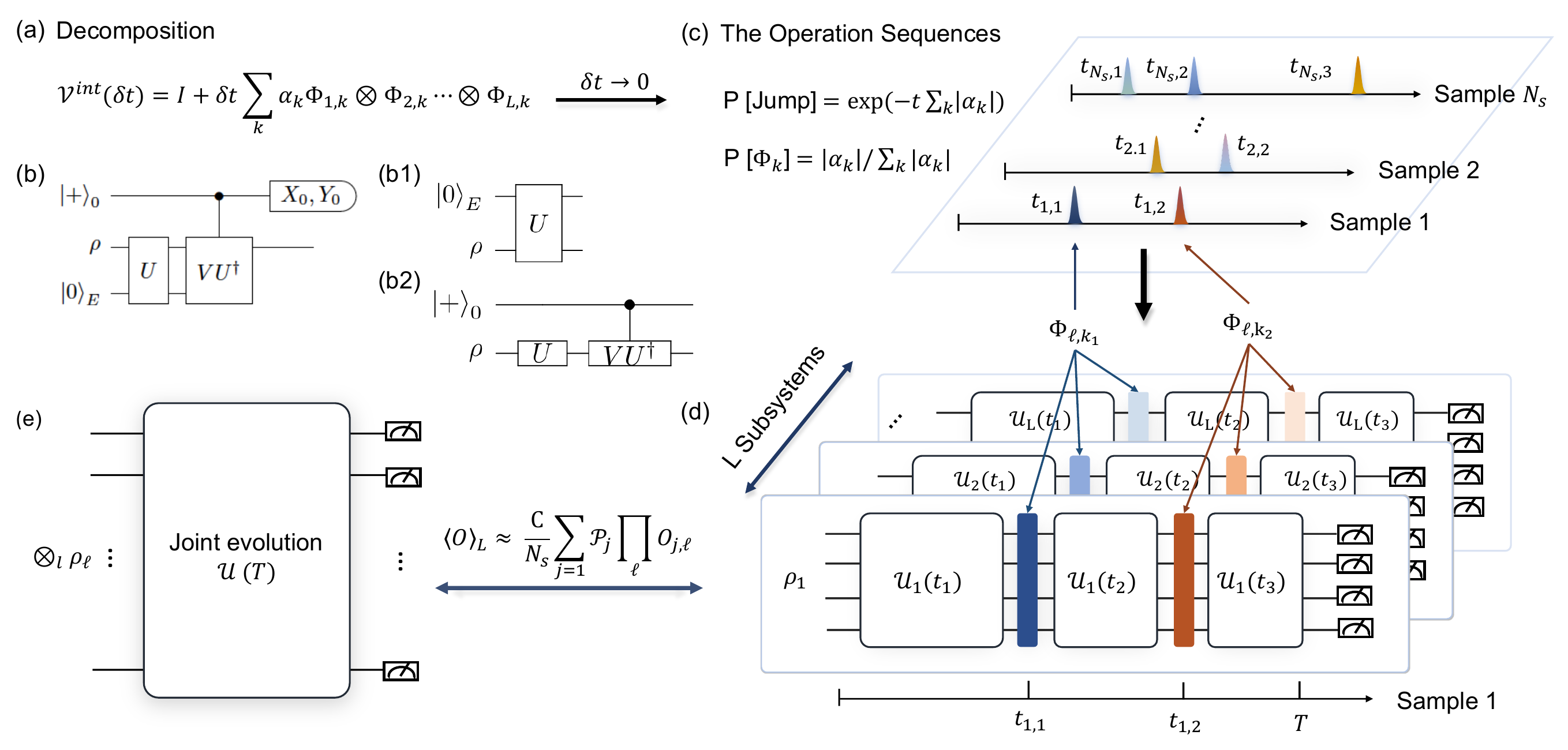}
 \caption{Schematic diagram of perturbative quantum simulation. {(a)}~The decomposition of interactions into local generalised quantum operations.  
 $\Phi_k = \Phi_{1,k} \otimes  \Phi_{2,k} \otimes \cdots \otimes  \Phi_{L,k} $ 
 with $\Phi_{l,k} = \tr_{E}[  U_{lE_l}(\rho_l \otimes \ket{0}\bra{ 0}_{E_l}) {V_{lE_l}}^\dag] $ acting on the $l$th subsystem.  
 {(b)}~The implementation of a generalised quantum operation $\Phi(\rho) = \tr_{E}[  \mathbf U(\rho \otimes \ket{ 0}\bra{ 0}_{E}) {\mathbf V}^\dag]$ using quantum circuits, which reduces to a quantum channel when $\mathbf U = \mathbf V$ in (b1) and unitary operations $\Phi(\rho) = U\rho V^{\dagger}$ when there is no ancilla in (b2).
 {(c)}~We can equivalently realise the discretised scheme with $\delta t \rightarrow 0$.  The operation sequences for $N_s$ samples are predetermined provided the decomposition.
 We either continuously apply the local time evolution or randomly apply a generalised quantum operation. 
 The time to apply the operations is determined by the probability $\textrm{P[Jump]}$ and the $k$th operation is applied with probability $\textrm{P}[\Phi_k]$.
%  (see \autoref{appendix:sub_discrete_time} of Supplementary Materials). 
The scheme in (c) does not assume time discretisation, and on average it is equivalent to the discretised one. The 
number of generalised quantum operations that applies during the evolution scales as $\mathcal{O}(\sum_k |\alpha_k|T)$.
 {(d)}  Schematic of the simulation process for 1 sample as an  example. For the $l$th subsystem, we evolve the state under local operations $\mathcal{U}_l$ and apply operations $\Phi_{l,k_1}$ and $\Phi_{l,k_2}$ at time $t_{1,1}$ and $t_{1,2}$, respectively.
We measure the output states with a product observable $O$ and obtain the outcomes $O_{j,l}$ for the $j$th sample.  We repeat the process for $N_s$ times. For any product input state, the expectation value of observable $O$ under the joint evolution $\mathcal{U}$ in (e) can be unbiasedly approximated by $\braket{O}_L \approx \frac{C}{N_s}\sum_{j=1}^{N_s} \mathcal{P}_j \prod_l O_{j,l}$ with the overhead $C$ and phase $\mathcal{P}_j$ determined by the decomposition.   
 More details can be found in \autoref{appendix:sub_discrete_time} and \ref{appendix:sub_MonteCarlo} in Supplementary Materials.
%  The true evolution in a manifold can be realised with  generalised local quantum operation in a small manifold.
 }
 \label{fig:schematic}
\end{figure*}

{\emph{Background.---}}\revise{
We consider to simulate the dynamics of a quantum many-body system.
Suppose the whole system is divided into $L$ subsystems according to topological structures or degrees of freedom, like the clustered molecules \footnote{\revise{For instance, the lattice Hamiltonian with multiple degrees of freedom and molecular systems with virtual and active orbitals or in a clustered structure (see Refs.~\cite{bauer2020quantum,garlatti2017portraying,rigol2007numerical} and Sec.~V of Supplemental Materials for more discussions) provide a natural partitioning strategy for the subsystems, similarly as that in perturbation theory}}, the Hamiltonian is  
$H=H^{\textrm{loc}} + V^{\textrm{int}}$, where
$H^{\rm loc}=\sum_l H_l$ is the local strong interaction with each $H_l$ acting on the $l$th subsystem, and $V^{\textrm{int}} =  \sum_j \lambda_jV_j^{\rm int}$
is the weak perturbative interaction between the subsystems. Here $V_j^{\rm int}$ are different types of interactions with real amplitudes $\lambda_j$. 
}

% where $H^{\rm loc}$ is the strong but local interaction and $V^{\textrm{int}}$ is weak but complex perturbations. In practice, 
% Therefore, we could consider 

{To simulate the dynamics of $U(t) = e^{-iHt}$, a representative perturbation treatment is via Dyson series expansion as
\begin{equation}
\label{eq:DysonExpansion}
    U(t)=1-{i} \int_{t_{0}}^{t} d t_{1} e^{{i} H^{\textrm{loc}}\left(t_{1}-t_{0}\right)} V^{\textrm{int}} e^{-{i} H^{\textrm{loc}}\left(t_{1}-t_{0}\right)} + \dots
\end{equation}
\revise{Then $U(t)$ becomes a linear combination of  trajectories consisting of different sequential unitary operators.
% , where each one undergoes local Hamiltonian evolution $e^{{i} H_lt}$ with interactions $V^{\textrm{int}}$ uniformly inserted during the evolution time.
}
When the local Hamiltonians $\{H_l\}$ are solvable, one can further represent the expansion graphically, such as via Feynman diagrams,  and compute expectation values of the time evolved state with different graphs corresponding to different expansion terms.
}
A major limitation of perturbation theory is the assumption of the simple hence solvable local Hamiltonians, which fails when $\{H_l\}$ become strongly correlated, as that happens in realistic systems.
% { In this case, it becomes hard to realise the local evolution $e^{{i} H_{lt}}$. One might propose to consider a different partition where the complex part of local Hamiltonians is assigned to the interaction, which, however, might lead to strong interaction and the expansion becomes inefficient or divergent.}
\revise{Indeed, even if no interaction under certain partitioning strategy with 
$V^{\rm int}=0$, no classical methods exist that can efficiently simulate the dynamics of general Hamiltonian $H^{\rm loc}=\sum_l H_l$, otherwise the computational class of bounded-error quantum polynomial time collapses. 
In the following, 
% leveraging a quantum computer that efficiently simulates local Hamiltonians, 
we introduce the framework of PQS, based on which we propose an explicit algorithm mimicking Dyson series expansion and show its optimality over more general theories. 
% leverages a quantum computer to simulate local Hamiltonians. We propose an explicit algorithm mimicking Dyson series expansion, and show its optimality over more general theories. 
}\\

{\emph{Framework.---}}Here, we focus on general ways that realise the joint time evolution channel $\mc U(\rho, T) = U(T)\rho U^\dag(T)$ by applying only local operations on each subsystem separately. 
To do so, we first introduce the concept of local \emph{generalised quantum operations}
\begin{equation}\label{Eq:goperation}
    \Phi(\rho) = \tr_{E}\left[ \mathbf U\left(\rho \otimes \ket{\mathbf 0}\bra{\mathbf 0}_{E}\right) {\mathbf V}^\dag\right].
\end{equation}
Here we denote ancillary states $\ket{\mathbf 0}\bra{\mathbf 0}_E = \ket{0}\bra{0}_{E_1}\otimes \cdots \otimes \ket{0}\bra{0}_{E_L}$ and unitary operators $\mathbf U = U_{1E_1}\otimes  \cdots \otimes U_{LE_L}$ and $\mathbf V= V_{1E_1}\otimes  \cdots \otimes  V_{LE_L}$, where $U_{jE_j}$ and $V_{jE_j}$ represents the operators acting only on the subsystem $j$ and the $j$th ancilla.
While the operation $\Phi(\rho)$ is nonphysical in general, it can be realised effectively using local operations and postprocessing (see \cite{NoteX}). 
Note that $\Phi(\rho)$ reduces to local quantum channels when $\mathbf U = \mathbf V$.
The key idea of PQS theories is to decompose the joint evolution into a set of generalised quantum operations, which separately act on each subsystem.
By  choosing a {spanning} set of $\{\Phi_k\}$ properly,  an infinitesimal evolution governed by the interaction $\mc V(\delta t)[\rho] = V^{\textrm{int}}(\delta t)\rho V^{\textrm{int}}(\delta t)^\dag$ can be decomposed as
\begin{equation}\label{Eq:Vdecommain}
    \mc V(\delta t)[\rho] = \mathcal{I}(\rho) + \delta t\sum_k \alpha_k \Phi_k(\rho)
\end{equation}
where $V^{\textrm{int}}(\delta t) = e^{-iV^{\textrm{int}}\delta t}$ represents the interacting unitary operations within  duration $\delta t$, and $\mathcal{I}$ is the identity operation.

{
Next, we consider a Trotterised joint evolution as $\mathcal{U}(T)=[\mc V(\delta t)\circ\bigotimes_l\mc U_l(\delta t)]^{T/\delta t}$. Using the decomposition in  \autoref{Eq:Vdecommain}, we can then expand  $\mathcal{U}(T)$ as a series of different trajectories. Here, each trajectory is defined by which operations act at each time, including the local time evolution $\mc U_l(\delta t)$ of each subsystem and one of the generalised quantum operation $\Phi_k(\rho)$ that on average emulates the nonlocal effect of   $\mc V^{\rm int}$. 
The whole evolution $\mc U(T)$ is now decomposed as a linear combination of local operations that act separately on each subsystem, which can be effectively realised in parallel. 
The expectation value of an arbitrary state can be obtained from local measurement results (see Sec. IB  in \cite{NoteX} for the derivation and implementation).

The above discretised scheme assumes a small discrete timestep and requires to apply the operations at each time step $\delta t$, which is unnecessary since the effect of the weak interacting operation $\mc V(\delta t)$ in a short time is close to the identity.
We address this problem by stochastically applying the operation $\Phi_{{k}}$ depending on the amplitude of its associated coefficient $|\alpha_k|$.
Taking a short time limit $\delta t\rightarrow 0$, we generate each trajectory according to the decomposition in   \autoref{Eq:Vdecommain} and stochastically realise the joint time evolution with operations separately acting on each subsystem.
The average of different trajectories reproduces the joint dynamics under $\mc U(T)$. 
Note that the number of generalised quantum operations required to  realise the joint  evolution scales proportionally to the interaction strength as $\mathcal{O}(\sum_k |\alpha_k|T)$, and  on average the stochastic implementation scheme is proven to be equivalent to the discretised scheme (See Sec. IC in \cite{NoteX}).
}

 We summarise the key steps of PQS in Algorithm~\ref{alg1_main} and illustrate the procedure in Fig.~\ref{fig:schematic}.

 \begin{algorithm}[H]
 \begin{algorithmic}[1]
 \State{Given a set of generalised quantum operations, find the decomposition  \autoref{Eq:Vdecommain}.}
 \State {Generate a sequence of trajectories where  each subsystem evolves and experiences random local generalised quantum operations.}
 \State{Sample from the trajectories.
     The average behaviour reproduces the joint evolution.}
 \end{algorithmic}
 \caption{Perturbative quantum simulation} 
 \label{alg1_main}
 \end{algorithm}

By applying our algorithm, the whole simulation process is now decomposed into the average of different ones, each of which only involves operations on the subsystems. Thus, we can effectively simulate $nL$ qubits with operations on subsystems with only $\mc O(n)$ qubits, and this also offers noise robustness of our method (see Sec.~VII in \cite{NoteX}).
Note that local dynamics $\mc U_l(t)$ could be implemented with any Hamiltonian simulation methods, such as product formulae~\cite{childs2017toward2,childs2019theory} or quantum signal processing~\cite{low2019hamiltonian,PhysRevLett.118.010501}, and our algorithm is compatible with both near-term and fault-tolerant quantum computers. \\

\begin{figure*}[ht]
\includegraphics[width =.95\linewidth]
{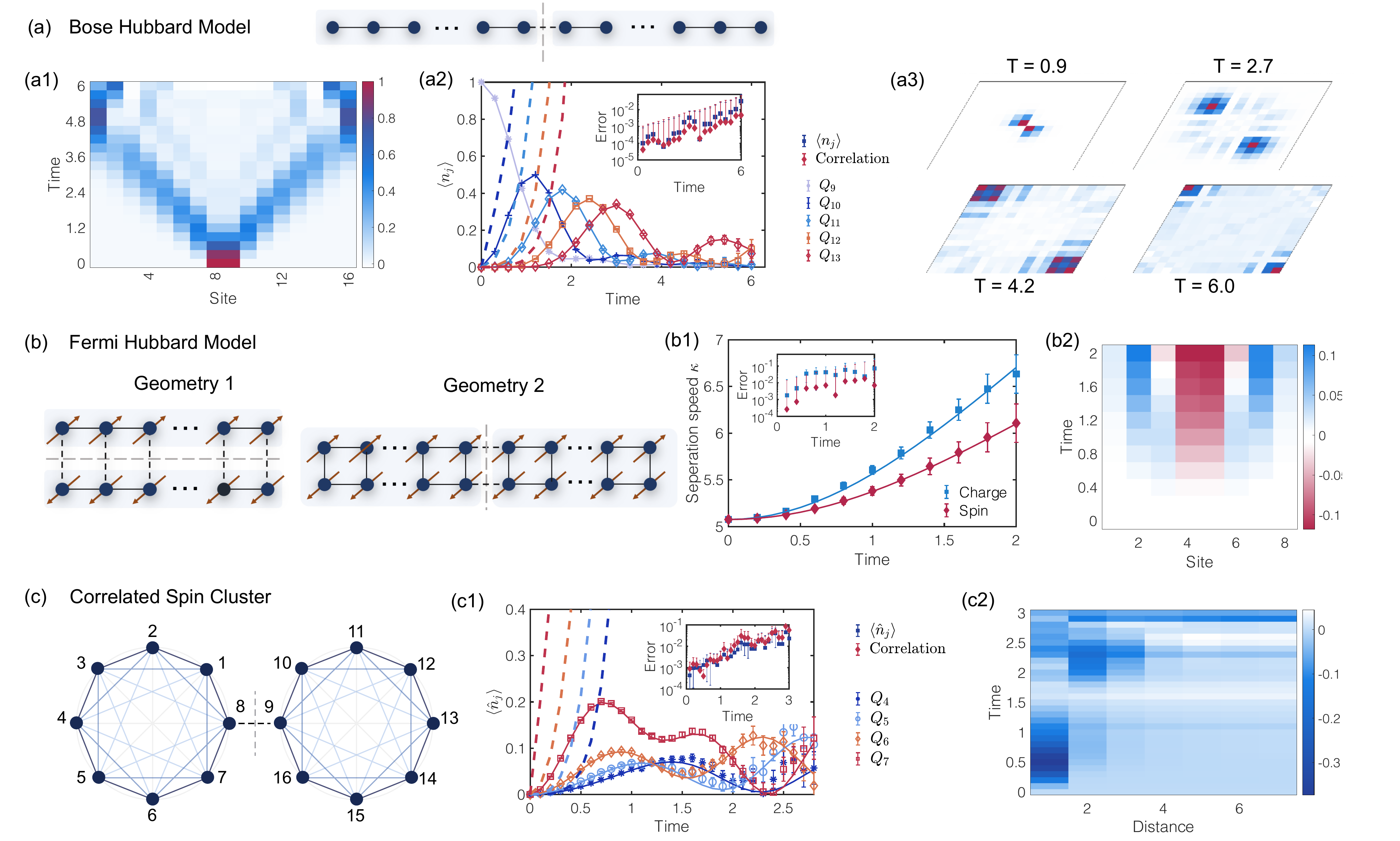}
 \caption{Dynamics simulation of interacting (a)~bosons, (b)~fermions, and (c) quantum spin systems with different topologies. 
Gray dashed lines in site-edge diagrams manifest subsystem partitioning for PQS. We group $8$ qubits as a subsystem to simulate $N=16$-qubit quantum systems using most $5 \times 10^5$ samples.  Solid lines represent exact results from direct simulation.
 {(a)} Quantum walk (QW) of spinless bosons on a $1$D array in the large onsite repulsion limit (see Sec.~\uppercase\expandafter{\romannumeral6} in \cite{NoteX} for the Hamiltonian). Two identical bosons are initially excited at the centre. 
 (a1) Density spreading $\langle \hat{n}_j \rangle = \langle \hat{b}_j^{\dagger} \hat{b}_j\rangle $ with bosonic operators $\hat{b}$ under time evolution. (a2) The density distribution at site $9$ - ${13}$ (labelling from left to right). The nearest-neighbour Lieb-Robinson bounds (dashed) capture the density spreading \cite{bravyi2006lieb,yan2019strongly,Gong_2021}. The inset shows the errors for the average density and the average density-density correlator $\hat{\rho}_{ij} = \langle \hat{b}_i^{\dagger}\hat{b}_j^{\dagger} \hat{b}_i\hat{b}_j \rangle$ with respect to the exact results. 
 (a3) Boson spatial antibunching in QW. The normalised correlator $\hat{\rho}_{ij}/\hat{\rho}_{ij}^{\max}$ at different $t$~\cite{lahini2012quantum,Cazalilla.RevModPhys.83.1405}.
 {(b)} Separation of charge and spin density (CSD) in a $1$D Fermi-Hubbard model $H= -J \sum_{j, \sigma }  \left( \hat{c}_{j, \sigma}^{\dagger}  \hat{c}_{j+1, \sigma}+\mathrm{h.c.}  \right)  +U \sum_{j} \hat n_{j, \uparrow} \hat n_{j, \downarrow}+\sum_{j,\sigma }   h_{j, \sigma} \hat n_{j, \sigma}$ ($\hat c_{j,\sigma}$, $\hat c_{j,\sigma}^{\dagger}$: fermionic  operators with spin $\sigma$, $U=J=0.5$)~\cite{Cazalilla.RevModPhys.83.1405}. Left: two partitioning strategies for small and large on-site potential $U$. The initial state is the ground state of a non-interacting Hamiltonian with quarter filling ($N_{\uparrow} = N_{\downarrow}$ $= 2$), in which the CSD are generated in the middle of the chain at $t = 0$~\cite{arute2020observation,KivLinearDepth,PhysRevApplied.9.044036}.
 (b1) The separation of charge (blue square) and spin (red diamond) densities.
We characterise the separation speed from the middle as $\kappa_{{\pm}} = \sum_{j=1}^{N} \left| j - (N+1)/2\right| (\langle \hat{n}_{j, \uparrow} \rangle\pm \langle \hat{n}_{j, \downarrow} \rangle)$ for charge ($+$) and spin ($-$) degrees of freedom with $\langle \hat{n}_j \rangle = \langle \hat{c}_j^{\dagger} \hat{c}_j\rangle$ ($N=8$). The inset shows the errors under evolution.
(b2) The difference of CSD under evolution. The relative separation is initially set as $0$.
{(c)} Information propagation of correlated Ising spin clusters with power law decay interactions $H_l^{\rm loc} = \sum_{ij} J_{ij} \hat\sigma_{l,i}^x \hat\sigma_{l,j}^x +h \sum_j \hat \sigma_{l,j}^z$  ($J_{ij}=   |i-j|^{-1}$) in the subsystems and interaction $V^{\rm int} =  \hat \sigma_{1,N}^x \hat\sigma_{2,1}^x$  on the boundary. The initial state is prepared as $\ket{\psi_0} = \hat \sigma_8^x \ket{{0}}^{\otimes N}$.  (c1) The signal of quasiparticle excitations at different sites, where the propagation is faster than the nearest-neighbour Lieb-Robinson 
velocity (dashed)~\cite{bravyi2006lieb,jurcevic2014quasiparticle,monroe2021programmable}.
% , the nearest-neighbour Lieb-Robinson bounds (dashed) does not capture all the signal for this propagation. 
(c2) The dynamics of the correlation function $C_d = \langle \hat{\sigma}_8^z\hat{\sigma}_{8+d}^z \rangle-\langle \hat{\sigma}_8^z \rangle\langle \hat{\sigma}_{8+d}^z \rangle$. The inset in (c1) shows the errors for the averaged quasiparticle excitations density and correlation functions. 
 }
\label{fig:main}
\end{figure*}

{\emph{Explicit protocol.---}}While the decomposition of  \autoref{Eq:Vdecommain} holds in general for an (over)complete set of $\{\Phi_k\}$, it may involve difficult-to-implement operations in  experiments. Here, we address this problem by developing an explicit decomposition with only local unitary operations. Specifically, we consider a natural expansion of $\mc V(\delta t)$ as
\begin{equation}
    \mc V(\delta t)[\rho] = \mathcal{I}(\rho) - i\delta t \sum_j \lambda_j(V_j^{\textrm{int}} \rho - \rho V_j^{\textrm{int}}),
    \label{eq:explicit}
\end{equation}
where all $V_j^{\textrm{int}}$ are tensor products of unitaries, and hence each term $\mathcal{I}(\rho)$, $V_j^{\textrm{int}} \rho$, or $\rho V_j^{\textrm{int}}$ corresponds to a specific generalised quantum operation. 
We emphasise that the expansion only involves unitary operations, and avoids the computational cost in diagrammatic perturbation theory, which greatly simplifies the implementation.
% Moreover, the expansion is universal and avoids the computational cost in diagrammatic perturbation theory~\cite{abrikosov2012methods,rigol2007numerical,wilson1985diagrammatic}. 
We further prove in Theorem~3 in \cite{NoteX} that the explicit decomposition corresponds to the infinite-order Dyson series expansion~\footnote{\revise{We discuss the implementation with a truncated expansion on a quantum computer in Sec.~\uppercase\expandafter{\romannumeral3} in \cite{NoteX}}}.

% of the joint unitary $U(t)$,
% %  \autoref{eq:DysonExpansion},
% and it indeed effectively implements each expanded term (trajectory) with a quantum computer and sums over the expansion via the average of different trajectories~\footnote{\revise{We discuss the implementation with a truncated expansion on a quantum computer in Sec.~\uppercase\expandafter{\romannumeral3} in \cite{NoteX}}}.

\begin{figure}[t]
\includegraphics[width =1.0\linewidth]
{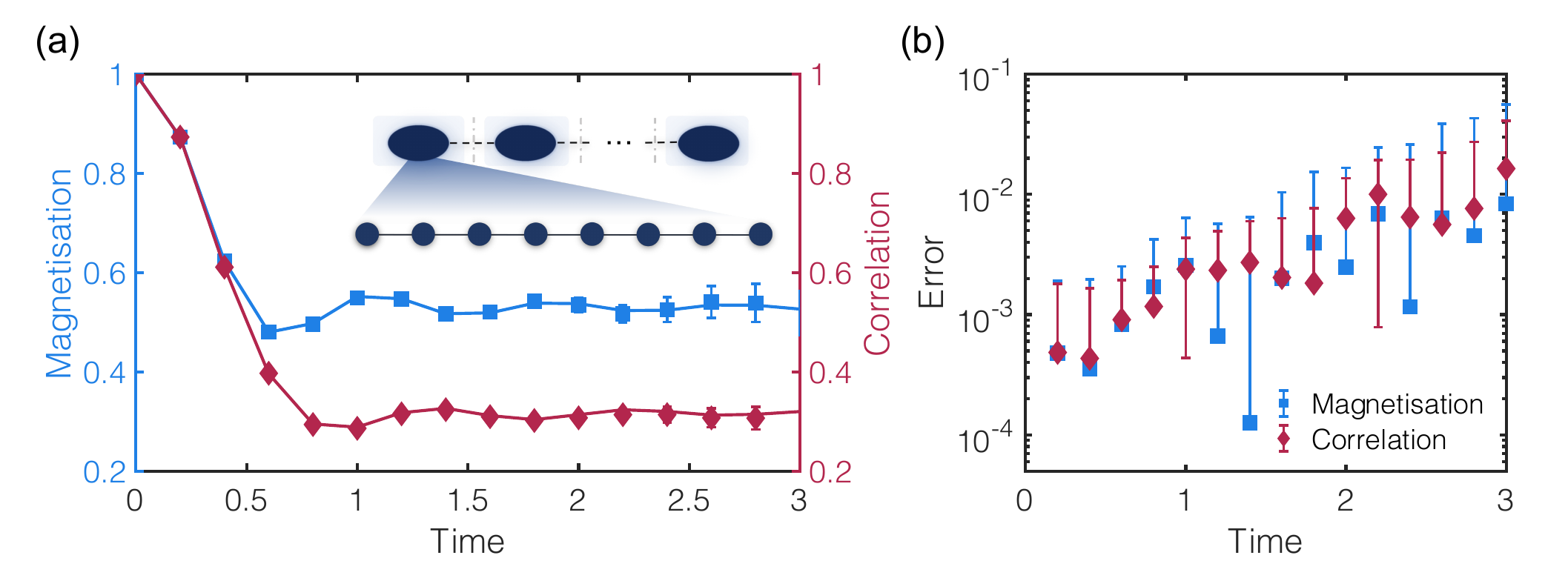}
\caption{{Dynamics simulation of $1$D $48$-site spin chains.} The subsystem and interaction Hamiltonians are $H_l^{\rm loc} = \sum_{i}  \hat\sigma_{l,i}^x \hat\sigma_{l,i+1}^x + \sum_i \hat\sigma_{l,i}^z$ and $ V_l^{\rm int} =  f_l \hat \sigma_{l,N}^x \hat\sigma_{l+1,1}^x$, respectively, and 
the interactions on the boundary are randomly generated from $[0,J/2]$. 
(a) The average magnetisation (in blue) $\frac{1}{N} \sum_{i} \braket{\hat{\sigma}_i^z}$ and nearest-neighbour correlation function (in red) $\frac{1}{N-1} \sum_{i} \braket{\hat{\sigma}_i^z \hat{\sigma}_{i+1}^z}$, compared with the TEBD method as a benchmark. The inset illustrates the geometry of the spin systems and the partitioning strategy where we group 8 adjacent qubits as subsystems. { (b)} The errors for the average magnetisation and correlation using $5 \times 10^5$ samples.
}
 \label{fig:main_corr}
\end{figure}

\begin{figure}[t]
\centering
\includegraphics[width =1\linewidth]
{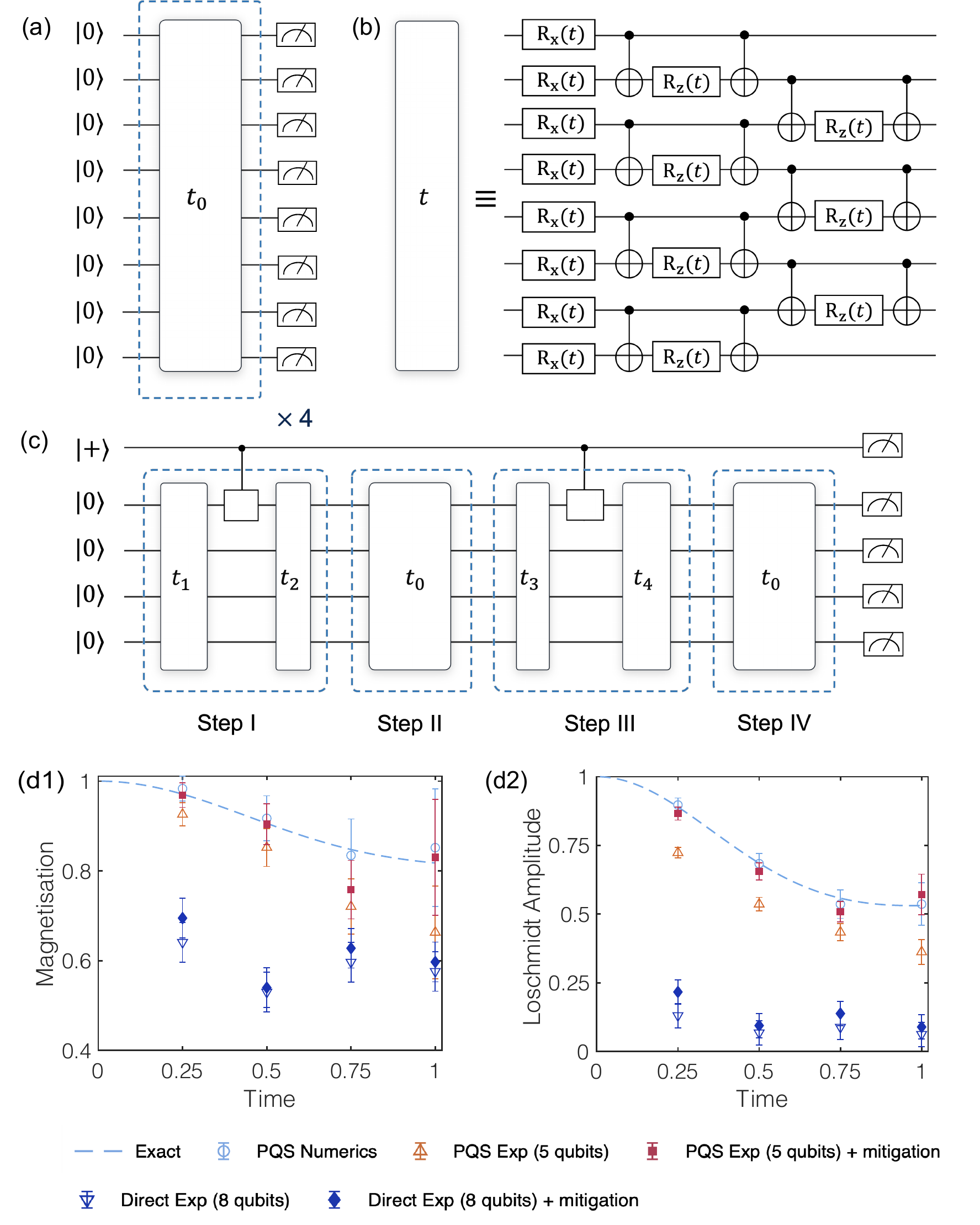}
 \caption{
Implementation and experimental results of the dynamical phase transition of 8 interacting spins.
The initial state $\ket{\psi_0}=\ket{0}^{\otimes 8}$ is evolved under $8$-site Ising Hamiltonian $H = \sum_j \hat{\sigma}_j^z\hat{\sigma}_{j+1}^z + 0.5 \sum_j \hat{\sigma}_j^x$ with $T=1$. {(a)}~Quantum circuit for $8$-qubit simulation based on first-order Trotterisation with four steps $t_0=1/4$. {(b)} The circuit block for a single-step evolution for time $t$ with  parallelisation. {(c)}~An example for the implementation of PQS to simulate $8$-qubit system with operations on $4+1$-qubit. The circuit blocks are similar as that in (b) with 4 qubits.
When a generalised operation is inserted into a Trotter step, we divide the step into two evolution and insert the operation between that. 
{(d)}~The magnetisation and Loschmidt amplitude in ferromagnetic and paramagnetic phases.  Here, the Loschmidt amplitude $\mathcal{G}(t) = \left| \braket{\psi_0 |e^{-iHt}|\psi_0}\right|^2$ characterises the dynamical echo back to the initial state~\cite{PhysRevLett.119.080501}, as an indicator of dynamical phase transition when it decreases to $0$.
We compare the results of exact simulation (dashed line), PQS (numerics, circle), PQS using IBMQ (5 qubits in (c), upper triangle) and the direct simulation using IBMQ (8 qubits in (a), lower triangle). We also show the results using  measurement error mitigation for PQS (solid square) and direct simulation (solid diamond). 
% We refer to Sec.~\uppercase\expandafter{\romannumeral7} in \cite{NoteX} for more experiment results.
 }
 \label{fig:main_IBM}
\end{figure}

Implementing the interaction $\mc V$ perturbatively using generalised quantum operations
introduces a sampling overhead $C$. Specifically, when measuring the output state of the perturbatively simulated state, 
the measurement accuracy is $\varepsilon = \mc O(C\sigma/\sqrt{N_s})$ given $N_s$ samples in contrast to $\varepsilon = \mc O(1/\sqrt{N_s})$ in direct simulation. Here, $\sigma$ is the standard deviation introduced from the expansion, normally less than~$1$. 
Assuming the general decomposition of  \autoref{Eq:Vdecommain}, the overhead is $C=e^{\sum_k|\alpha_k|T}$.
Different decomposition of  \autoref{Eq:Vdecommain} would lead to different coefficients and hence different overhead.
We further prove that the explicit decomposition in  \autoref{eq:explicit} has the minimal simulation cost, provided that the Pauli operators of each $V_i$ satisfy a certain mild condition (see Theorem 2 in \cite{NoteX} for the proof of optimality and illustrative examples here in~\footnote{For example, the condition holds when $V^{\rm int} = \lambda_1 X_aY_bZ_c + \lambda_2 Y_aZ_bX_c + \lambda_3 Z_aX_bY_c$ with Pauli operators $X$, $Y$, and $Z$ acting on three subsystems $a$, $b$, and $c$.}).
Since the overhead increases exponentially with $\lambda_T=\sum_i|\lambda_i|T$, PQS cannot simulate arbitrary systems with strong $V^{\rm int}$ or long time $T$.
To have a reasonable overhead $C$, the algorithm is efficient when $\lambda_T = \mathcal{O} (1)$, aligning with the spirit of perturbation theory.
Yet, the overhead is  independent of the initial state, size, and interaction strengths of the subsystems. With a constant $\lambda_T$, PQS can be applied to study intricate quantum many-body systems with strong subsystem interactions. 
% Remark that compared with other recent works that simulate clustered Hamiltonians or circuits using gate decomposition~\cite{peng2020simulating,barratt2020parallel,mitarai2021constructing,fujii2020deep,yuan2020quantum,Mitarai2021overheadsimulating}, PQS directly simulates the interacting dynamics with a much smaller overhead.
As shown shortly, PQS can be used to  probe interesting physical phenomena directly, benchmark NISQ processors, simulate large quantum circuits, etc.\\

% measurement counts are bounded by exponentials in the interaction strength and time, but are independent of the initial state, size and interactions of the subsystems.

% Since the overhead Thus PQS cannot simulate arbitrary systems with strong $V^{\rm int}$.
% While PQS cannot simulate arbitrary systems, 
% Yet, PQS is efficient given samll  could tackle general strongly interacting subsystems, indicating a potential advantage over classical perturbative simulation.

{\emph{Numerical and experimental results.---}}We apply PQS to study many-body physical phenomena in different systems
% in interacting bosons, fermions, and quantum spin systems 
with different topological structures. 
As shown in Fig.~\ref{fig:main}, we investigate (a) the quantum walk of bosons on a one-dimensional lattice, (b) the separation of charge and spin excitations of fermions with two-dimensional topology, 
and (c) the correlation propagation of quantum spin systems of two clusters. 
We design appropriate partitioning strategies, in which the whole system consists of two subsystems and each subsystem consists of $8$ qubits.
In each example, we present the corresponding task-specific partitioning strategy of the quantum systems.
Using the explicit decomposition strategy, we exploit $8+1$ qubits to simulate each subsystem and classically emulate the quantum system with numerical results shown in Fig.~\ref{fig:main}. All unique features are detected just as we directly simulate the whole system. Indeed, the numerical results align with those of the exact simulation, thus verifying the reliability of the theory. 
We refer to Sec.~\uppercase\expandafter{\romannumeral6} in \cite{NoteX} for other physical systems, including the long-range spin chains, and simulation details.

These numerical tests are restricted to 16 qubits since the exact simulation of larger quantum systems becomes exponentially costly.
To benchmark PQS for larger systems, we investigated a $1$D $48$-site spin chain with  nearest-neighbour correlations, using the time-evolving block decimation (TEBD) method with matrix product states as the reference.
As shown in Fig.~\ref{fig:main_corr}, our simulation results coincide with those of TEBD, which again verifies the reliability of PQS for simulating multiple subsystems. Intriguingly, PQS only needs to manipulate $8+1$ qubits to recover the joint dynamics of the $48$-qubit system. 

We only consider the time evolution of small and classically simulable quantum systems for benchmarking our method.  
However, for all the examples considered here, since the simulation cost is independent of the interaction and initial states of the subsystems, PQS also works when tackling a much larger subsystem with more complicated subsystem interactions. 
In practice, when we increase the subsystem size to around $n=50$ qubits and consider general strong interactions, 
PQS could outstrip the capabilities of classical simulation and reliably probe properties of quantum systems with a small-size quantum processor.

In contrast to direct simulation,  PQS could also be more robust to noise attributed to the reduction of quantum sources~\footnote{\revise{A smaller quantum device is usually much more accurate than a larger quantum processor due to crosstalks or other types of errors when controlling large quantum systems. Our method could thus serve as a benchmark of the computing result for large-scale problems.}}.
To verify such an advantage, we  experimentally study the dynamical phase transition of an $8$-qubit Ising model with nearest-neighbour correlations on IBMQ hardware. By dividing the system into two subsystems, we use a $4+1$-qubit processor to implement our PQS algorithm and compare the results with conventional direct simulation with $8$ qubits, as shown in Fig.~\ref{fig:main_IBM}. For a total evolution time $T = 1$,  a first-order Trotterisation is used, which has four steps and a negligible Trotter error.
Fig.~\ref{fig:main_IBM}(d1,d2) show the magnetisation and Loschmidt amplitude in ferromagnetic phases.
The experimental results clearly demonstrate PQS achieves higher simulation accuracy than direct simulation. 
It is also found that with  measurement error mitigation, PQS approaches the exact result~\cite{bravyi2021mitigating}, and outperforms direct simulation consistently. 
More experimental results and detailed discussions on the \revise{implementation and noise robustness of PQS} can be found in Sec.~\uppercase\expandafter{\romannumeral7} in \cite{NoteX}.

{\emph{Conclusion and discussion.---}}Our theoretical, numerical, and experimental results indicate that quantum simulation and perturbation theory are not only compatible but complementary.
The PQS algorithm leverages quantum computers to simulate the major component of the Hamiltonian, alleviating the constraint of a classical perturbation method, and uses classical perturbation to approximate the interaction, circumventing limited quantum resources in near-term or early-stage fault-tolerant quantum computers.
\revise{
Since PQS is a hybrid method that combines quantum computing and classical perturbation theory, it inherits their advantages as well as their limitations, such as high-dimensional systems with strong correlations $V^{\rm int}$ and long time $T$.
% Thus, PQS along cannot simulate arbitrary systems, such as large and strongly correlated 3D systems.
% The major limitation of PQS thus comes from the assumption on weak perturbation.
% that of classical perturbation theories, which generally only work for weakly interacting systems. 
Yet, PQS is applicable to
intermediate-size systems, such as a square lattice with tens to hundreds of qubits, and it is particularly useful for large systems with weak inter-subsystem interactions, such as (quasi) one-dimensional  systems and clustered subsystems.}
Our numerical and experimental results demonstrate wide applicability of PQS methods for studying new physical phenomena, and its potential application in benchmarking large quantum processors with small ones, an emerging demand in the NISQ era. 
Meanwhile, we could integrate other classical perturbation treatments of the interaction with quantum computing, such as the one that expands according to the interaction strength.
% , which may provide a more efficient PQS.
\revise{
We might also consider other hybrid approaches, such as tensor networks, to effectively solve complex many-body systems while alleviating the simulation cost.
One may also apply the idea of PQS to more efficiently emulate large quantum circuits using smaller ones~\cite{peng2020simulating,barratt2020parallel,mitarai2021constructing,fujii2020deep,yuan2020quantum,Mitarai2021overheadsimulating}.
}\\

% Remark that compared with 
% other recent works that simulate clustered Hamiltonians or circuits using gate decompositions~\cite{peng2020simulating,barratt2020parallel,mitarai2021constructing,fujii2020deep,yuan2020quantum,Mitarai2021overheadsimulating}, PQS directly simulates the interacting dynamics with a much smaller overhead.\\

% Furthermore, perturbation theories have also been applied in analogue quantum simulation for synthesising effective Hamiltonians (see \cite{PhysRevLett.105.190403}). Whether our PQS method could be generalised to analogue quantum simulation is an interesting future direction.

% \vspace{30pt}

\begin{acknowledgments}
 J.S. thanks Xiao-Ming Zhang and Chenbing Wang for the valuable discussions.
X.Y. was supported by Simons Foundation and the National Natural Science Foundation of China Grant No.~12175003.
P.H. was supported by AFOSR FA9550-19-1-0369.
S.E. is supported by Moonshot R\&D, JST, Grant No.\,JPMJMS2061; MEXT Q-LEAP Grant No.\,JPMXS0120319794, and PRESTO, JST, Grant No.\, JPMJPR2114.
We acknowledge use of the IBM quantum cloud experience for this work. The views expressed are those of the authors and do not reflect the official policy or position of IBM or the IBMQ team.
The numerics are supported by the University of Oxford Advanced Research Computing (ARC) facility and the high-performance Computing Platform of Peking University.
\end{acknowledgments}

% \bibliographystyle{unsrt}
%\bibliography{bib_simulation}

%\newpage
%\clearpage

\widetext

\section*{Supplementary Materials: Perturbative quantum simulation}

\tableofcontents

\section{Perturbative quantum simulation ---  general framework}
\label{appendix:framework}

\subsection{Generalised quantum operation}
\label{appendix:operation}
In our main text, we introduced the generalised quantum operation as
\begin{equation}
    \Phi(\rho) = \tr_E[U(\rho\otimes \ket{0}\bra{0}_E)V^\dag],
\end{equation}
where $U$ and $V$ could be different unitary operators that applies jointly on $\rho$ and $\ket{0}_E$. We show several properties of the generalised quantum operation $\Phi(\rho)$.
\begin{itemize}
    \item The generalised quantum operation $\Phi(\rho)$ has bounded Schatten norm. Specifically, the Schatten norm of a matrix is $\|M\|_p = \tr[|M|^p]^{1/p}$ for $p\in [1,\infty)$ and we have
    \begin{equation}
        \|\Phi(\rho)\|_p \le \|\Phi(\rho)\|_1 \le \|U(\rho\otimes \ket{0}\bra{0}_E)V^\dag\|_1 =\|\rho\otimes \ket{0}\bra{0}_E\|_1 = \|\rho\|_1.
    \end{equation}
    Here the two inequality follows from the non-increasing of the Schatten norm over $p$ and the non-increasing of the trace norm under partial trace. 
    
    Nevertheless, since $\Phi(\rho)$ could be complex, it might not be a quantum channel in general.
    
    \item The real and imaginary part of $\Phi(\rho)$ could be expressed as a linear combination of completely positive trace non-increasing quantum channels. Specifically, they could be
    obtained with the following circuit. 
\begin{align*}
\Qcircuit @C=0.8em @R=1.2em {
\lstick{\ket{+}_0}&\qw&\ctrl{1}& \measureD{X_0, Y_0}\\
\lstick{\rho}&\multigate{1}{U}&\multigate{1}{VU^\dag }&\qw\\
\lstick{\ket{0}_E}&\ghost{U}&\ghost{VU^\dag }&\\
}
\end{align*}
The output state before the measurement is
\begin{equation}
\begin{aligned}
    \rho_{\rm out}=&\ket{0}\bra{0}_0 U\rho\otimes \ket{0}\bra{0}_E U^\dag + \ket{0}\bra{1}_0 U\rho\otimes \ket{0}\bra{0}_E V^\dag\\
    &+ \ket{1}\bra{0}_0 V\rho\otimes \ket{0}\bra{0}_E U^\dag + \ket{1}\bra{1}_0 V\rho\otimes \ket{0}\bra{0}_E V^\dag.
\end{aligned}
\end{equation}
The real and imaginary part of $\Phi(\rho)$ can be obtained from the $X$ and $Y$ measurements
\begin{equation}
    \begin{aligned}
    \mathrm{Re}[\Phi(\rho)] &= \tr_0[\rho_{\rm out}X_0],\\
    \mathrm{Im}[\Phi(\rho)] &= \tr_0[\rho_{\rm out}Y_0].\\
    \end{aligned}
\end{equation}

\item The measurement of the output state is realised straightforwardly. For example, the real and imaginary parts of $\tr[\Phi(\rho)O]$ could be realised with the following circuit.
\begin{align*}
\Qcircuit @C=0.8em @R=1.2em {
\lstick{\ket{+}_0}&\qw&\ctrl{1}& \measureD{X_0, Y_0}\\
\lstick{\rho}&\multigate{1}{U}&\multigate{1}{VU^\dag }&\measureD{O}\\
\lstick{\ket{0}_E}&\ghost{U}&\ghost{VU^\dag }&\\
}
\end{align*}

\item When $U=V$, it reduces to a quantum channel $\mc N$
\begin{equation}
    \mc N(\rho) = \tr_E[U(\rho\otimes \ket{0}\bra{0}_E)U^\dag].
\end{equation}
with the circuit
\begin{align*}
\Qcircuit @C=0.8em @R=1.2em {
\lstick{\ket{0}_E}&\multigate{1}{U}&\qw&\\
\lstick{\rho}&\ghost{U}&\qw&\\
}
\end{align*}
% \begin{align*}
% \Qcircuit @C=0.8em @R=1.2em {
% \lstick{\rho}&\multigate{1}{U}&\qw&~~{\ket{+}_0}&\qw&\ctrl{1}&\qw\\
% \lstick{\ket{0}_E}&\ghost{U}&\qw&~~{\rho}&\multigate{1}{U}&\multigate{1}{VU^\dag }&\qw\\
% }
% \end{align*}
When there is no ancillary $E$, it becomes 
$
    \Phi(\rho) = U\rho V^\dag,
$
with the circuit
\begin{align*}
\Qcircuit @C=0.8em @R=1.2em {
\lstick{\ket{+}_0}&\qw&\ctrl{1}&\qw \\
\lstick{\rho}&\gate{U}&\gate{VU^\dag }&\qw\\
}
\end{align*}
which plays a key role in our explicit scheme. 

\item Given two generalised quantum operations
\begin{equation}
    \begin{aligned}
    \Phi_1(\rho) &= \tr_{E_1}[U_1(\rho\otimes \ket{0}\bra{0}_{E_1})V_1^\dag],\\
    \Phi_2(\rho) &= \tr_{E_2}[U_2(\rho\otimes \ket{0}\bra{0}_{E_2})V_2^\dag],\\
    \end{aligned}
\end{equation}
the concatenated operation 
\begin{equation}
   \Phi_2 \circ \Phi_1(\rho) = \tr_{E_1E_2}[U_2U_1(\rho\otimes \ket{0}\bra{0}_{E_1,E_2})V_1^\dag V_2^\dag],
\end{equation}
is also a generalised quantum operation. The real and imaginary part of $\Phi_2 \circ \Phi_1(\rho)$ could be obtained from measuring the ancillary qubit on the $X_0$ and $Y_0$ basis with the following circuit. 
\begin{align*}
\Qcircuit @C=0.8em @R=1.2em {
\lstick{\ket{+}_0}&\qw&\ctrl{1}& \measureD{X_0, Y_0}\\
\lstick{\rho}&\multigate{1}{U_2U_1}&\multigate{1}{V_2V_1U_1^\dag U_2^\dag }&\qw\\
\lstick{\ket{0}_{E_1,E_2}}&\ghost{U_2U_1}&\ghost{V_2V_1U_1^\dag U_2^\dag}&\\
}
\end{align*}

It can be equivalently realised as follows using  two ancillary qubits. 
\begin{align*}
\Qcircuit @C=0.8em @R=1.2em {
\lstick{\ket{+}_0}&\qw&\ctrl{2}&\qw&\qw& \measureD{X_0, Y_0}\\
\lstick{\ket{+}_{0'}}&\qw&\qw&\qw&\ctrl{1}& \measureD{X_{0'}, Y_{0'}}\\
\lstick{\rho}&\multigate{1}{U_1}&\multigate{1}{V_1U_1^\dag  }&\multigate{2}{U_2}&\multigate{2}{V_2U_2^\dag }&\qw\\
\lstick{\ket{0}_{E_1}}&\ghost{U_1}&\ghost{V_1 U_1^\dag }\\
\lstick{\ket{0}_{E_2}}&\qw&\qw&\ghost{U_2}&\ghost{V_2 U_2^\dag}&\\
}
\end{align*}
\end{itemize}
In particular, the circuit factorise into two independent ones when (1) $\rho$ is a tensor product of two states $\rho=\rho_1\otimes \rho_2$ (2) $U_1$ and $V_1$ applies on $\rho_1$ and $\ket{0}_{E_1}$; $U_2$ and $V_2$ applies on $\rho_2$ and $\ket{0}_{E_2}$.

\subsection{Algorithm description --- discrete time}
\label{appendix:sub_discrete_time}

% The above section discuss that for a pure state, we can evolve the state by applying local unitaries to equivalently implement the evolution under the non-local interactions. Now, we extend the algorithm to mixed state. 
The goal of the algorithm is to simulate the dynamics of a many-body Hamiltonian, for example 
% $L$ strongly correlated subsystems with weak interactions between subsystems with the Hamiltonian
\begin{equation}
H=H^{\textrm{loc}} + V^{\textrm{int}},
\end{equation}
where $H^{\textrm{loc}}$ corresponds to the strong but local interaction and $V^{\textrm{int}}$ corresponds to weak perturbations. In practice, we can always divide the whole system into $L$ subsystems, and thus consider 
\begin{equation}
H^{\rm loc}=\sum_l H_l
\end{equation}
as the local Hamiltonians with each $H_l$ acting on the $l$th subsystem, and 
\begin{equation}
V^{\textrm{int}} =  \sum_j \lambda_jV_j^{\rm int}
\end{equation}
as the weak perturbation interaction between the subsystems with different interactions $V_j^{\rm int}$ and coefficients $\lambda_i$. The local Hamiltonians and the perturbation  interactions depend on the partitioning strategy of the subsystems, and we refer to \autoref{appendix:numerical} for the partitioning for different physical systems.
We assume a time independent Hamiltonian in the following discussion; however, our results apply to general time dependent Hamiltonians. 
% For simplicity, we assume that the whole system consists of $L$ systems, 
% with  controllable operations on a small quantum computer.
% Suppose we can implement the local operations on each subsystem,
% where $H^{\rm loc}=\sum_l H_l$ corresponds to the local Hamiltonian with each $H_l$ acts on the $l$th subsystem, and $V^{\textrm{int}}$ corresponds to the weak interaction between the subsystems. Such a 
Since we are considering Hamiltonian simulation with a quantum computer, we assume that every $H_l$ 
could be decomposed as a linear combination of tensor product of Pauli operators and each $V_j^{\rm int}$ is a tensor product of Pauli operators. 
% For example, we have 
% \begin{equation}
%     V = \sum_j \lambda_jV_j^{\rm int},
% \end{equation}
% where $V_j^{\rm int}$ is the tensor product of operators acting on the subsystems and we assume that the interaction strength $\lambda$ is weak. 

We describe the algorithm assuming discretised time. We show shortly how to take the limit of infinitesimal timesteps.  We aim to simulate the time evolution of $H$ with time $t$, 
\begin{equation}
    \mc U(T)[\rho] = U(T)\rho U(T)^\dag,
\end{equation}
where $U(t) = e^{-iHt}$. Considering discrete time $\delta t$, we have
\begin{equation}
\begin{aligned}
\mc U(T) =& \prod_{i=1}^{T/\delta t}\mc U(\delta t) = \prod_{i=1}^{T/\delta t}\bigg[\mc V^{\textrm{int}}(\delta t)  \circ\bigotimes_{l=1} \mc U_l(\delta t)\bigg]+\mc O(T\delta t),
\end{aligned}
\end{equation}
where $\mc U_l(t) = U_l(t)\rho U_l(t)^\dag$ with $U_l(t) = e^{-iH_lt}$ and $\mc V^{\textrm{int}}(t) = V^{\textrm{int}}(t)\rho V^{\textrm{int}}(t)^\dag$ with $V^{\textrm{int}}(t) = e^{-iV^{\textrm{int}}t}$. Here $\mc O(T\delta t)$ corresponds to the Trotter error, which vanishes when taking the limit of $\delta t\rightarrow 0$. 
% The evolution of a quantum state $\rho_0$ under Hamiltonian $H$ with a small time $\delta t$ is
% \begin{equation}
% 	\mc U(\rho_0)  \approx \mc U^{\textrm{int}}  \circ\bigotimes_{l=1} \mc U_l(\rho_0),
% \end{equation}
% where $\mc U^{\rm int}(\rho) = e^{-iH^{\rm int}\delta t}\rho e^{-iH^{\rm int}\delta t}$ and  $\mc U_l(\rho) = e^{-iH_l\delta t}\rho e^{-iH_l\delta t}$.  
Note that the evolution consists of local evolution $\mc U_l(\delta t)$ on the $l$th subsystem term and the joint evolution $\mc V^{\textrm{int}}(\delta t)$.

The next step is to decompose the joint non-local operation  $\mc V^{\textrm{int}}(\delta t)$ into local operations that separately act on the subsystems. In particular, we consider a set of generalised quantum operations as 
\begin{equation}
\begin{aligned}
\Phi_k(\rho) = \tr_E \left[ \mathbf{U}_k\left(\rho \otimes  \mathbf{\ket{0}\bra{0}}_{E} \right)  \mathbf{V}_k\right]
\end{aligned}
\end{equation}
where we denote $ \mathbf{U}_k = U_{1E_1, k}\otimes U_{2E_2, k} \otimes \cdots \otimes U_{LE_L, k}$,  $ \mathbf{V}_k= V_{1E_1,k}^{\dagger}\otimes V_{2E_2,k}^{\dagger} \cdots \otimes    V_{LE_L,k}^{\dagger}$, and $ \mathbf{\ket{0}\bra{0}}_E =\ket{0}\bra{0}_{E_1}\otimes \ket{0}\bra{0}_{E_2} \cdots \otimes \ket{0}\bra{0}_{E_L}$,
and each $ U_{lE_l, k}$ and  $V_{lE_l, k}$ is applied jointly on the $l$th subsystem and the ancilla $E_l$. 
{Denoting $\Phi_{l,k}(\rho_l) = \tr_{E_l} \left[ {U}_{lE_{l},k}\left(\rho_l \otimes  \mathbf{\ket{0}\bra{0}}_{E_l} \right)  {V}_{lE_{l},k}^{\dagger}\right]$ to be the generalised quantum operation acting on the $l$th subsystem, we thus have
\begin{equation}
   \Phi_k  =  \Phi_{1,k} \otimes \Phi_{2,k} \otimes\cdots\otimes  \Phi_{L,k},
\end{equation}
which applies separately on each subsystem. }
When  a sufficient number of $\Phi_k$ is chosen, we can always decompose the instant joint evolution  $\mc V^{\textrm{int}}(t)$ as a linear combination of local operations,
\begin{equation}\label{Eq:Vdecomposition132}
    \mc V^{\textrm{int}}(\delta t) = \mc I + \delta t\sum_k \alpha_k \Phi_k=\mc I + \delta t\sum_k \alpha_k \Phi_{1,k} \otimes \Phi_{2,k} \otimes\cdots\otimes  \Phi_{L,k},
\end{equation}
where $\mc I$ corresponds to the identity channel $\mc I(\rho)=\rho$ and $\alpha_k$ are complex coefficients. For example, we can choose $\{\Phi_k\}$ to be a complete basis for all quantum channels. When the set of $\{\Phi_k\}$ is chosen, we can find the coefficients $\alpha_k$ via linear programming.
% We assume $\alpha_k$ are positive since their phase can always be absorbed into 

% We consider the following general decomposition
%     \begin{equation}\label{Eq:}
% 	\mc U^{\rm int}(\rho) = \sum_{k} \lambda_k  \tr_E \left[ \mathbf{U}_k\left(\rho \otimes  \mathbf{\ket{0}\bra{0}}_{E} \right)  \mathbf{V}_k\right] ,
% \end{equation}

% % We assume that $U$ and $V$ applied on jointly on the system and the ancilla $E_1,E_2,...,E_L$. 
% It is convenient to define a generalised quantum operation as $\Phi_{l,k} =\tr_{E_l} U_{l,k}(\cdot)V_{l,k}^{\dagger}  $.
% It is worth mentioning that the above quantum operation reduces to a quantum channel when $\mathbf{U}=\mathbf{V}$.

Now we can express the joint evolution as
\begin{equation}\label{Eq:perturexpansion0}
\begin{aligned}
\mc U(T) = \prod_{i=1}^{T/\delta t}\bigg[\bigg(\mc I + \delta t\sum_k \alpha_k \Phi_k\bigg)  \circ\bigotimes_{l=1} \mc U_l(\delta t)\bigg]+\mc O(T\delta t).
\end{aligned}
\end{equation}
Denote $\Phi_0 = \mc I$, $c(\delta t)=1+\sum_k|\alpha_k|\delta t$, $p_0(\delta t) = 1/c(\delta t)$, $p_k(\delta t)=|\alpha_k|\delta t/c(\delta t)$, $\theta_k =-i \ln(\alpha_k/|\alpha_k|)$, we can re-express the above equation as
\begin{equation}\label{Eq:perturexpansion}
\begin{aligned}
\mc U(T) &= \prod_{i=1}^{T/\delta t}\bigg[c(\delta t)\bigg(\sum_{k}e^{i\theta_k}p_k(\delta t)\Phi_k\bigg)  \circ\bigotimes_{l=1} \mc U_l(\delta t)\bigg]+\mc O(T\delta t),\\
&= c(\delta t)^{T/\delta t}\sum_{k_1,k_2,\dots,k_{T/\delta t}}\prod_{i=1}^{T/\delta t}\bigg[e^{i\theta_{k_i}}p_{k_i}(\delta t)\Phi_{k_i}\circ\bigotimes_{l=1} \mc U_l(\delta t)\bigg]+\mc O(T\delta t),\\
&= c(\delta t)^{T/\delta t}\sum_{\mathbf k}e^{i\theta_{\mathbf k}}p_{\mathbf k}\prod_{i=1}^{T/\delta t}\bigg[\bigotimes_{l=1}\Phi_{l,k_i}\circ\bigotimes_{l=1} \mc U_l(\delta t)\bigg]+\mc O(T\delta t),\\
&= c(\delta t)^{T/\delta t}\sum_{\mathbf k}e^{i\theta_{\mathbf k}}p_{\mathbf k}\bigotimes_{l=1}\bigg[\prod_{i=1}^{T/\delta t}\bigg(\Phi_{l,k_i}\circ \mc U_l(\delta t)\bigg)\bigg]+\mc O(T\delta t).\\
\end{aligned}
\end{equation}
Here ${\mathbf k}=(k_1,\dots,k_{T/\delta t})$, $p_{\mathbf k} = p_{k_1}p_{k_1}\dots p_{k_{T/\delta t}}$, {$\theta_{\mathbf k} = \theta_{k_1}+\theta_{k_2}\dots +\theta_{k_{T/\delta t}}$}. In the main text, we denote the phase as $\mathcal{P}_{\mathbf k} = e^{i\theta_{\mathbf k}}$.
The whole evolution $\mc U(T)$ is now decomposed as a linear combination of operations that act locally on each subsystem. We can thus effectively realise the joint evolution  using only local operations.

% Now suppose we decompose each interaction unitary operator $\mc  U^{\textrm{int}}  $ into a linear combination of local operations on subsystems, and we have
% \begin{equation}\label{Eq:decompositionl}
% \mc  U^{\textrm{int}} =\sum_{k} \lambda_{k} \Phi_{1,k} \otimes \Phi_{2,k} \otimes\cdots\otimes  \Phi_{L,k}, 
% \end{equation}
% where each $\Phi_{l,k}$ acts on the $l$th subsystem.

We next discuss the measurement of non-local observable.
Suppose that the initial state $\rho(0)$ is decomposed as 
\begin{equation}
	\rho(0) = \sum_{k_0} \alpha_{k_{0}} \bigotimes_{l=1}\rho_{l,k_0}
\end{equation}
% Then the time evolution of the state is decomposed as 
% \begin{equation}
% \begin{aligned}
% \prod_{i=1}^{T/\delta t}\mc U(\rho_0) 
% &= \prod_{i=1}^{T/\delta t} \mc U^{\textrm{int}} \circ \bigotimes_{l=1} \mc U_l\big(\sum_{k_0} p_{k_{0}} \bigotimes_{l=1}\rho_{l,k_0}\big)\\
% &= \prod_{i=1}^{T/\delta t} \sum_{k_0,\mathbf k} p_{k_0}\lambda_{k_i}  \Phi_{1,k_i} \otimes \Phi_{2,k_i} \otimes\cdots\otimes  \Phi_{L,k_i}  \circ \bigotimes_{l=1} \mc U_l(\rho_{l,k_0})\\
% &= \sum_{k_0,\mathbf k} p_{k_0} \Lambda_{\mathbf k}\bigotimes_{l=1} \left(     \prod_{i=1}^{T/\delta t}\Phi_{l,k_i} \circ\mc U_l(\rho_{l,k_0})\right)\\
% \end{aligned}
% \end{equation}
% with $\Lambda_{\mathbf k} = \prod_{k_1,\dots k_T/\delta t} \lambda_{k_i}$ and ${\mathbf k}=(k_1,\dots,k_{T/\delta t})$.
% We thus show that the joint evolution channel $\mathcal{U}(\rho,T)$ can be realised with operations separately acting on each subsystem. 
% and measure observables on the time-evolved state using this channel method.
and measure an observable like 
\begin{equation}
	O=\sum_{k_O} \alpha_{k_O} \prod_{l=1}O_{l,{k_O}},
\end{equation}
then we have
\begin{equation}\label{eq:generaldecomp}
	\tr\bigg[\mc U(T)[\rho(0)] O\bigg] = c(\delta t)^{T/\delta t} \sum_{{k_O},k_0,\mathbf  k} \alpha_{k_0}\alpha_{k_O} e^{\theta_{\mathbf k}}p_{\mathbf k}\prod_{l=1}  \tr\bigg[     \prod_{i=1}^{T/\delta t}\bigg(\Phi_{l,k_i} \circ\mc U_l(\delta t)\bigg)[\rho_{l,k_0}]O_{l,k_O})\bigg]+\mc O(T\delta t).
	\end{equation}
Here each term $\tr\bigg[     \prod_{i=1}^{T/\delta t}\bigg(\Phi_{l,k_i} \circ\mc U_l(\delta t)\bigg)[\rho_{l,k_0}]O_{l,k_O})\bigg]$ can be obtained from operations only on the $l$th subsystem. The  expectation value of the arbitrary joint state is now a linear combination of products of local measurement results.

\subsection{Monte Carlo implementation and continuous time}
\label{appendix:sub_MonteCarlo}

\subsubsection{Discrete time Monte Carlo method}
The number of expanded terms is proportional to $N_V^{T/\delta t}$, with $N_{V}$ being the number of terms in the expansion of  \autoref{Eq:Vdecomposition132}. Although $N_V^{T/\delta t}$ increases exponentially, we do not need to measure all the expanded terms and the Monte Carlo method could more efficiently obtain the measurement outcome.

In particular, the decomposition of  \autoref{eq:generaldecomp} can be written in a general form of 
\begin{equation}
	\braket{O} =  \sum_{k} q_k \prod_{l=1}  \tr\left[  \Phi_l(\rho_{l,k})O_{l,k}\right] = C\sum_{k} e^{i\theta_k}p_k \prod_{l=1}  \tr\left[  \Phi_l(\rho_{l,k})O_{l,k}\right]+\mc O(T\delta t),	
\end{equation}
with $k=({k_O},k_0,\mathbf  k)$, $q_k = c(\delta t)^{T/\delta t}\alpha_{k_0}\alpha_{k_O} e^{i\theta_{\mathbf k} }p_{\mathbf k}$, $C = \sum_k q_k=c(\delta t)^{T/\delta t}\sum_{k_0}|\alpha_0|\sum_{k_O}|\alpha_O|$, $\theta_k=-i\ln(q_k/|q_k|)$, $p_k=|q_k|/C$, and $\Phi_l = \prod_{i=1}^{T/\delta t}\bigg(\Phi_{l,k_i} \circ\mc U_l(\delta t)\bigg)$. To obtain the measurement $\braket{O}$, we  can use the following Monte Carlo random sampling method,

\begin{enumerate}
    \item Generate random numbers $k$ according to the probability $\{p_k\}$;
    \item For the $l$th subsystem, prepare state $\rho_{l,k}$, apply the operation $\Phi_l$, and measure the observable $O_{l,k}$ to get $\braket{O_{l,k}}$.
    \item Multiply all the outcomes $\braket{O_{l,k}}=\tr[\Phi_l(\rho_{l,k})O_{l,k}]$ of different subsystems, as well as the phase $e^{i\theta_k}$ and $C$.  
    \item Repeat steps 1-3 $N_s$ time and output $O_{\rm est} = \sum_k Ce^{i\theta_k}\prod_l \braket{O_{l,k}}$.
\end{enumerate}
Ignoring the effect of Trotter error with a finite timestep, the expansion guarantees that the output is an unbiased estimation of the exact measurement outcome. 
Suppose each $O_{l,k}$ is a Pauli measurement, then  with failure probability $\delta$, the estimation error scales as
\begin{equation}
    \varepsilon = \mc O\bigg({C}\sqrt{\frac{\log_2{1/\delta}}{N_s}}\bigg).
\end{equation}
Since  the coefficient $C$ boosts the error, it quantifies the cost of the random sampling process. Suppose the input state is a product state, then the additional cost that the perturbative expansion introduces is 
$C = c(\delta t)^{T/\delta t}$. We will shortly give a detailed analysis of this cost in \autoref{appendix:sub_p1_cost}. 

% While this finite timestep approach requires to discretise time  so that we can insert the $\Phi$ channel after the $\mc U$ evolution. In the following, we show how to equivalently run the algorithm without discretised time.
% Suppose 

% Focusing on   \autoref{Eq:perturexpansion0}, the whole evolution is expanded as the sum of different evolution trajectories, where all of them continuously evolves under $\mc U_l(\delta t)$ and possibly jumps with an additional operation $\Phi_k$ at each timestep. 

% with small $\delta t$, 

% In practice, it could be challenging to `continuously' interchange the original local evolution within a sufficiently small time step $\delta t$. In this section, we discuss the Monte Carlo method to stochastically implement the joint evolution and prove its equivalence to the continuous version.
% Since $[\Phi_{{l,k_i}},\mc U_l]$ may not be zero, 

A major caveat of the above scheme is that it assumes a small discrete timestep and requires to continuously interchange the subsystem evolution $\mc U_l$ and $\Phi_{{l}}$ with a sufficiently small time step $\delta t$. 
In practice, it could be challenging to `continuously' interchange the  subsystem evolution within a sufficiently small time step $\delta t$.
We show in the next subsection that we can apply an equivalent Monte Carlo method to stochastically implement the joint evolution. As such, a general Hamiltonian simulation method other than Trotterisation could be applied to reduce the algorithmic error.

\subsubsection{Stochastic implementation}

% The key is to notice that the majority part of the decomposition of  \autoref{Eq:decompositionl} is the identity evolution. That is, for the evolution of the interaction with small time $\delta t$, we have
% \begin{equation}
% \label{eq:stochastic_decomp}
% \mc  U^{\textrm{int}} =\lambda_{0}\mc I + \sum_{k\ge 1} \lambda_{k} \Phi_{1,k} \otimes \Phi_{2,k} \otimes\cdots\otimes  \Phi_{L,k} = c\bigg(p_{0}\mc I + \sum_{k\ge 1} p_{k} \Phi_{1,k} \otimes \Phi_{2,k} \otimes\cdots\otimes  \Phi_{L,k} \bigg), 
% \end{equation}
% with $\mc I$ being the identity channel, $\lambda_{0}\approx 1$, $\lambda_{k}\propto\delta t$, $c=\sum_k\lambda_k= 1+O(\delta t)$, $p_k=\lambda_k/c\propto\delta t$. 
% Here we assume $\lambda_k\ge 0$ without loss of generality. 

We first rewrite  \autoref{Eq:perturexpansion0} as follows
\begin{equation}
\begin{aligned}
\mc U(T)
&= c(\delta t)^{T/\delta t}\prod_{i=1}^{T/\delta t} \bigg(p_{0}\mc I + p_{\ge 1}\tilde\Phi \bigg)\circ \bigotimes_{l=1} \mc U_l(\rho(0)),
\end{aligned}
\end{equation}
where $p_{\ge 1} = \sum_{k\ge 1}p_k$, $\tilde\Phi = {\sum_{k\ge 1}\alpha_k\Phi_k}/{\sum_k|\alpha_k|}$. We note that, at each timestep, we always evolve each subsystem according to  $\mc U_l$, and with a small probability $p_{\ge 1}$, we evolve under $\tilde\Phi$. Since the probability $p_{\ge 1}\propto\delta t$ is negligible when taking the limit of $\delta t\rightarrow 0$, we can equivalently realise it with a continuous decaying or jump process. 
% where at each time, we evolve the subsystem with $\mc U_l$ and we apply the quantum operation $\Phi_{l,k}$ with negligible probability $p_k$. Taking the limit of $\delta t\rightarrow 0$, the procedure of applying the channel $\Phi_{l,k}$ is equivalent to the continuous decaying process. 
Specifically, we can realise the evolution $\mc U(T)$  with the following stochastic process

\begin{enumerate}
	\item Generate a uniformly distributed random number $p_{\rm jp}\in[0,1]$. 
	\item Determine $t_{\rm jp}$ by solving $p_{\rm jp}=Q(t)$ with $Q(t) =e^{-\Gamma(t)}  $, $\Gamma(t)=t \sum_{k\geq 1} \tilde p_k$, and $\tilde p_k=\lim_{\delta t\rightarrow 0}p_k/\delta t = {|\alpha_k|}$.
	\item Evolve each subsystem state with $\mc U_l$ to time $t$ and update $t=t+t_{\rm jp}$. 
	\item Generate another random number $q_{m}\in[0,1]$ to determine $\Phi_{k}$ and apply $\Phi_{l,k}$ to the $l$th subsystem.
	\item Repeat Step $1-4$ until $t=T$. 
\end{enumerate}

Therefore, we can stochastically realise the decomposition without assuming a discrete time. Meanwhile, other advanced Hamiltonian simulation algorithms such as Qubitisation could be used for each time evolution at step 3. We also note that the jump time $t_{\rm jp}$ and hence the evolution could be predetermined, which makes its implementation almost as easy as conventional Hamiltonian simulation methods. 

Now suppose we have a product input state $\rho(0) = \bigotimes_l \rho_l(0)$ and product measurement $O = \bigotimes_l O_{l}$, the stochastic Monte Carlo implementation of the general perturbative method is summarised as follows. When the input state or the measurement is not in a product form, we can similarly decompose them as we discuss above. 

% to estimate the expectation value of observable $O = \bigotimes_l O_{l}$

\begin{algorithm}[H]
\begin{algorithmic}[1]
\State{Get $C$, $\{\alpha_j\}$, $\left\{\tilde p_j = {|\alpha_k|}\right\}$,  and $\theta_i = -i\ln(\alpha_k/|\alpha_k|)$ from interaction channel $\mc V$, set $\bigg\{s_j=\frac{\sum_{i=1}^j \tilde p_i}{ \sum_{i} \tilde p_i}$\bigg\} and $\Gamma(t)=t \sum_{k} \tilde p_k$.}
\For {$m=1$ to $N_s$}
\State {Randomly generate $q_0\in[0,1]$, set $t=0$, $n=0$, $\theta=0$.}
\While{$t\le T$}
\State{Get $t_{\textrm{jp}}^{n}$ by solving $\mathrm{exp}\left(- \Gamma(t^n_{\mathrm{jp}}) \right)=q_n$.}
\State{Randomly generate $q'_n\in[0,1]$.}
\State{Set $j_n=j$ if $q'_n \in [s_{j-1},s_j]$ and update $\theta=\theta_{j_n}+\theta$.}
\State{Update $t=t+t_{\textrm{jp}}^{n}$ and $n=n+1$.}
\EndWhile
\For {$l=1$ to $L$}
\State{Set $\rho_l=\rho_l(0)$ and $\bar O=0$.}
\For{$k = 0:n-1$}
\State{Evolve $\rho_l$ under $\mc U_l$ for time $t_{\textrm{jp}}^{k}$ and apply $\Phi_{l, j_k}$.}
\EndFor
\State{Evolve $\rho_l$ under $\mc U_l$ for time $T-\sum_{k=0}^{n-1}t_{\textrm{jp}}^{k}$.
\State{Measure $O$ of $\rho_l$ to get $O_{l,m}$.}
}
\EndFor
\State{Update $\bar{O}=\bar O+Ce^{i\theta} \prod_l O_{l,m}/N_s$}
\EndFor
\end{algorithmic}
\caption{Perturbative quantum simulation. \\Input: initial state $\rho(0) = \bigotimes_l \rho_l(0)$, number of samples $N_s$, local evolution  $\mc U_l$, decomposition of the interaction $\mc V^{\textrm{int}}(\delta t) =\mc I + \delta t\sum_k \alpha_k \Phi_{1,k} \otimes \Phi_{2,k} \otimes\cdots\otimes  \Phi_{L,k}$ with quantum operations $ \Phi_{l,j}$, measurement $O = \bigotimes_l O_{l}$. Output: $\bar O$.  } 
\label{alg1}
\end{algorithm}

\subsubsection{Equivalence between the two Monte Carlo methods}
\label{section: Equivalence}
We now prove the equivalence between the stochastic approach and the discrete time Monte Carlo approach with $\delta t\rightarrow 0$. Following the above discussion, we can regard the discrete time Monte Carlo approach as a decaying process. Specifically, at each timestep, it has probability $p_{\ge 1}$ to apply an additional operation $\tilde \Phi$. 
% Since the continuous error mitigation scheme requires to apply instant recovery operation at each time $t$, it also makes its realisation challenge in practice. 
% Provided the quantum operations in Eq. (\ref{eq:stochastic_decomp}), we can interpret it as with probability $1-\sum_i\tilde p_{i}\delta t$ we do nothing, and with probability $\tilde p_{i}\delta t$ we apply a corresponding correction operation. We also multiply $c\cdot \alpha_i$ to the output measurement. 
% We can regard the event that applies the correction operations as a jump similar to the stochastic Schr\"odinger equation approach. 
Starting at time $t=0$ with the limit of $\delta t\rightarrow 0$, the probability that there is no `decay' event until time $t$ is 
\begin{equation}
	Q(t) = \lim_ {\delta t\rightarrow 0}\prod_{i=0}^{t/\delta t} \bigg(1-\sum_{k \geq 1} \tilde p_{k}\delta t\bigg) =e^{-t\alpha},  
\end{equation}
where $\alpha = \sum_k|\alpha_k|$.
% where $\Gamma(t)=\sum_{i\geq 1} \tilde p_i(t)$.
The probability to have a decay event in the time interval $[t, t+dt]$ is 
\begin{align}\label{Eq:appdixPdt}
P(t)dt = \alpha e^{-t\alpha} dt.
\end{align}

For the stochastic method, we generate a uniformly distributed random variable $q \in [0,1]$ and solve 
\begin{align}
q = e^{-t_{\mathrm{jp}}\alpha}, 
\label{eqn:jumptime}
\end{align}
to determine the jump time $t_{\mathrm{jp}}$. Then the probability that jump happens at time $t_{\mathrm{jp}}$ or in particular between $[t_{\mathrm{jp}},t_{\mathrm{jp}}+dt]$ is  
\begin{equation}
	|dq| = \alpha e^{-t_{\mathrm{jp}}\alpha}dt = P(t_{\mathrm{jp}}) dt,
\end{equation}
which agrees with  \autoref{Eq:appdixPdt}. 
We can thus use the uniformly distributed random variable $q$ to determine the jump time to equivalently simulate the discrete time Monte Carlo approach. 

%where $\Gamma(t)=\sum_{j\geq 1} p_j(t)$. We generate a uniformly distributed random variable $q'' \in [0,1]$. We assign the interval $[q'',q''+dq'']$ to the jump event which happens in the time interval $[t_{\mathrm{jp}},t_{\mathrm{jp}}+dt]$. Therefore we have
%\begin{equation}
%\begin{aligned}
%q''&= \int_0^{t_{\mathrm{jp}}} dt P(t) \\
%&=1-e^{-\int_0^{t_{\mathrm{jp}}} \Gamma(t')dt'}.
%\end{aligned}
%\end{equation}
%Then, by using a uniformly distributed random variable $q \in [0,1]$ ($q=1-q''$) and solving 
%\begin{align}
%q = e^{-\int_0^{t_{\mathrm{jp}}} \Gamma(t')dt'}, 
%\end{align}
%we can determine the jump time $t_{\mathrm{jp}}$. 

At the jump time $t_{\mathrm{jp}}$, we apply the quantum operations other than the identity operation. We can determine the quantum operation by generating another uniformly distributed random number $q' \in [0,1]$. If $q' \in [s_{k-1},s_k]$, we set the quantum operation
to $\mathcal{B}_k$, where $s_k(t)=(\sum_{j=1}^k \tilde p_j) / (\sum_{j}^{N_{\rm op}} \tilde p_j)$ and $N_{\rm op}$ is the number of the quantum operations during the evolution.

% \sun{We can pre-determine the jump time ${t_{\rm jp}^1}, t_{\rm jp}^2, ..., t_{\rm jp}^k$ from Eq. (\ref{eqn:jumptime}).
% For time-independent noise, the jump time can be simply determined as $t_{\rm jp} = -\log(q)/\sum_{i \geq 1} {\tilde p_i}$ with $q$ randomly generated from $[0,1]$. Given evolution time $T$, the average number of recovery operations is proportional to $\mathcal{O}(\lambda T)$. In the numerics, the average number of recovery operations is about $0.3$ times per evolution on average given a realistic noise model and simulation task.
% }

\subsection{Cost analysis}
\label{appendix:sub_p1_cost}

The above perturbative quantum simulation (PQS) method introduces a sampling overhead quantified by
\begin{equation}
    C = \lim_{\delta t\rightarrow0}\prod_{i=1}^{T/\delta t} c(\delta t)=\lim_{\delta t\rightarrow0}\prod_{i=1}^{T/\delta t} (1+\alpha\delta t)=e^{T\alpha},
\end{equation}
where $\alpha = \sum_k |\alpha_k|$. Since the simulation accuracy is now $C$ times larger, we need to have $C=\mc O (1)$ and hence $\alpha T = \mc O(1)$ in order to get an accurate result. This could be satisfied when $T$ and $\alpha$ are not too large, i.e., when the product of the simulation time and the interaction strength is constant. While $\alpha$ roughly measures the interaction strength, its analytical relationship to the interaction Hamiltonian $V^{\rm int}$ is not obvious. This is because the value of $\alpha$ depends on the choice of the generalised quantum operations and the decomposition. We can thus define the minimal value of $\alpha$ by optimising over all possible decompositions,
\begin{equation}
    \alpha_{\min} = \min_{\{\Phi_k\}} \alpha(\{\alpha_k,\Phi_k\}),
\end{equation}
where we write $\alpha(\{\alpha_k,\Phi_k\})$ as a function of the generalised quantum operations and the minimisation is over all possible decomposition strategies. 
% Nevertheless, 
Here we give an analytical lower bound to $\alpha_{\min}$ as a function of the interaction $V^{\rm int}$. We show in the next section an explicit decomposition strategy that achieves this lower bound. 

We consider the Choi state of the instant evolution $\mc V^{\rm int}(\delta t)$ by inputting tensor products of the maximally entangled states. Specifically, inputting $\ket{\phi}_{l,l'} = \sum_j \ket{jj}_{l,l'}/\sqrt{d}$ to the $l$th subsystem with $d$ being the dimension, the output state $\phi^{\rm int}_{1,1',\dots, L,L'}$ is the Choi state,
\begin{equation}
    \phi^{\rm int}_{1,1',\dots, L,L'} = \mc V^{\rm int}(\delta t)\bigg[\bigotimes_{l}{\phi}_{l,l'}\bigg].
\end{equation}
Suppose a decomposition of $\mc V^{\rm int}(\delta t)$ is
\begin{equation}\label{Eq:decomVint}
    \mc V^{\rm int}(\delta t) = \sum_k \tilde\alpha_k \Phi_{1,k} \otimes \Phi_{2,k} \otimes\cdots\otimes  \Phi_{L,k},
\end{equation}
where we have put $\mc I$ into the summation and denote $\tilde\alpha_k$ as the new coefficient incorporating $\delta t$. The relation between $\alpha$ and $\tilde \alpha = \sum_k |\tilde \alpha_k|$ is
\begin{equation}
    \alpha = \lim_{\delta t\rightarrow 0}\frac{\tilde \alpha - 1}{\delta t}.
\end{equation}
Since $\alpha$ depends linearly on $\tilde \alpha$, we can equivalently minimise  $\tilde \alpha$. 

Define isomorphisms $S$ and $T$ of a general matrix $M = \sum_{i,j}M_{i,j}\ket{i}\bra{j}$ as
\begin{equation}
\begin{aligned}
            S(M) &= \sum_{i,j}M_{i,j}\ket{i}\ket{j},\\
            T(M) &= \sum_{i,j}M_{i,j}\bra{i}\bra{j}.
\end{aligned}
\end{equation}
Several useful properties of the $S$ and $T$ are
\begin{itemize}
    \item The definitions of $S(M)$ and $T(M)$ are basis dependent. 
%     In particular, consider a different basis $\{\ket{\psi_i} = U\ket{i}\}$, we have $M = U^\dag U M U^\dag U = U^\dag \sum_{i,j}M_{i,j}\ket{x_i}\bra{x_j} U = \sum_{i,j,i',j'}M_{i,j}U^*_{i',i}\ket{x_{i'}}\bra{x_{j'}}U_{j,j'} = \sum_{i',j'}\tilde M_{i',j'}\ket{x_{i'}}\bra{x_{j'}}$ with $\tilde M_{i',j'} = \sum_{i,j}U^*_{i',i}M_{i,j}U_{j,j'}$ and 
% \begin{equation}
% \begin{aligned}
%             S(M) &= \sum_{i,j}M_{i,j}\ket{i}\ket{j} = \sum_{i',j'}\tilde M_{i',j'}\ket{x_{i'}}\ket{x_{j'}},\\
%             T(M) &= \sum_{i,j}M_{i,j}\bra{i}\bra{j} = \sum_{i',j'}\tilde M_{i',j'}\bra{x_{i'}}\bra{x_{j'}}.
% \end{aligned}
% \end{equation}
\item When applying matrices $U$ and $V$ to $M$, we have
\begin{equation}
\begin{aligned}
            S(UMV) &= U\otimes V^{T}\sum_{i,j}M_{i,j}\ket{i}\ket{j} = U\otimes V^{T}S(M),\\
            T(UMV) &= \sum_{i,j}M_{i,j}\bra{i}\bra{j}U^T\otimes V = T(M)U^T\otimes V.
\end{aligned}
\end{equation}
    \item $S(M)$ and $T(M)$ are related as follows
    \begin{equation}
    S(M) = \left[T(M^*)\right]^\dag.
    \end{equation}
    This is true because $\left[T(M^*)\right]^\dag = \left[\sum_{i,j}M^*_{i,j}\bra{i}\bra{j}\right]^\dag = \sum_{i,j}M_{i,j}\ket{i}\ket{j} =  S(M)$. 
    \item The norms of $S$ and $T$ are the same
    \begin{equation}
        S(M)^\dag\cdot S(M) = T(M) \cdot T(M)^\dag = \tr[M^\dag M] = \|M\|_2^2,
    \end{equation}
    which corresponds to Schatten-2 norm of $M$. 
    This is because $S(M)^\dag \cdot S(M) = \sum_{i',j'}M^*_{i',j'}\bra{i'}\bra{j'}\sum_{i,j}M_{i,j}\ket{i}\ket{j} = \sum_{i,j}M^*_{i,j}M_{i,j} = \tr[M^\dag M]$. The proof is similar for $T(M)$. 
    
    \item Suppose we denote $\ket{M}  = S(M)$ then $T(M) = [S(M^*)]^\dag = \bra{M^*}$.
\end{itemize}

By applying $S$ to the $l,l'$ systems and $T$ to the rest systems, we get a matrix
\begin{equation}
    {\psi}^{\rm int}_{l,l'} = S_{l,l'}\circ \bigotimes_{j\neq l,j'\neq l'}T_{j,j'} (\phi^{\rm int}_{1,1',\dots, L,L'})
\end{equation}
We can thus lower bound $\tilde \alpha$ as follows.

% We apply the isomorphism to all the subsystem matrices and the output state is
% \begin{equation}
%     \tilde{\phi}^{\rm int}_{1,1',\dots, L,L'} = S_{1,1'}\circ S_{2,2'} \circ \cdots S_{L,L'}(\phi^{\rm int}_{1,1',\dots, L,L'}).
% \end{equation}
% Now $\tilde{\phi}^{\rm int}_{1,1',\dots, L,L'}$ becomes a pure state of $L$ systems, with each one consisting of two subsystems, such as $(l,l')$. Denote the reduced density matrix of the $l$th system $(l,l')$ as
% \begin{equation}
%     \tilde{\phi}^{\rm int}_{l,l'} = \tr_{\{1,1',\dots, L,L'\}-\{l,l\}}\left[\tilde{\phi}^{\rm int}_{1,1',\dots, L,L'}\right],
% \end{equation}
% we can thus lower bound $\tilde \alpha$ as follows.

\begin{theorem}
\label{thm:cost_bound}
Given a decomposition of  \autoref{Eq:decomVint} with generalised quantum operations $\{\Phi_{l,k}\}$, we have
\begin{equation}
    \tilde \alpha \ge  \max_{l}\big\|{\psi}^{\rm int}_{l,l'}\big\|_1,
\end{equation}
where $\|A\|_1 = \tr[\sqrt{AA^\dag}]$ is the trace norm. 
\end{theorem}

\begin{proof}
Given the above decomposition, the Choi state of $\mc V^{\rm int}(\delta t)$ is
\begin{equation}
    \phi^{\rm int}_{1,1',\dots, L,L'} = \sum_k \tilde\alpha_k \Phi_{1,k}({\phi}_{1,1'}) \otimes \Phi_{2,k}({\phi}_{2,2'}) \otimes\cdots\otimes  \Phi_{L,k}({\phi}_{L,L'}).
\end{equation}
Considering ${\psi}^{\rm int}_{1,1'}$ as an example, we have 
\begin{equation}
    {\psi}^{\rm int}_{1,1'} = \sum_k \tilde\alpha_k \ket{\Phi_{1,k}({\phi}_{1,1'})} \otimes \bra{\Phi_{2,k}({\phi}_{2,2'})^*} \otimes\cdots\otimes  \bra{\Phi_{L,k}({\phi}_{L,L'})^*},
\end{equation}
where $\ket{\Phi_{L,k}({\phi}_{l,l'})} = S_{l,l'}(\Phi_{L,k}({\phi}_{l,l'}))$ and $\bra{\Phi_{L,k}({\phi}_{l,l'})^*} = T_{l,l'}(\Phi_{L,k}({\phi}_{l,l'}))$. 
% The definition is consistent because $\bra{\tilde\Phi^*_{L,k}({\phi}_{l,l'})} = \left(\ket{\tilde\Phi^*_{L,k}({\phi}_{l,l'})}\right)^\dag = \left(S_{l,l'}(\Phi^*_{L,k}({\phi}_{l,l'}))\right)^\dag = T_{l,l'}(\Phi_{L,k}({\phi}_{l,l'}))$.
Based on the triangle inequality of the trace norm, we have
\begin{equation}
    \begin{aligned}
    \|{\psi}^{\rm int}_{1,1'}\|_1 &\le \sum_k |\tilde\alpha_k| \bigg\|\ket{\Phi_{1,k}({\phi}_{1,1'})} \otimes \bra{\Phi_{2,k}({\phi}_{2,2'})^*} \otimes\cdots\otimes  \bra{\Phi_{L,k}({\phi}_{L,L'})^*}\bigg\|_1,
    \end{aligned}
\end{equation}
The trace norm of each term is
\begin{equation*}
    \begin{aligned}
  &\bigg\|\ket{\Phi_{1,k}({\phi}_{1,1'})} \otimes \bra{\Phi_{2,k}({\phi}_{2,2'})^*} \otimes\cdots\otimes  \bra{\Phi_{L,k}({\phi}_{L,L'})^*}\bigg\|_1 \\
  =& \tr\left[\sqrt{\braket{\Phi_{1,k}({\phi}_{1,1'})|\Phi_{1,k}({\phi}_{1,1'})}\braket{\Phi_{2,k}({\phi}_{2,2'})^*|\Phi_{2,k}({\phi}_{2,2'})^*}\cdots \braket{\Phi_{L,k}({\phi}_{L,L'})^*|\Phi_{L,k}({\phi}_{L,L'})^*} }\right].          
    \end{aligned}
\end{equation*}
Note that
\begin{equation}
\begin{aligned}
    \braket{\Phi_{l,k}({\phi}_{l,l'})|\Phi_{l,k}({\phi}_{l,l'})} = \braket{\Phi_{l,k}({\phi}_{l,l'})^*|\Phi_{l,k}({\phi}_{l,l'})^*} =\tr[\Phi_{l,k}({\phi}_{l,l'})^\dag \Phi_{l,k}({\phi}_{l,l'})].
\end{aligned}    
\end{equation}
Here we used the norms of $S$ and $T$. Based on the property of generalised quantum operations $\Phi_{l,k}$, we have
\begin{equation}
     \tr[\Phi_{l,k}({\phi}_{l,l'})^\dag \Phi_{l,k}({\phi}_{l,l'})]   = \|\Phi_{l,k}({\phi}_{l,l'})]\|_2^2 \le = 1.
\end{equation}
Combining the above results, we thus have
\begin{equation}
    \|{\psi}^{\rm int}_{1,1'}\|_1 \le \sum_k |\tilde \alpha_k| = \tilde \alpha.
\end{equation}
Since the inequality holds for any ${\psi}^{\rm int}_{l,l'}$, we have
\begin{equation}
    \tilde \alpha\ge \max_l \|{\psi}^{\rm int}_{l,l'}\|_1,
\end{equation}
which completes the proof.

% and the trace non-increasing property of the generalised quantum operations. In particular, we have 

% $\bigg\|\ket{\tilde\Phi_{1,k}({\phi}_{1,1'})} \otimes \bra{\tilde\Phi_{2,k}({\phi}_{2,2'})} \otimes\cdots\otimes  \bra{\tilde\Phi_{L,k}({\phi}_{L,L'})}\bigg\|_1\le 1$. 
\end{proof}

Now consider the specific form of $\mc V^{\rm int}(\delta t)[\rho] = \rho + \delta t(- iV^{\rm int}\rho+i\rho V^{\rm int})$ and define the interaction part as
\begin{equation}
    \bar{\mc V}^{\rm int}[\rho] = - iV^{\rm int}\rho+i\rho V^{\rm int}. 
\end{equation}
We can then similarly define the Choi state of $\tilde{\mc V}^{\rm int}$ as
\begin{equation}
    \bar \phi^{\rm int}_{1,1',\dots,L,L'} = \bar{\mc V}^{\rm int}\left[\bigotimes_l \phi_{l,l'}\right],
\end{equation}
and the matrices
\begin{equation}
    \bar{\psi}^{\rm int}_{l,l'} = S_{l,l'}\circ \bigotimes_{j\neq l,j'\neq l'}T_{j,j} (\bar{\phi}^{\rm int}_{1,1',\dots, L,L'}).
\end{equation}
Then consider the decomposition  \autoref{Eq:Vdecomposition132}, we have
\begin{corollary}
Given a decomposition of  \autoref{Eq:Vdecomposition132} with generalised quantum operations $\{\Phi_{l,k}\}$, we have
\begin{equation}
     \alpha \ge  \max_{l}\big\|\bar{\psi}^{\rm int}_{l,l'}\big\|_1,
\end{equation}
where $\|A\|_1 = \tr[\sqrt{AA^\dag}]$ is the trace norm. 
\end{corollary}
In the next section, we will consider a specific decomposition strategy and show how to use the analytical lower bound to prove its optimality.
% $\Phi^{\rm int}$ 
% we have 
% \begin{equation}
%     \mc V^{\rm}(\delta t)[\rho] = e^{-i V^{\rm int}}\rho e^{i V^{\rm int}}.
% \end{equation}

% Suppose $c=1+\tilde c\delta t$, the simulation cost is $C = \lim_{\delta t\rightarrow0}\prod_{i=1}^{T/\delta t} c=e^{T\tilde c }$. Therefore, the algorithm is efficient for $T\tilde c=\mc O(1)$. \\

% We can decompose the non-local unitary operator into a complete set of basis operations.
% Table I lists $16$ single-qubit basis operations, which consists of Hadamard gate H, the phase gate S, and a projection operation.
% The resource cost introduced by the quasi-probability decomposition scales as $\exp(\mathcal{O}(\sum_i |\lambda_i|))$.
% We can minimize the resource cost using linear program, which can be written as

\subsection{A complete basis operation set}
\label{appendix: basisoperation}
We can choose the generalised quantum operations to be a complete set of basis operations. 
In particular,  every single qubit operation can be decomposed into a linear combination of $16$ basis operations. This is because every single qubit operation (including projective measurements) can be expressed with square matrices with $4 \times 4 = 16$ elements by using the Pauli transfer representation~\cite{greenbaum2015introduction}. Therefore, $16$ linearly independent operations are sufficient to emulate arbitrary single qubit operations. Table~\ref{tab:bases} displays one efficient set of single-qubit basis operations in Ref.~\cite{endo2018practical}.    

\begin{table}[h!t!]
\begin{center}
\begin{tabular}{|c|l||c|l||c|l||c|l|}
\hline
1 & ~~$[I]$ (no operation) &
2 & ~~$[\sigma^{\rm x}] $ &
3 & ~~$[\sigma^{\rm y}] $ &
4 & ~~$[\sigma^{\rm z}] $ \\
\hline
5 & ~~$[R_{\rm x}] = [\frac{1}{\sqrt{2}}(I +i \sigma^{\rm x})]$ &
6 & ~~$[R_{\rm y}] = [\frac{1}{\sqrt{2}}(I +i \sigma^{\rm y})] $ &
7 & ~~$[R_{\rm z}] = [\frac{1}{\sqrt{2}}(I+i \sigma^{\rm z})] $ &
8 & ~~$[R_{\rm yz}] = [\frac{1}{\sqrt{2}}(\sigma^{\rm y} + \sigma^{\rm z})]$ \\
\hline
9 & ~~$[R_{\rm zx}] = [\frac{1}{\sqrt{2}}(\sigma^{\rm z} + \sigma^{\rm x})] $ &
10 & ~~$[R_{\rm xy}] = [\frac{1}{\sqrt{2}}(\sigma^{\rm x} + \sigma^{\rm y})]$ &
11 & ~~$[\pi_{\rm x}] = [\frac{1}{2}(I + \sigma^{\rm x})] $ &
12 & ~~$[\pi_{\rm y}] = [\frac{1}{2}(I + \sigma^{\rm y})]$ \\
\hline
13 & ~~$[\pi_{\rm z}] = [\frac{1}{2}(I + \sigma^{\rm z})] $ &
14 & ~~$[\pi_{\rm yz}] = [\frac{1}{2}(\sigma^{\rm y} +i \sigma^{\rm z})]$ &
15 & ~~$[\pi_{\rm zx}] = [\frac{1}{2}(\sigma^{\rm z} +i \sigma^{\rm x})]$ &
16 & ~~$[\pi_{\rm xy}] = [\frac{1}{2}(\sigma^{\rm x} +i \sigma^{\rm y})] $ \\
\hline
\end{tabular}
\end{center}
\caption{
Sixteen basis operations. These operations are composed of single qubit rotations and measurements. $[I]$ denotes an identity operation (no operation), $[\sigma^{i}]~( i=x,y,x)$ corresponds to operations applying Pauli matrices. $[\pi]$ corresponds to projective measurements. 
}
\label{tab:bases}
\end{table}

Here, we denote the complete basis operations as $\{\mathcal{B}_i\}$. For multiple qubit systems, tensor products of single qubit operations, e.g., $\mathcal{B}_i \otimes \mathcal{B}_j$ also forms a complete basis set for composite systems. Therefore, we can decompose any $n$-qubit interaction into the basis $\{\mathcal{B}_i\}^{\otimes n}$. While the decomposition is universal, it may produce a large decomposition coefficient, making it inefficient to implement.
We can thus consider an over-complete basis with generalised quantum operations and find an optimised decomposition. Specifically, consider a set of over-complete basis $\{\Phi_k\}$ which includes the identity channel, our target is to solve the following problem.

\begin{equation}
\begin{aligned}
\min C_1 &= \sum_{k} \alpha^+_k - \sum_{k'} \alpha^-_{k'},\\
\textrm{such that}~&\mc V^{\rm int}(\delta t)= \sum_{k} \alpha^+_k \Phi_{k} - \sum_{k'} \alpha^-_{k'}\Phi_{k'},\\
&\alpha^+_k, \alpha^-_k\ge 0.
\end{aligned}
\end{equation}

There are a few problems here. First, the optimisation becomes exponentially costly when the channel acts on a large $n$ qubits. Second, the basis operation also contains measurement and state preparation, which might be challenging in experiment. 
In the next section, we give another explicit decomposition strategy that resolves these problems. The explicit decomposition could be optimal under mild conditions and it only requires unitary operations without measurements or state preparation.

\section{An explicit decomposition method}
\label{appendix:explicit}

\subsection{Method description}
In this section, we consider an explicit decomposition which only involves unitary operations. Supposing $V^{\rm int} = \sum_j \lambda_jV_j^{\rm int}$ with each $V_j^{\rm int}$ being a tensor product of Pauli operators, we consider the expansion
\begin{equation}\label{Eq:Hdecom}
\begin{aligned}
	\mc V^{\rm int}(\delta t)[\rho] &= \mc I(\rho) - i\delta t(V^{\rm int}\rho I -\rho V^{\rm int})+ O(\delta t^2),\\
	&= \mc I(\rho) - i\delta t\sum_j \lambda_j (V_j^{\rm int}\rho-\rho V_j^{\rm int})+ O(\delta t^2),
\end{aligned}
\end{equation}
where both $V_j^{\rm int}\rho$ and $\rho V_j^{\rm int}$ are generalised quantum operations. Suppose $V_j^{\rm int} = \bigotimes_l V_{l,j}^{\rm int}$ and the input state is a product state $\rho = \bigotimes_l \rho_l$ the above decomposition could be expressed generally as
\begin{equation}
\begin{aligned}
	\mc V^{\rm int}(\delta t)\bigg[\bigotimes_l \rho_l\bigg] &= c(\delta t)\sum_k e^{i\theta_k}p_k \bigotimes_l \left[\tilde U_{l,k} \rho_l \tilde V_{l,k}\right]+ O(\delta t^2).
\end{aligned}
\end{equation}
Here each $\tilde U_{l,k}$ and $\tilde V_{l,k}$ could be $I$ and $V_{l,j}^{\rm int}$, $c(\delta t) = 1+2\delta t\sum_j|\lambda_j|$,  $p_k$ and $\theta_k$ are defined correspondingly.
Denoting the unitary evolution of the $l$th subsystem as $U_l(\delta t)$ and following the notation of the above discussion,  the joint evolution of all the subsystems is
\begin{equation}\label{Eq:explicitwhole}
    \mc U(T)\bigg[\bigotimes_l \rho_l\bigg] = C\sum_{\mathbf k}e^{i\theta_{\mathbf k}}p_{\mathbf k}\bigotimes_l \left[\tilde U_{l,k_{T/\delta t}}U_l(\delta t)\dots \tilde U_{l,{k_1}}U_l(\delta t) \rho_l U^\dag_l(\delta t)\tilde V_{l,k_1}\dots U^\dag_l(\delta t)\tilde V_{l,k_{T/\delta t}}\right],
\end{equation}
where $C = c(\delta t)^{T/\delta t} =e^{2T\lambda}$ with $\lambda = \sum_j|\lambda_j|$. Now we have decoupled the joint evolution as a linear combination of independent evolution of each subsystem. When we further implement the stochastic Monte Carlo method, the evolution of each subsystem looks like
\begin{equation}
    \rho_{l,\mathbf k} = \tilde U_{l,k_{N_{\rm jp}}}U_l(t_{N_{\rm jp}})\dots \tilde U_{l,k_1}U_l(t_1) \rho_l U^\dag_l( t_1)\tilde V_{l,k_{1}}\dots U^\dag_l(t_{N_{\rm jp}})\tilde V_{l,k_{N_{\rm jp}}},
\end{equation}
where $N_{\rm jp}$ is the number of jumps or decay events and $t_1+t_2+\dots+t_{N_{\rm jp}} = T$. Here each $\tilde U_{l,k_{i}}$ is either $I$ or one of $\{V^{\rm int}_{l,j}\}$. When we measure $O_l$, it becomes
\begin{equation}
    \tr[\rho_{l,\mathbf k} O_l] = \tr[\tilde U_{l,k_{N_{\rm jp}}}U_l(t_{N_{\rm jp}})\dots \tilde U_{l,k_1}U_l(t_1) \rho_l U^\dag_l( t_1)\tilde V_{l,k_{1}}\dots U^\dag_l(t_{N_{\rm jp}})\tilde V_{l,k_{N_{\rm jp}}}O],
\end{equation}
which could be implemented with the following circuit
\begin{align*}
\Qcircuit @C=0.8em @R=1.2em {
\lstick{\ket{+}}&\qw&\ctrl{2}& \qw& \qw&\qw& \qw&\ctrl{2}&\measureD{X, Y}\\
&&&&\cdots&&\\
\lstick{\rho_l}&\gate{U_{l}(t_1)}&\gate{\tilde U_{l,k_1}\tilde V^\dag_{l,k_1}}&\qw& \qw& \qw&\gate{U_{l}(t_{N_{\rm jp}})}&\gate{\tilde U_{l,k_{N_{\rm jp}}}\tilde V^\dag_{l,k_{N_{\rm jp}}}}&\measureD{O_l}\\
}
\end{align*}
The measurement result of the whole evolution state is 
\begin{equation}
    \tr\left[\mc U(T)\bigg[\bigotimes_l \rho_l\bigg]\cdot\bigotimes_l O_l\right] = C\sum_{\mathbf k}e^{i\theta_{\mathbf k}}p_{\mathbf k}\prod_{l} \tr[\rho_{l,\mathbf k} O_l].
\end{equation}
Therefore, after measuring each $\tr[\rho_{l,\mathbf k} O_l]$, we can obtain the exact measurement result.

\subsection{Cost analysis}
\label{appendix:Explicit_Cost}
According to the above discussion, the cost associated with the explicit expansion is 
\begin{equation}
    C = e^{2T\lambda},
\end{equation}
with $\lambda = \sum_j|\lambda_j|$. We show that the expansion is optimal, i.e., with the smallest cost when $V^{\rm int}$ satisfies the following condition.

\begin{condition}\label{condition:1}
Suppose $V^{\rm int}$ acts nontrivially on the set of subsystems $\mc S$. Given $V^{\rm int} = \sum_j\lambda_j V^{\rm int}_j$ with each $V^{\rm int}_j = \bigotimes_l V^{\rm int}_{l,j}$ and $V^{\rm int}_{l,j}$ being a tensor product of Pauli operators, we have 
\begin{equation}
\begin{aligned}
\tr\left[V^{\rm int}_{l,j}\right] &= 0, \, \forall j,\forall l\in \mc S,\\
\tr\left[V^{\rm int}_{l,j}V^{\rm int}_{l,j'}\right] &= 0, \, \forall j\neq j',\forall l\in \mc S.
\end{aligned}
\end{equation}
\end{condition}
\noindent The first condition requires that each $V^{\rm int}_{l,j}$ is non-identity and the second condition requires two interaction terms of the same system are orthogonal. When we say $V^{\rm int}$ acts nontrivially on subsystems $\mc S$, it means that for any $l\in \mc S$, at least one of $V^{\rm int}_j$ has non-identity Pauli operators on subsystem $l$.

We summarise the result as follows.
\begin{theorem}
\label{thm:min_cost}
Suppose the interaction $V^{\rm int}$ satisfies Condition~\ref{condition:1}. The explicit expansion of  \autoref{Eq:Hdecom} has the minimal cost under all possible decomposition strategies.  
\end{theorem}
\begin{proof}
We assume that $V^{\rm int}$ acts nontrivially on all the systems. The proof is similar to the general case.
The Choi state of the interaction part $\tilde{\mc V}^{\rm int}[\rho] = - iV^{\rm int}\rho+i\rho V^{\rm int}$ is
\begin{equation}
    \tilde \phi^{\rm int}_{1,1',\dots,L,L'} = \tilde{\mc V}^{\rm int}\left[\bigotimes_l \phi_{l,l'}\right] = -i\sum_j\lambda_j\bigg[ \bigotimes_l V^{\rm int}_{l,j}\ket{\phi}_{l,l'}\bra{\phi}_{l,l'}-\bigotimes_l \ket{\phi}_{l,l'}\bra{\phi}_{l,l'}V^{\rm int}_{l,j}\bigg].
\end{equation}
Focusing on $\tilde \psi^{\rm int}_{1,1}$ for example, we have
\begin{equation}
    \tilde \psi^{\rm int}_{1,1'} =  -i\sum_j\lambda_j\bigg[ V^{\rm int}_{1,j}\ket{\phi}_{1,1'}\otimes \ket{\phi}_{1,1'}\bigotimes_{l\ge 2} \bra{\phi}_{l,l'}(V^{\rm int}_{l,j})^T\otimes \bra{\phi}_{l,l'}-\ket{\phi}_{1,1'}\otimes (V^{\rm int}_{1,j})^T\ket{\phi}_{1,1'}\bigotimes_l \bra{\phi}_{l,l'}\otimes \bra{\phi}_{l,l'}V^{\rm int}_{l,j}\bigg].
\end{equation}
Denoting 
\begin{equation}
    \begin{aligned}
            \ket{\psi_{1,j,a}} &= V^{\rm int}_{1,j}\ket{\phi}_{1,1'}\otimes \ket{\phi}_{1,1'},\\
            \ket{\psi_{1,j,b}} &= \ket{\phi}_{1,1'}\otimes (V^{\rm int}_{1,j})^T\ket{\phi}_{1,1'},\\
            \bra{\psi_{l,j,a}} &= \bra{\phi}_{l,l'}(V^{\rm int}_{l,j})^T\otimes \bra{\phi}_{l,l'},\\
            \bra{\psi_{l,j,b}} &= \bra{\phi}_{l,l'}\otimes \bra{\phi}_{l,l'}V^{\rm int}_{l,j},\\
    \end{aligned}
\end{equation}
we can express $\tilde \psi^{\rm int}_{1,1}$ as
\begin{equation}\label{Eq:expansionpsiint11}
    \tilde \psi^{\rm int}_{1,1'} =  -i\sum_j \lambda_jc\bigg[\ket{\psi_{1,j,a}}\bigotimes_{l\ge 2}\bra{\psi_{l,j,a}} - \ket{\psi_{1,j,b}}\bigotimes_{l\ge 2}\bra{\psi_{l,j,b}}\bigg].
\end{equation}
When $\{V_{l,j}\}$ satisfies Condition~\ref{condition:1}, elements in $\{\ket{\psi_{1,j,a}},\ket{\psi_{1,j,b}}\}$ are mutually orthogonal, i.e., 
\begin{equation}
\begin{aligned}
    \braket{\psi_{1,j,x}|\psi_{1,j',y}} &= 0, \forall j\neq j'\,\,\textrm{or}\,\,x\neq y,\, x,y\in\{a,b\}.
\end{aligned}
\end{equation}
Similarly elements in $\{\bra{\psi_{l,j,a}}, \bra{\psi_{l,j,b}}\}$ are mutually orthogonal. Therefore,  \autoref{Eq:expansionpsiint11} is a singular value decomposition of $\tilde \psi^{\rm int}_{1,1}$ and we have
\begin{equation}
    \|\tilde \psi^{\rm int}_{1,1'}\|_1 = 2\sum_{j}|\lambda_j|.
\end{equation}
The above proof holds for all other $\tilde \psi^{\rm int}_{l,l'}$.

\end{proof}

Here, we give several examples of $V^{\rm int}$ that satisfy the condition. First, the condition is satisfied when there is only one interaction term $V^{\rm int}_1$.

\begin{corollary}
The decomposition is optimal when $V^{\rm int}$ only has one term (a tensor product of Pauli matrices).
\label{corollary:one_term}
\end{corollary}

The Condition~\ref{condition:1} could also hold when the interaction $V^{\rm int}$ has multiple terms. For example, consider three subsystems and denote $(l,m)$ to be the $m$th qubit of the $l$th subsystem. The following example interactions satisfy  Condition~\ref{condition:1}. 
\begin{equation}
\begin{aligned}
    V^{\rm int} =& aX_{1,1}\cdot X_{2,1}\cdot X_{3,1} + bX_{1,2}\cdot X_{2,2}\cdot X_{3,2} + cX_{1,3}\cdot X_{2,3}\cdot X_{3,3},\\
    V^{\rm int} =& aX_{1,1}\cdot X_{2,1}\cdot X_{3,1} + bY_{1,1}\cdot Y_{2,1}\cdot Y_{3,1} + cZ_{1,1}\cdot Z_{2,1}\cdot Z_{3,1},\\
    V^{\rm int} =& aX_{1,1}Y_{1,2}\cdot X_{2,1}Z_{2,2}\cdot Y_{3,1}Y_{3,2}+bX_{1,1}Z_{1,2}\cdot Z_{2,1}Z_{2,2}\cdot X_{3,1}Y_{3,1}\\
                &+cZ_{1,1}Y_{1,2}\cdot X_{2,1}Y_{2,2}\cdot Z_{3,1}Y_{3,2}+dZ_{1,1}Z_{1,2}\cdot Z_{2,1}Y_{2,2}\cdot X_{3,1}Z_{3,1}.
\end{aligned}
\end{equation}

Since Condition~\ref{condition:1} requires that the interaction terms of each subsystem are mutually orthogonal, it limits the number of interaction terms.  
\begin{proposition}
When $V^{\rm int} = \sum_{j=1}^{N_{\rm int}}\bigotimes_{l}V^{\rm int}_{l,j}$ satisfies Condition~\ref{condition:1} with $N_{\rm int}$ being the number of terms and $n$ being the minimal number of qubits of each subsystem. Suppose the minimal weight of each $V^{\rm int}_{l,j}$ is $k$, and we have
\begin{equation}
    N_{\rm int} \le 3^k{n\choose k}.
\end{equation}
\end{proposition}
\noindent In particular, when $k=1$, i.e., the interaction on each subsystem only act on one qubit, Condition~\ref{condition:1} requires $V^{\rm int}$ to have at most $3n$ terms. A more detailed condition for an optimal decomposition could be an interesting future work. 

\subsection{Discussion}

The major limitation of PQS is the sample cost.
In particular, the variance in our method increases exponentially with the total coupling strength, we thus need exponential samples for systems with large couplings. Such limitations are indeed fundamental, which would generally appear whenever we try to approximate larger systems with smaller ones, in a similar sense to the hardness of classical simulation of arbitrary quantum systems. Nevertheless, since PQS directly simulates the interacting dynamics, it has a smaller overhead compared with other recent works that simulate clustered Hamiltonians or circuits using the gate decomposition~\cite{peng2020simulating,barratt2020parallel,mitarai2021constructing,fujii2020deep,yuan2020quantum,Mitarai2021overheadsimulating}.

In addition, one could consider implementing finite-order Dyson series expansion to reduce the sampling overhead at the expense of increasing simulation error. As discussed in \autoref{appendix:Dyson}, there is a trade-off between the sampling cost and the simulation error.

\section{Dyson series method}
\label{appendix:Dyson}
We show in this section that the above explicit expansion method could be reformulated via the Dyson series expansion. 

\subsection{Method description}
We first introduce the general method. 
Consider the time evolution with Hamiltonian $H = H^{\rm loc} + V^{\rm int}$,
\begin{equation}
\left[H^{\rm loc} + V^{\rm int}\right]|\psi(t)\rangle=i  \frac{\partial|\psi(t)\rangle}{\partial t},
\end{equation}
with $H^{\rm loc} = \sum_lH_l$ and $V^{\rm int}=\sum_j \lambda_j V_j^{\rm int} = \sum_j \lambda_j \prod_l V_{l,j}^{\rm int}$. It becomes 
\begin{equation}
\lambda e^{{i} H^{\rm loc}\left(t-t_{0}\right)} V^{\rm int} e^{-{i} H^{\rm loc}\left(t-t_{0}\right)}\left|\psi_{I}(t)\right\rangle=i  \frac{\partial\left|\psi_{I}(t)\right\rangle}{\partial t}
\end{equation}
under the interaction picture with
\begin{equation}
|\psi(t)\rangle=e^{-{i} H^{\rm loc}\left(t-t_{0}\right)}\left|\psi_{I}(t)\right\rangle.
\end{equation}
A solution with Dyson series is
\begin{equation}
\begin{aligned}
	\left|\psi_{I}(t)\right\rangle=&\bigg[1-{i} \int_{t_{0}}^{t} d t_{1} e^{{i} H^{\rm loc}\left(t_{1}-t_{0}\right)} V^{\rm int} e^{-{i} H^{\rm loc}\left(t_{1}-t_{0}\right)}\\
	&-\int_{t_{0}}^{t} d t_{1} \int_{t_{0}}^{t_{1}} d t_{2} e^{{i} H^{\rm loc}\left(t_{1}-t_{0}\right)} V^{\rm int} e^{-{i} H^{\rm loc}\left(t_{1}-t_{0}\right)} e^{{i} H^{\rm loc}\left(t_{2}-t_{0}\right)} V^{\rm int} e^{-{i} H^{\rm loc}\left(t_{2}-t_{0}\right)}+\ldots\bigg]\left|\psi\left(t_{0}\right)\right\rangle
\end{aligned}
\end{equation}
To measure any observable $O=\bigotimes_l O_l$ with $O_I=e^{{i} H^{\rm loc}\left(t-t_{0}\right)} O e^{-{i} H^{\rm loc}\left(t-t_{0}\right)}$, we have
\begin{equation}
	\begin{aligned}
		\braket{\psi_{I}(t)|O_I|\psi_{I}(t)} = & \braket{\psi\left(t_{0}\right)|e^{{i} H^{\rm loc}\left(t-t_{0}\right)} O e^{-{i} H^{\rm loc}\left(t-t_{0}\right)}|\psi\left(t_{0}\right)},\\
		&-2(t-t_0) \int_{t_{0}}^{t} \frac{d t_{1}}{t-t_0}  \Re\bigg[{i } \braket{\psi\left(t_{0}\right)| e^{{i} H^{\rm loc}\left(t-t_{0}\right)} O e^{-{i} H^{\rm loc}\left(t-t_{1}\right)} V^{\rm int} e^{-{i} H^{\rm loc}\left(t_{1}-t_{0}\right)}|\psi\left(t_{0}\right)}\bigg],\\
	\end{aligned}
\end{equation}
up to the first order expansion. Suppose $\ket{\psi(t_0)}=\ket{\psi_1(t_0)}\dots \ket{\psi_L(t_0)}$, the first term is
\begin{equation}
    \braket{\psi\left(t_{0}\right)|e^{{i} H^{\rm loc}\left(t-t_{0}\right)} O e^{-{i} H^{\rm loc}\left(t-t_{0}\right)}|\psi\left(t_{0}\right)} = \prod_l \braket{\psi_l\left(t_{0}\right)|e^{{i} H_l\left(t-t_{0}\right)} O e^{-{i} H_l\left(t-t_{0}\right)}|\psi_l\left(t_{0}\right)},
\end{equation}
where each term can be easily measured by evolving each subsystem state with the Hamiltonian $H_l$ and measure $O_l$. To measure the second term, we can uniformly sample $t_1$ from $t_0$ to $t$ and use the following circuit 
\begin{align*}
\Qcircuit @C=0.8em @R=1.2em {
\lstick{(\ket{0}+\ket{1})/\sqrt{2}}&\qw&\ctrlo{1}&\qw& \measureD{X, Y}\\
\lstick{\ket{\psi_l\left(t_{0}\right)}}&\gate{e^{-{i} H_l\left(t_{1}-t_{0}\right)}}&\gate{V_{l,j}^{\rm int}}&\gate{e^{-{i} H_l\left(t-t_{1}\right)}}&\measureD{O}\\
}
\end{align*}
to get $\braket{O}_{l,j, t_1}=\braket{\psi_l\left(t_{0}\right)| e^{{i} H_l\left(t-t_{0}\right)} O_l e^{-{i} H_l\left(t-t_{1}\right)} V^{\rm int}_{l,j} e^{-{i} H_l\left(t_{1}-t_{0}\right)}|\psi\left(t_{0}\right)}$. Then we have
\begin{equation}
    \braket{\psi\left(t_{0}\right)| e^{{i} H^{\rm loc}\left(t-t_{0}\right)} O e^{-{i} H^{\rm loc}\left(t-t_{1}\right)} V^{\rm int} e^{-{i} H^{\rm loc}\left(t_{1}-t_{0}\right)}|\psi\left(t_{0}\right)} = \sum_j \lambda_j \braket{O}_{1,j,t_1}\braket{O}_{2,j,t_1}\cdots \braket{O}_{L,j,t_1}.
\end{equation}
We can also randomly measure the first term or the second term, as well as each term of the above summation. The cost is now 
\begin{equation}
    C_1 = 1 + 2T\lambda,
\end{equation}
with $\lambda = \sum_j |\lambda|$.
Due to the first order expansion, the approximation error is 
\begin{equation}
    \varepsilon_1 = \mc O\left(e^{|V^{\rm int}|T}(|V^{\rm int}|T)^{2}\right).
\end{equation}
We can similarly consider expansion to the $k$th order, then the cost and the expansion error are
\begin{equation}
    \begin{aligned}
    C_k &= \sum_{n=0}^k (2T\lambda)^n/n!,\\
    \varepsilon_k &= \mc O\left(e^{|V^{\rm int}|T}(|V^{\rm int}|T)^{k+1}/(k+1)!\right).
    \end{aligned}
\end{equation}
With the limit of $k\rightarrow\infty$, we have
\begin{equation}
    \begin{aligned}
    C_{\infty} &= e^{2T\lambda},\\
    \varepsilon_{\infty} &= 0.
    \end{aligned}
\end{equation}
In this case, the cost is the same as the one for the explicit expansion. 
In the next subsection, we show that they are actually equivalent.
Several posted work has proposed the quantum simulation with truncated Dyson series on a universal quantum computer in Refs.~\cite{kieferova2019simulating,chen2021quantum}.

\subsection{Relation to the perturbative quantum simulation method}

% The decomposition of  \autoref{Eq:Vdecommain} holds in general for a complete set of $\{\Phi_k\}$, we show an explicit decomposition as 
% \begin{equation}\label{Eq:explicit_decomp_appendix}
%     \mc V(\rho, \delta t) = \mathcal{I}(\rho) - i\delta t \sum_j \lambda_j(V_i \rho - \rho V_i),
% \end{equation}
% which is universal and simplifies the implementation without involving the ancillary qubits in  \autoref{Eq:goperation}. 

The algorithm using Dyson series implements each expanded term (trajectory) with a quantum computer and sums over the expansion via the average of different trajectories. This is very similar to the above perturbative quantum simulation method. We show that they are actually equivalent. 
\begin{theorem}
\label{thm:dyson}
The infinite-order Dyson series method is equivalent to the perturbative quantum simulation method with the explicit decomposition. 
\end{theorem}
\begin{proof}
To see the equivalence, we first consider a pure state formalism for the perturbative quantum simulation method with the explicit decomposition. Suppose the interaction term $\mc V^{\rm int}(\delta t)$ is decomposed as follows,
\begin{equation}\label{Eq:Vint}
\begin{aligned}
	\mc V^{\rm int}(\delta t)[\rho] = \mc I(\rho) - i\delta t\sum_j \lambda_j (V_j^{\rm int}\rho-\rho V_j^{\rm int})+ O(\delta t^2).
\end{aligned}
\end{equation}
Suppose $\rho$ is a pure state $\rho=\ket{\psi}\bra{\psi}$,  then 
\begin{equation}\label{Eq:Vint_pure}
\begin{aligned}
	\mc V^{\rm int}(\delta t)[\ket{\psi}\bra{\psi}] =& \mc I(\rho) - i\delta t\sum_j \lambda_j (V_j^{\rm int}\ket{\psi}\bra{\psi}-\ket{\psi}\bra{\psi}V_j^{\rm int})+ O(\delta t^2),\\
	=& \bigg(I-i\delta t\sum_j \lambda_j V_j^{\rm int}\bigg) \ket{\psi}\bra{\psi} \bigg(I+i\delta t\sum_j \lambda_j V_j^{\rm int}\bigg)+ O(\delta t^2).
\end{aligned}
\end{equation}
Then the whole time evolution with a pure input state $\ket{\psi(t_0)}$ is
\begin{equation}\label{Eq:pureexpand}
\begin{aligned}
	U(T)\ket{\psi(t_0)} = \bigg[
	e^{-iH^{\rm loc}\delta t}\bigg(I-i\delta t V^{\rm int}\bigg)\bigg]^{{T/\delta t}} \ket{\psi(t_0)}+ O(\delta t^2),
\end{aligned}
\end{equation}
and we have
\begin{equation}
    \mc U(T)[\psi(t_0)] = U(T)\ket{\psi(t_0)}\bra{\psi(t_0)}U(T)^\dag.
\end{equation}
Then each expanded term in $\mc U(T)[\psi(t_0)]$ corresponds to the expanded terms in $U(T)\ket{\psi(t_0)}$ and $\bra{\psi(t_0)}U(T)^\dag$. 
Now we expand the product of  \autoref{Eq:pureexpand} and group the terms according to the number of $V^{\rm int}$ as
\begin{equation}
\begin{aligned}
	U(T)\ket{\psi(t_0)} =& \bigg[e^{-iH^{\rm loc}T} -i\delta t \sum_{i = 1}^{T/\delta t}e^{-iH^{\rm loc}(T-i \delta t)}V^{\rm int}e^{-iH^{\rm loc}i\delta t}  ,\\
	&-\delta t^2 \sum_{i_1\ge i_2 = 1}^{T/\delta t}e^{-iH^{\rm loc}(T-i_1 \delta t)}V^{\rm int}e^{-iH^{\rm loc}(i_1-i_2)\delta t}V^{\rm int}e^{-iH^{\rm loc}i_2\delta t}\bigg] \ket{\psi(t_0)}+ O(\delta t^2).
\end{aligned}
\end{equation}
Multiplying $e^{iH^{\rm loc}T}$ and taking the limit of $\delta t\rightarrow 0$ we have
\begin{equation}
\begin{aligned}
	\lim_{\delta t\rightarrow 0}e^{iH^{\rm loc}T}U(T)\ket{\psi(t_0)} =& \bigg[1-{i} \int_{t_{0}}^{T} d t_{1} e^{{i} H^{\rm loc}\left(t_{1}-t_{0}\right)} V^{\rm int} e^{-{i} H^{\rm loc}\left(t_{1}-t_{0}\right)}\\
	&-\int_{t_{0}}^{T} d t_{1} \int_{t_{0}}^{t_{1}} d t_{2} e^{{i} H^{\rm loc}\left(t_{1}-t_{0}\right)} V^{\rm int} e^{-{i} H^{\rm loc}\left(t_{1}-t_{0}\right)} e^{{i} H^{\rm loc}\left(t_{2}-t_{0}\right)} V^{\rm int} e^{-{i} H^{\rm loc}\left(t_{2}-t_{0}\right)}+\ldots\bigg]\left|\psi\left(t_{0}\right)\right\rangle,
\end{aligned}
\end{equation}
which coincides with the Dyson series expansion. 
\end{proof}

We remark that the expansion is universal and avoids the  computational cost in diagrammatic perturbation theory. 
The algorithm with explicit decomposition in  \autoref{Eq:Hdecom} effectively implements each expanded term (trajectory) and realise the joint time evolution by summing over the expansion via the average of different trajectories.

% The fixed decomposition corresponds to the infinite-order Dyson series expansion of the unitary $U(t)=1-{i} \sum_j \lambda_j\int_{t_{0}}^{t} d t_{1} e^{{i} H^{\textrm{loc}}\left(t_{1}-t_{0}\right)} V_i\left(t_{1}\right) e^{-{i} H^{\textrm{loc}}\left(t_{1}-t_{0}\right)} + \dots$ and the algorithm  implements each expanded term (trajectory) with a quantum computer and sums over the expansion via the average of different trajectories. 

% \section{Improved accuracy with smaller circuits}

% \section{Error mitigation?}
% Compatible to the error mitigation techniques.

\section{Higher-order moments analysis}
In the above discussion, we showed how to use the perturbative quantum simulation method to get linear observable measurement. In this section, we show that the PQS method applies to measurement on higher-order moments of the state. We take the subsystem purity as an example, and we note that the result applies to general measurements. Without loss of generality, we consider the purity $\tr [\rho_1^2(T)]$ of the first subsystem, and we denote the set without the first system as $\mc S = \{2,..,L\}$. Following the  PQS method with the explicit expansion in  \autoref{Eq:explicitwhole}, we have
\begin{equation}
    \mc U(T)\bigg[\bigotimes_l \rho_l\bigg] = C\sum_{\mathbf k}e^{i\theta_{\mathbf k}}p_{\mathbf k}\bigotimes_l \left[\tilde U_{l,k_{T/\delta t}}U_l(\delta t)\dots \tilde U_{l,{k_1}}U_l(\delta t) \rho_l U^\dag_l(\delta t)\tilde V_{l,k_1}\dots U^\dag_l(\delta t)\tilde V_{l,k_{T/\delta t}}\right],
\end{equation}
where the input state is $\bigotimes_l \rho_l$, $\tilde U_{l,k}$ and  $\tilde V_{l,k}$ are either $I$ or $V^{\rm int}_{l,j}$.
Now we calculate the reduced density matrix of the first subsystem,
\begin{equation}
\begin{aligned}
	\rho_1(T)
 &=\tr_{\mc S}[\mc U(T)\bigg[\bigotimes_l \rho_l\bigg]],\\
 &=\sum_{\mathbf k} \beta_{\mathbf k} \tilde U_{1,k_{T/\delta t}}U_1(\delta t)\dots \tilde U_{1,{k_1}}U_1(\delta t) \rho_1 U^\dag_1(\delta t)\tilde V_{1,k_1}\dots U^\dag_1(\delta t)\tilde V_{1,k_{T/\delta t}},
	\end{aligned}
	\end{equation}
where
\begin{equation}
    \beta_{\mathbf k} = C e^{i\theta_{\mathbf k}}p_{\mathbf k}\prod_{l\in \mc S} \tr\left[\tilde U_{l,k_{T/\delta t}}U_l(\delta t)\dots \tilde U_{l,{k_1}}U_l(\delta t) \rho_l U^\dag_l(\delta t)\tilde V_{l,k_1}\dots U^\dag_l(\delta t)\tilde V_{l,k_{T/\delta t}}\right],
\end{equation}
which could be measured for each ${\mathbf k}$. For the purity of the first subsystem, we have
% Following the derivation of unitary channel on pure state, we have
% \begin{equation}
% \begin{aligned}
% 	\ket{\psi(T)} & = \sum_{\mathbf j}  \tilde\Lambda_{\mathbf j}\bigotimes_l\bigg[\prod_{i=1}^{T/\delta t}  U_l V_{l,{j_i}}\bigg]\ket{\psi_{l}(0)} \\
% 	&= \sum_{\mathbf j}  \tilde\Lambda_{\mathbf j} \bigg[\prod_{i=1}^{T/\delta t}  U_1 V_{1,{j_i}}\bigg]\ket{\psi_{1}(0)} 
%  \bigotimes_{l \in S} \bigg[\prod_{i=1}^{T/\delta t}  U_l V_{l,{j_i}}\bigg]\ket{\psi_{l}(0)} \\
%  &= \sum_{\mathbf j}  \tilde\Lambda_{\mathbf j}  U_1^j \ket{\psi_1(0)} \bigotimes_{l \in S} U_l^j   \ket{\psi_{l}(0)} 
% 	\end{aligned}
% \end{equation}
% where we denote $U_l^j := \prod_{i=1}^{T/\delta t}  U_l V_{l,{j_i}}$ that acts on the $l$th subsystem for simplicity.

% The reduced density matrix for the first subsystem under evolution is 
% \begin{equation}
% \begin{aligned}
% 	\rho_1(T)
%  &= \sum_{\mathbf j, \mathbf k}  \left[  \tilde\Lambda_{\mathbf j} \tilde\Lambda_{\mathbf k} \prod_{l \in S}  \bra{\psi_{l}(0)}  (U_l^k)^{\dagger} U_l^j   \ket{\psi_{l}(0)}   \right] U_1^j\ket{\psi_1(0)}  
%  \bra{\psi_1(0)} (U_1^k)^{\dagger}\\
%  & = \sum_{\mathbf j, \mathbf k} \alpha_{\mathbf j, \mathbf k} \ket{\phi_{j}} \bra{ \phi_{ k}}
% 	\end{aligned}
% 	\end{equation}
% where we denote $\ket{\phi_{j}} := U_1^j\ket{\psi_1(0)}  $.
Therefore, we have 
\begin{equation}
\begin{aligned}
	\rho_1^2(T)
  = \sum_{\mathbf k,\mathbf k'} \beta_{\mathbf k}\beta_{\mathbf k'} &\tilde U_{1,k_{T/\delta t}}U_1(\delta t)\dots \tilde U_{1,{k_1}}U_1(\delta t) \rho_1 U^\dag_1(\delta t)\tilde V_{1,k_1}\dots U^\dag_1(\delta t)\tilde V_{1,k_{T/\delta t}} \\
  \cdot&\tilde U_{1,k'_{T/\delta t}}U_1(\delta t)\dots \tilde U_{1,{k'_1}}U_1(\delta t) \rho_1 U^\dag_1(\delta t)\tilde V_{1,k'_1}\dots U^\dag_1(\delta t)\tilde V_{1,k'_{T/\delta t}}.
\end{aligned}
\end{equation}
Suppose the initial state is pure $\rho_1=\ket{\psi_1}\bra{\psi_1}$, we have
\begin{equation}
\begin{aligned}
	\tr[\rho_1^2(T)]
  = \sum_{\mathbf k,\mathbf k'} \beta_{\mathbf k}\beta_{\mathbf k'} & \bra{\psi_1} U^\dag_1(\delta t)\tilde V_{1,k_1}\dots U^\dag_1(\delta t)\tilde V_{1,k_{T/\delta t}}\tilde U_{1,k'_{T/\delta t}}U_1(\delta t)\dots \tilde U_{1,{k'_1}}U_1(\delta t)\ket{\psi_1} \\
  \times& \bra{\psi_1} U^\dag_1(\delta t)\tilde V_{1,k'_1}\dots U^\dag_1(\delta t)\tilde V_{1,k'_{T/\delta t}}\tilde U_{1,k_{T/\delta t}}U_1(\delta t)\dots \tilde U_{1,{k_1}}U_1(\delta t)\ket{\psi_1}.
\end{aligned}
\end{equation}
We note that the two overlap terms could be evaluated with the circuits that are similar to the ones used for measuring linear observables.
In practice, we can use the Monte Carlo method to estimate the purity.
The sample complexity for the purity estimation is related to $ \sum_{\mathbf k,\mathbf k'}|\beta_{\mathbf k}\beta_{\mathbf k'}|\propto C^2$.

Other higher order moments can be derived similarly and we leave it to the dedicated readers.

\section{Limitations, applications, and comparisons}

In this section, we discuss limitations and potential applications   of the PQS algorithm. We also compare our method to existing hybrid embedding methods.

% In this section, we discuss potential applications and the practical implementation of the PQS algorithm in more details. 
\subsection{Limitations and applications}

Since PQS is a hybrid method that combines quantum computing and classical perturbation theory, it inherits their advantages as well as their  limitations. The major limitation of PQS comes from the limitation of classical perturbation theories, which generally only work for weak interacted systems. For PQS, since we simulate the coupling of subsystems using the perturbation theory, it is efficient only if that coupling is weak, and thus the system that our method can simulate is limited. 
Specifically, our method cannot work for large systems with general arbitrary two-body interactions, such as strongly correlated electrons,  
high-dimensional strongly interacting lattice Hamiltonians, or scenarios where the timescale is long. 
The potential solution is discussed in \autoref{appendix:explicit}.
It is worth to note that similar limitations prevail in almost all modern classical computing methods apart from perturbation theory, such as density functional theory (DFT), quantum Monte Carlo (QMC), tensor and neural network methods, etc; yet, as we discuss below, these limitations do not prevent their wide applications for realistic problems.

While we have noted that our method is limited compared to a universal quantum simulation algorithm, we also need to point out that these universal quantum algorithms generally rely on a fault-tolerant quantum computer, which is still challenging to realize with the current technology. This explains why many current works in quantum computing focus on the so-called NISQ era {or the early stage of fault-tolerant quantum computing}, where both the size (number of qubits) and (circuit) depth of the quantum hardware are limited. As we elaborate below, our method does have broad applications from simulating intricate quantum many-body systems,  probing interesting physics phenomena, to benchmarking larger quantum processors for NISQ devices and early stage of fault-tolerant quantum computers.  

% In theory, the PQS method combines the complementary strengths of quantum computing and perturbation theory. Since there is no assumption on the major component, locality of the interactions or the initial state, PQS goes beyond conventional classical perturbative approaches.
% With this merit,  our method has broad applications in simulating quantum many-body systems of physical interest as well as probing the interesting phenomena. We showcased the real-world physical examples in the following.
% it applicable to simulate large systems with weak inter-subsystem interactions or intermediate systems with general interactions.

    Theoretically, our method combines the complementary strengths of quantum computing and perturbation theory, to respectively simulate the subsystem and the inter-subsystem interactions. PQS would be most powerful to investigate large systems with weak inter-subsystem interactions or intermediate systems with general interactions. Since there is no assumption on the subsystem interactions, locality of the inter-subsystem interactions, or the initial state of the subsystems, the method is widely applicable.

    The most promising and exciting application of perturbative quantum simulation is for clustered subsystems with weak subsystem-wise interactions. 
    % Consider that the whole quantum system is divided into several subsystems, where particles in the same subsystem are strongly interacted, and the particles from different subsystems are weakly interacted. 
    Since all subsystem could have arbitrarily strong interactions, the whole system, in general, might not be efficiently solvable using classical approaches.
    Our methods can thus be used to efficiently study the dynamics of these systems.
    One prominent example is for simulating 1D systems, where we could easily divide the systems into clustered subsystems and have weak subsystem-wise interactions. Although, one may argue that, under certain assumptions (local and gapped Hamiltonians), the ground state of 1D systems could be efficiently solvable using matrix product states,  hence the ability of simulating 1D system is not surprising. 
    However, we need to point out that
    simulating the dynamics of general 1D systems is actually a very challenging task for classical methods (see  Ref.~\cite{kim2017holographic}).

    Aside from the physical systems featuring these geometrical cluster properties, PQS is also  applicable for systems with multiple  degrees of freedom.
    Indeed, quantum many-body systems that consist of both weak and strong correlations in different levels of the system could be suitable for our methods.
    For example, considering the Hamiltonian of molecules, we can divide the system into electrons and nuclei. Then we can separately simulate the two subsystems of the electrons and process their correlation classically to surpass the Born-Oppenheimer approximation. We may use the PQS method to investigate the dynamical correlations beyond the Born-Oppenheimer approximation.

 As we have demonstrated in our main text, we can apply PQS to study the quantum walk of bosons, dynamical phase transition, the propagation of correlations, and spin-charge separation of bosons, fermions, or spins. Apart from these applications, our method could also be applied for studying other more general dynamical behaviours, such as molecular reactions of the dimer, and the electron-phonon interaction in superconducting models. Following a recent study  of cluster simulation schemes~\cite{peng2020simulating}, our method might also be applicable to variational quantum simulation for molecular Hamiltonians. 
 With proper embedding methods, such as DFT, DMFT, or DMET, PQS might also be used to as a subroutine to probe physical problems of practical relevance, for example, ones in the thermodynamic limit. 
  
Furthermore, PQS would be helpful for studying general strongly interacting problems with short time and benchmarking near-term quantum computation. 
% Even if the PQS method may not be as powerful as a large-scale universal quantum computer, it could give a solution with a compromised accuracy using a smaller quantum device. Specifically, 
Using the PQS method, a larger problem with $NL$ qubits could be processed by a $N+1$ qubit quantum device. 
% The calculation result could already be useful for probing many-body quantum physics phenomena, but more importantly, 
Since a smaller quantum device is generally much more accurate than a larger quantum processor due to crosstalks or other types of errors when controlling large quantum systems, our method could serve as a benchmark of the computing result for larger problems. This advantage was also clearly demonstrated using the IBM quantum cloud experiences in \autoref{appendix:IBM_implementation}. Therefore, when we construct a larger quantum hardware or aim to use it to demonstrate quantum advantages for solving a larger problem, we can first run our PQS with a smaller device to test the performance of the hardware and the  feasibility for faithful implementation.

\subsection{Comparison to hybrid embedding methods}

While quantum computing could potentially solve classically intractable problems, it  also has limitations in terms of  the circuit depth and   qubit number in the near future. Thus, it still remains a challenge to be able to solve practical problems using current and near-term quantum computers.  On the other hand, noticing the fact that most quantum many-body systems have mixed strong and weak correlation, we only need to solve the strongly correlated degrees of freedom using quantum computer and calculate the remaining part at a mean-field level using classical computational method. Along this line, several {hybrid} methods have been proposed by exploiting different classical methods, such as density matrix embedding theory (DMET)~\cite{DMET2012,DMET2013,Wouters16,rubin2016hybrid}, dynamical mean-field theory (DMFT)~\cite{Troyer16,rungger2020dynamical,Chen_2021}, 
density functional theory embedding~\cite{Ivano21}, quantum defect embedding theory~\cite{Ma20, sheng21},
tensor network~\cite{barratt2021parallel,yuan2020quantum}, entanglement forging~\cite{PRXQuantum.3.010309,huembeli2022entanglement}, virtual orbital approximation~\cite{PhysRevX.10.011004}, quantum Monte Carlo~\cite{Li_2019,Huggins_2022,2022arXiv220514903X,2022arXiv220509203M,2021arXiv211202190P,2022arXiv220610431Z}, etc. Most  hybrid methods have their own assumptions and specific applications, how to invent new hybrid methods with less stringent assumptions and more general applications still remain an open and exciting direction. In the following, we give a detailed comparison between our method and existing ones for the interested reader.

The hybrid methods listed above could be somehow understood as an embedding method. That is, with the help of certain classical means, we could solve a large problem only using a smaller quantum computer. Below, we provide some perspective on embedding methods.

\begin{itemize}
\item The methods exploiting density matrix embedding theory~\cite{DMET2012,DMET2013,Wouters16,rubin2016hybrid}, dynamical mean field theory~\cite{Troyer16,rungger2020dynamical,Chen_2021}, 
density functional theory embedding~\cite{Ivano21}, and quantum defect embedding theory~\cite{Ma20, sheng21} are more like conventional embedding methods, which approximates the solution by solving a self-consistent condition.

\item The methods exploiting tensor network~\cite{barratt2021parallel,yuan2020quantum}, entanglement forging~\cite{PRXQuantum.3.010309,huembeli2022entanglement}, and virtual orbital approximation~\cite{PhysRevX.10.011004} essentially introduces some ansatz that represents a larger quantum state with a smaller quantum computer and variationally solve the problem.

\item At last, the recently proposed quantum computing quantum Monte Carlo methods~\cite{Li_2019,Huggins_2022,2022arXiv220514903X,2022arXiv220509203M,2021arXiv211202190P,2022arXiv220610431Z} aims to stochastically realize the imaginary time evolution by a proper basis rotation generated from quantum circuits. 
\end{itemize}

The above methods aim to solve the eigenstate problems, and their successes rely on different assumptions. For instance, the hybrid methods exploiting density matrix embedding theory assumes the self-consistent (mean-field) condition via 1 reduced density matrices, which could not guarantee to find the true ground state. 
In DMET, one  match the 1 reduced density matrices of the interacting subsystem (fragment + bath) and the noninteracting system by tuning the external potential of the noninteracting system (as a variational parameter).
The assumption in DMET is that on the level of expectation value, the true ground state of the large system can be approximated by the ground state of a noninteracting system, which does not hold rigorously in general cases.
Besides, one can find that  only the static property of the ground state is matched, while the dynamical part is missing. Only in the extreme case it recovers the dynamical mean-field theory (DMFT), which is rigorous in the infinite dimension, while the phenomena usually emerge in the low dimension.

For the entanglement forging method, it assumes low entanglement between the subsystems and high expressivity of the circuit ansatz for local systems.
The representation essentially relies on the expressivity of the parametrized quantum circuit. This  could be insufficient for complex systems with large entanglement and circuit ans\"atze with limited expressivity.
Besides, in their demonstration of application in the ground state problem, it inherits the limitation of the optimization problem in variational algorithms (such as in VQE), while the problem size is reduced to half of the original problem. 

In what follows, we compare our method to these existing hybrid methods and highlight the differences.
An apparent difference is that we focus on quantum dynamics, a fundamentally different but meaningful problem that has wide applications and draws great interests from researchers from different areas. A way to efficiently simulate the dynamics of large quantum systems would be greatly helpful for understanding many interesting physics phenomena, such as phase transition, nonadiabatic evolution, or quench dynamics, etc. Our algorithm would be of particularly useful to understand dynamical effects of large quantum systems using near-term and early fault-tolerant quantum computers with limited qubit number and circuit depth.
    
Meanwhile, compared to the existing methods, the underlying mechanism is fundamentally different. Our algorithm relies on a stochastic implementation of the Dyson series expansion of the Hamiltonian evolution operator. Our algorithm is exact, in the sense that it directly simulates the dynamics of subsystems, and deterministically gives an unbiased estimation of the joint time evolution operator without any additional assumptions.
It does not rely on any self-consistent (mean-field) condition of embedding methods or variational optimization of parameterized ans\"atze of entanglement forging.

\section{Numerical Simulation}
\label{appendix:numerical}

In this section, we explore the concrete applications of our perturbative approach in simulating the dynamics of quantum many-body physics problems with operations on a small quantum computer or quantum simulator. We focus on the algorithm with the explicit decomposition, and we numerically test our algorithm in simulating several interacting physics with different topologies as examples. 
Fig.~\ref{fig:topo_SM} illustrates four different topological structures and the explicit partitioning strategies considered in this work.
In the following subsections, we show how to simulate different dynamics of the interacting systems using maximally $8+1$ qubits.

\begin{figure}[ht!]
\includegraphics[width =0.78\textwidth]{ 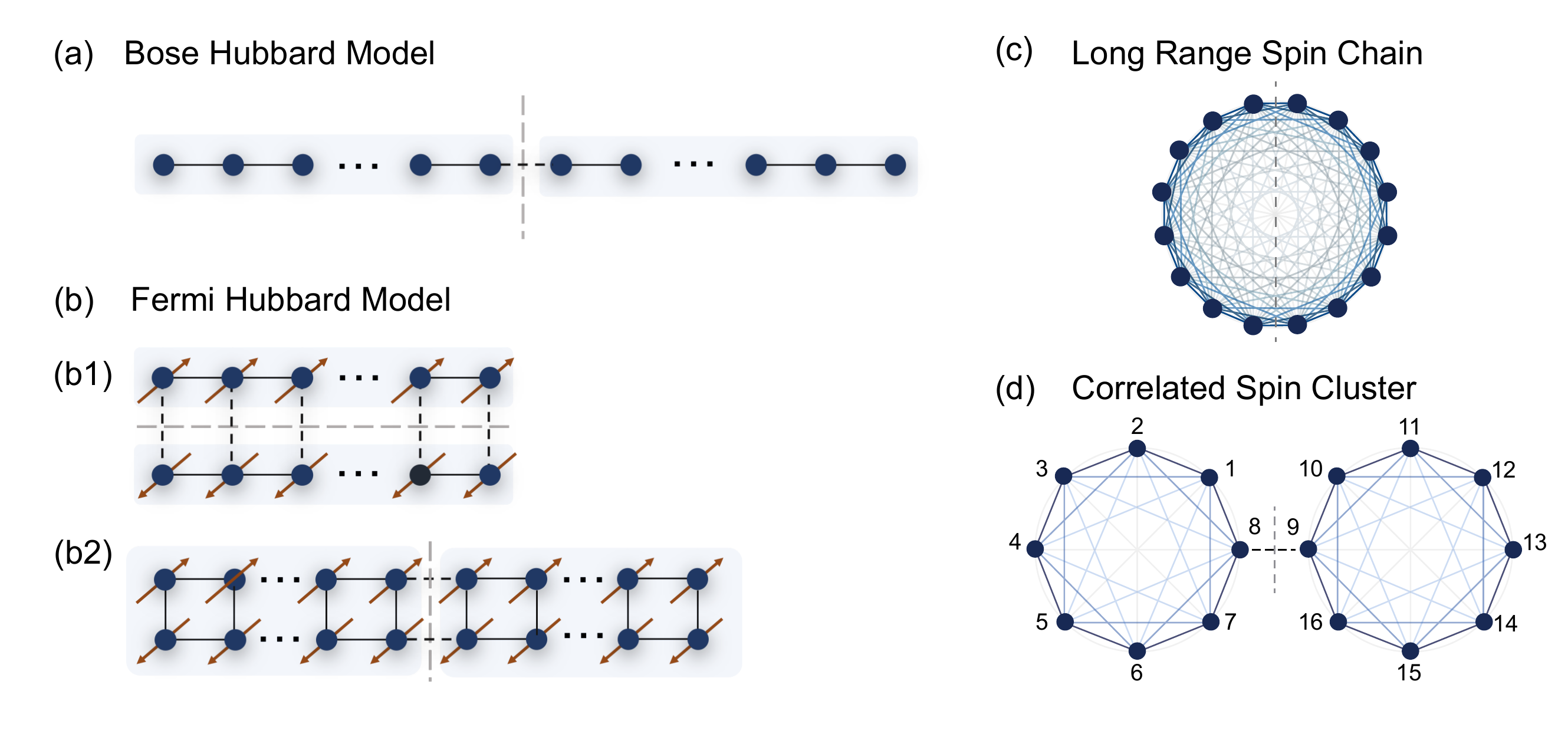}
 \caption{
Four different topological geometries and the partitioning strategies corresponding to Bose Hubbard model, Fermi Hubbard model, long-range spin chain and correlated spin cluster considered in this work.
}
 \label{fig:topo_SM}
\end{figure}

\subsection{Interacting bosons}
\label{appendix:qw}
We consider the physics of interacting spinless bosons on a lattice~\cite{Cazalilla.RevModPhys.83.1405}, which could be described by the extended Bose-Hubbard Hamiltonian
\begin{equation}
H=-\sum _{\left\langle i,j\right\rangle }t_{ij}{\hat {b}}_{i}^{\dagger }{\hat {b}}_{j}+{\frac {U}{2}}\sum _{i}{\hat {n}}_{i}\left({\hat {n}}_{i}-1\right)+  \sum _{i}h_i {\hat {n}}_{i},
\end{equation}
where 
% $\left\langle i,j\right\rangle$ denotes summation over the lattice sites,
$\hat {b}_{i}$ and $\hat {b}_{i}^{\dagger }$ are the bosonic annihilation and creation operators, ${\hat {n}}_{i}={\hat {b}}_{i}^{\dagger }{\hat {b}}_{i}$ gives the number of particles on the site $i$, $t_{ij}$ describes the hopping strength,  $U$ describes the on-site interaction, and $h_i$ is the on-site chemical potential that can be tuned in various quantum systems. 
The model reduces to the Bose-Hubbard model $H_{\rm BHM}$ when only nearest-neighbour hopping is allowed, i.e., $t_{ij}=\delta_{|i-j|,1}t$.
While the Bose-Hubbard model is not exactly solvable for finite values of $U$ and $t$, in the large $U$ limit $U/t \rightarrow + \infty$, this model reduces to the Tomonaga-Luttinger gas Hamiltonian, which describes the collective behaviour of hard-core bosons~\cite{Cazalilla.RevModPhys.83.1405}.
Using the Holstein and Primakoff transformation, the Bose-Hubbard model is mapped onto the XX spin chain model
\begin{equation}
H_{\rm BHM}=J\sum_{j}   \left(\hat{\sigma}_{j}^{x} \hat{\sigma}_{j+1}^{x}+\hat{\sigma}_{j}^{y} \hat{\sigma}_{j+1}^{y}\right) +\frac{1}{2} \sum_{j} {h_i} \left(\hat{\mathcal{I}}_{j} - \hat{\sigma}_{j}^{z} \right) 
\label{eq:BHM_SM}
\end{equation}
with $\hat{\sigma}_{j} $ representing the Pauli operator on the $j$th site and the effective interaction $J=-2t$.
The hard-core bosons  can also be related to the one-dimensional free spinless fermions using the Jordan-Wigner transformation. 

% Yan \textit{et al} experimentally demonstrated quantum walks of strongly correlated microwave photons with a 12-qubit superconducting processor \cite{yan2019strongly}. 
 
The quantum walks of the $1$D translationally invariant (i.e., the hopping strengths are the same $t_{ij} = t$) bosons were experimentally demonstrated in Ref.~\cite{yan2019strongly}. Their device system, a 12-qubit superconducting processor, can be well described by the hard-core boson Hamiltonian in  \autoref{eq:BHM_SM}. 
In our numerical simulation, we break the translational invariance and investigate the collective excitations including the density distribution and correlations of bosons with several reduced interaction strength. 
We consider two clusters of the interacting bosons with tunable hopping strength $t_{ij} = t'$ on the boundary of subsystems.
The Hamiltonian can be expressed as $H=H_1+H_2 + V^{\rm int}$ with the local Hamiltonian and interactions on the boundary as
\begin{equation}\label{eq:hamil_qw_SM}
\begin{aligned}
    H_l^{\rm loc} &=J\sum_{j}   \left(\hat{\sigma}_{j}^{x} \hat{\sigma}_{j+1}^{x}+\hat{\sigma}_{j}^{y} \hat{\sigma}_{j+1}^{y}\right) + \frac{1}{2}  \sum_{j} {h_i} \left(\hat{\mathcal{I}}_{j} - \hat{\sigma}_{j}^{z} \right),\\
    V^{\rm int} &= J^{\prime} (\sigma_{1,N}^x\sigma_{2,1}^x+\sigma_{1,N}^y\sigma_{2,1}^y).
\end{aligned}
\end{equation}
 Here, $\sigma_{l,i}$ represents Pauli operators acting on the $i$th site of $l$th subsystem, and the interactions at the boundary is $J^{\prime} = -2t'$.
 Note that this reduces to the Bose-Hubbard model when $t = t'$.

\begin{figure}
\includegraphics[width =1.0\textwidth]{ 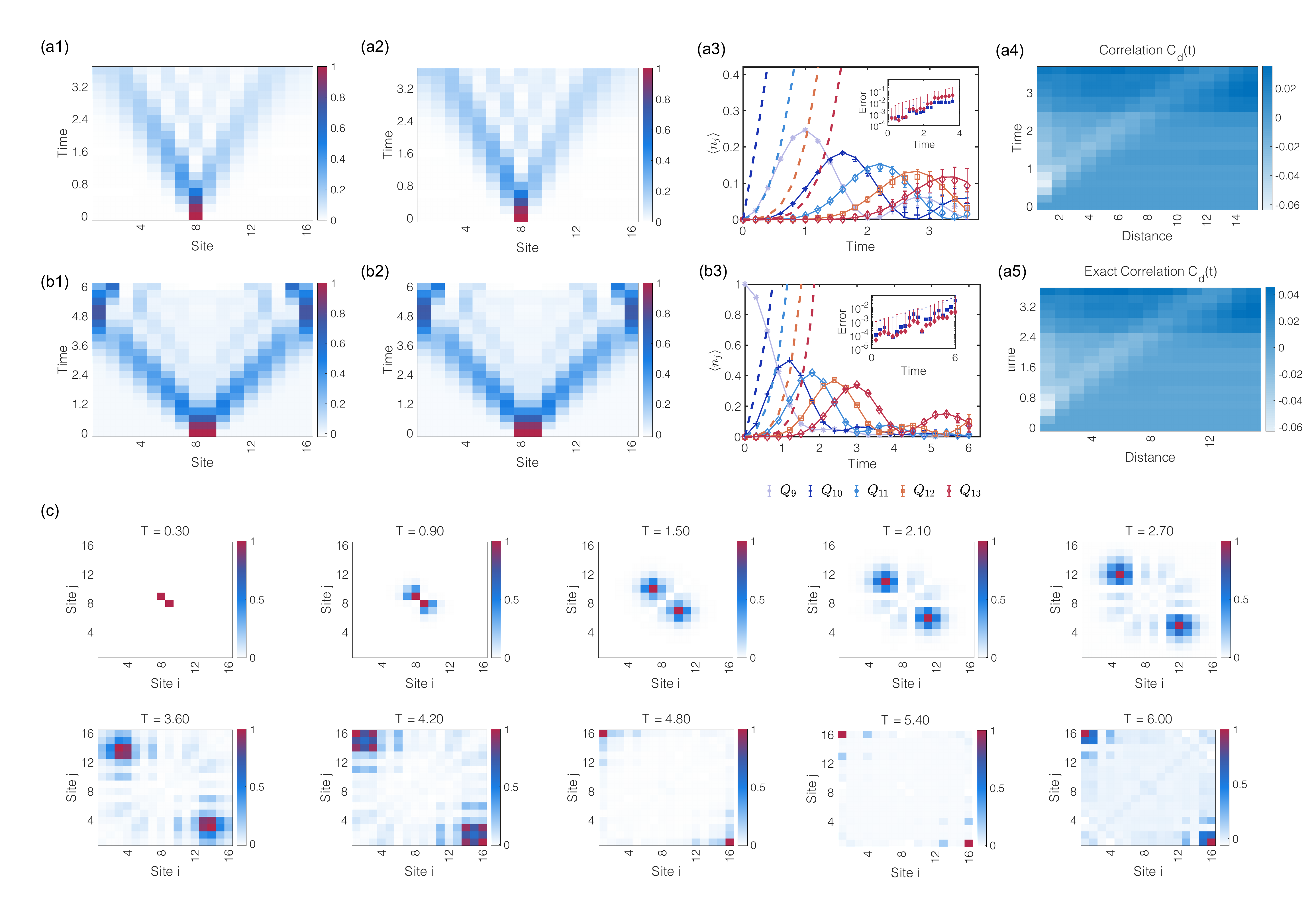}
 \caption{
Dynamics of $16$ interacting bosons on a $1$D array in the large onsite repulsion
limit $U/t \rightarrow \infty$.
(a) Quantum walk after single particle excitation at the centre $\ket{\psi_0} = \hat{b}_8^{\dagger} \ket{\mathbf{0}}$. We set the interaction strength as $J = 0.5$ and $J^{\prime}=0.8J$.
(a1) and (a2) show the simulation results and exact results for time-evolved density evolution  $\hat n_j = \langle \hat b_j^{\dagger} \hat b_j \rangle$, respectively.
(a3) The density distribution $\hat n_j$ at different sites $Q_9 $ to $Q_{13}$ under time evolution. The nearest-neighbour Lieb-Robinson bounds (dashed line) capture the maximum propagation speed of density spreading.
(a4) and (a5) show the evolution of the averaged two-body correlation functions $\bar C_{d}(t) = \frac{1}{N-d} \sum_{j=1}^{N-d} C_{j,j+d}(t)$, which exhibit similar light cone propagation. 
 The inset figure in (a3) shows the errors for the density and the averaged two-body correlation functions. 
(b) Quantum walk after two particle excitations at the centre $\ket{\psi_0} = \hat{b}_8^{\dagger} \hat{b}_9^{\dagger}  \ket{\mathbf{0}}$. We set the interaction strength as $J = 0.5$ and $J^{\prime}=0.5J$.
(b1) and (b2) show the simulated results and exact results for time-evolved density evolution  $\hat n_j = \langle \hat b_j^{\dagger} \hat b_j \rangle$, respectively.
(b3) The density distribution $\hat n_j$ at different sites $Q_9 $ to $Q_{13}$ under time evolution. The nearest-neighbour Lieb-Robinson bounds are shown by the dashed line. The inset figure in (b3) shows the errors for the density and the two-body density-density correlation functions $\hat\rho_{ij}=\braket{\hat b_i^{\dagger } \hat b_j^{\dagger } \hat b_i \hat b_j}$. 
 (c) Spatial antibunching and fermionisation in the quantum walk of two indistinguishable bosons. The two bosons are excited at the centre. The normalised density-density correlation functions $\hat \rho_{ij}/\hat \rho_{ij}^{\max}$ at several time $T$. The off-diagonal correlations appear under evolution, which shows the antibunching and fermionisation of  strongly correlated bosons. This phenomenon is well captured by the non-interacting spinless fermions. In this numerical simulation, we set the sampling number as $5 \times 10^5$.
}
 \label{fig:qw_SM}
\end{figure}

Now, we divide the whole system into two parts and simulate the  dynamics of interacting bosons using our perturbative algorithm with the explicit decomposition. Our method thus enables the simulation of the $16$-qubit problem with only $8+1$ qubits. 
It is worth noting that the explicit decomposition is optimal with respect to all possible decomposition strategies, as proven in Theorem~\ref{thm:min_cost}. 
We first demonstrate the dynamics after local perturbation under the interacting Hamiltonian.
Previous work has extensively studied the propagation speed of quantum information in quantum many-body systems with finite range couplings, which is limited by a maximal group speed, known as the Lieb-Robinson velocity  $v_g$~\cite{bravyi2006lieb,jurcevic2014quasiparticle}. Information propagates faster than $v_g$ is exponentially suppressed, which exhibits a light-cone-like information propagation analogous to the relativistic theory.   
One can consider a local perturbation to the initial state $\ket{\psi_0}$ as $\ket{\psi(t=0)} = O_A\ket{\psi_0}$ in the region $A$.
As proven in Ref.~\cite{bravyi2006lieb}, the change of the expectation of the observable $O_B$ in the region $B$ under time evolution can be bounded by
\begin{equation}
\left|\langle\psi(t)|  O_B |\psi(t)\rangle-\langle\psi_0|O_B| \psi_0\rangle\right|  =|\langle\psi|O_A^{\dagger}[O_B(t),  O_A]| \psi\rangle| \leq \|[O_B(t), O_A] \|,
\end{equation}
where $O_A(t)$ represents the operator in the Heisenberg picture.
This establishes how local operations $O_A$ affect the observables $O_B$  under time evolution. 
If the interactions decrease exponentially with distance, one can bound the unequal time commutator by 
\begin{equation}
    \| |[O_B(t), O_A] \| \leq  C\|O_A\|\|O_B\| \exp\bigg[-\frac{d-v_g|t|}{\xi}\bigg],
\label{eq:Lieb_bound}
\end{equation}
where $d$ is the distance of between the region $A$ and $B$ (shortest path connecting $A$ and $B$), and $c$, $v_g$, and $\xi$ are positive constants depending on $g = \max_{i,j}|J_{ij}|$. 
For the nearest-neighbour interaction, one can have a tighter bound by $\left|\langle\psi(t)|  O_B |\psi(t)\rangle-\langle\psi_0|O_B| \psi_0\rangle\right| \leq I_d(4Jt)$, where $d$ is the distance of between the site $A$ and $B$, $c$ and $v$ are the velocity constant, and $I_d$ is the modified Bessel function of the first kind~\cite{jurcevic2014quasiparticle}.
In our simulation, the particle number is conserved, and we consider the observable as the occupation number operator $O_A = \hat{n}_j$ and local perturbation as $O_B = \prod_{j\in B} \hat{\sigma}_j^x \ket{\psi_0}$.

Now, we study the propagation of density distribution  and non-local two-body correlations after local excitations.
We first excite one boson at the centre by $\ket{\psi_0} = \hat{b}_8^{\dagger} \ket{\mathbf{0}}$, where $\ket{\mathbf{0}}$ is the vacuum, and then study the density spreading of the boson under the interacting Hamiltonian with interaction strength $J = 0.5$ and $J^{\prime} = 0.8J$. As shown in Fig.~\ref{fig:qw_SM}(a1), the evolution of density $\hat n_j = \langle \hat b_j^{\dagger} \hat b_j \rangle$ indicates a light-cone-like propagation. The propagation is well captured by the nearest-neighbour Lieb-Robinson bound (dashed line), as shown in Fig.~\ref{fig:qw_SM}(a3). 
Then, we study the distribution of correlations after the single-particle excitation. 
We consider the averaged  non-local correlations as
\begin{equation}
\bar C_d (t) =  \frac{1}{N-d} \sum_{j=1}^{N-d} C_{j,j+d}(t)
\end{equation}
with the two-body correlation function
$C_{ij}(t)=\braket{\sigma_i^z \sigma_j^z}-\braket{\sigma_i^z }\braket{\sigma_j^z }$.
We see the correlation grows nonlocally under evolution, and it also exhibits a clear light cone propagation, as shown in Fig.~\ref{fig:qw_SM}(a4). The exact dynamics are shown in Fig.~\ref{fig:qw_SM}(a2,~a5) for comparison.

Next, we show the results for the strong correlation effects with two bosons excitations.
The two adjacent bosons display spatial bunching effects in the non-interacting case while it gradually transform to spatial antibunching in the large $U$ case, which is similar to non-interacting spinless fermions, theoretically investigated in~\cite{lahini2012quantum}.
The fermionisation phenomenon of the $1$D translationally invariant bosons in the large $U$ limit was also experimentally demonstrated in Ref.~\cite{yan2019strongly}. 
Here, we consider the correlated Hamiltonian in  \autoref{eq:hamil_qw_SM} with reduced interaction strength $J^{\prime} = 0.5J$ on the boundary. At $t=0$, we excite two adjacent indistinguishable particles at the centre $\ket{\psi_0} = \hat{b}_8^{\dagger} \hat{b}_9^{\dagger}\ket{\mathbf{0}}$. 
We first show the density spreading in Fig.~\ref{fig:qw_SM}(b1,b3), which exhibits similar propagation as the single particle excitation case.
The dynamics of two particle excitation can be sensitive to the particle statistics due to interference.
As proposed in Ref.~\cite{lahini2012quantum}, the fermionisation or bosonisation of the particle statistics could be distinguished by measuring the two-body density-density correlators
\begin{equation}
\hat\rho_{ij}=\braket{\hat b_i^{\dagger } \hat b_j^{\dagger } \hat b_i \hat b_j}.
\end{equation}
In Fig.~\ref{fig:qw_SM} (c), we show the time evolution of the density operator $\hat n_j$ and density-density correlators of two bosons placed at the adjacent centre. 
The long-range anticorrelations appearing in the off-diagonal pattern reveal the fermionisation of  strongly correlated bosons with reduced interaction strength. We can also see the interference pattern in Fig.~\ref{fig:qw_SM} (c) during the evolution as an indication of interactions between the bosons.

\subsection{Interacting fermions}
\label{appendix:fh}
 
In this section, we consider the one-dimensional interacting fermions with spin degrees of freedom, which is described by the Fermi-Hubbard Hamiltonian as
\begin{equation}
\begin{aligned} 
H= -J \sum_{j, \sigma }  \left( \hat{c}_{j, \sigma}^{\dagger}  \hat{c}_{j+1, \sigma}+\mathrm{h.c.}  \right)  +U \sum_{j} \hat n_{j, \uparrow} \hat n_{j, \downarrow}+\sum_{j,\sigma }   h_{j, \sigma} \hat n_{j, \sigma} \end{aligned}
\end{equation}
where $\hat c_j$ ($\hat c_j^{\dagger}$) is the fermionic annihilation (creation) operators on the $j$th site and spin state $\sigma \in\{\uparrow, \downarrow \}$,
and $\hat n_j = \hat c_j^{\dagger}\hat c_j $ is the particle density operator.
One-dimensional interacting fermions can be well captured by the  Luttinger liquid theory, which shows that the spin and charge of the electrons disintegrate into two separate collective excitations, spinon (holon) excitations with only spin (charge) degrees of freedom. 
For self-consistency, we briefly review the theory of bosonisation and discuss the separation of spin and charge excitations, following the discussion in~\cite{fradkin2013field,arute2020observation}.

The Fermi surface of interacting electrons in $1$D only has two points, and therefore it could be reduced to the effective Hamiltonian describing the excitation from one point to the other.
The effective Hamiltonian ignoring spins can be expressed as $H = H^{\rm loc} +V_{ee}$ where $\hat{H}_{0} = \sum_{\zeta = \pm 1} \sum_{q} v_{\mathrm{F}} q  \hat c_{\zeta q}^{\dagger}  \hat c_{\zeta q}$ and $V_{ee}  =\frac{1}{2 L} \sum_{k k^{\prime} q} V_{\mathrm{ee}}(q) \hat \hat c_{k-q}^{\dagger} \hat c_{k^{\prime}+q}^{\dagger}\hat c_{k^{\prime}} \hat c_{k}$, which describes the allowed scattering near the Fermi surface. Here, $\zeta = \pm 1$ represents the left or the right side of the Fermi surface, and $v_F$ is the Fermi velocity.
For one-dimensional electrons, the density modulation is the elementary excitation, and thus it is natural to introduce the bosonic operator as
\begin{equation} 
\hat b_{\zeta q}^{\dagger}  =\sqrt{\frac{2 \pi}{L q}} \sum_{k} \hat c_{\zeta, k+q }^{\dagger} \hat c_{\zeta, k}  \end{equation}
to map the interacting fermions to the free bosons, where $L$ is normalisation constant.
The creation and annihilation operations of bosons satisfy the commutation relation as
\begin{equation}
\left[\hat b_{\zeta q }, \hat b_{\zeta^{\prime} q^{\prime} }^{\dagger} \right]=\delta_{\zeta \zeta^{\prime}} \delta_{q q^{\prime} }
\end{equation}

% \begin{equation}
% H=\sum_{q>0, \zeta=\pm 1} v_{F} q b_{\zeta q}^{\dagger} b_{\zeta q}+\frac{1}{2 L} \sum_{q>0, \zeta=\pm 1} g_{2} \rho_{\zeta}(q) \rho_{-\zeta}(-q)+\frac{1}{2 L} \sum_{q>0, \zeta=\pm 1} g_{4} \rho_{\zeta}(q) \rho_{\zeta}(-q)
% \end{equation} which 
 
Therefore, the full interacting Hamiltonian can be mapped to the non-interacting Hamiltonian in terms of the bosonic operators as 
\begin{equation}
H=\sum_{q>0, \sigma, \sigma^{\prime} \zeta=\pm 1}\left[v_{F} q \delta_{\sigma \sigma^{\prime}} \hat b_{\zeta q \sigma}^{\dagger} \hat b_{\zeta q \sigma}+\frac{q g_{2}}{4 \pi}\left(\hat b_{\zeta q \sigma}^{\dagger}  \hat b_{-\zeta q \sigma^{\prime}}^{\dagger}+ \hat b_{\zeta q \sigma} \hat b_{-\zeta q \sigma^{\prime}}\right)+\frac{q g_{4}}{2 \pi} \hat b_{\zeta q \sigma}^{\dagger} \hat b_{\zeta q \sigma^{\prime}}\right],
\end{equation}
where $g_2$ and $g_4$ measure the strength of the interaction in the vicinity of the Fermi points as conventionally used in the literature, and $\sigma$ denotes the spin degrees of freedom.
We can write the above Hamiltonian into the bosnic operator of charges and spins $\hat b_{\zeta q c}^{\dagger}=\frac{1}{\sqrt{2}}\left(\hat b_{\zeta q \uparrow}^{\dagger}+ \hat b_{\zeta q \downarrow}^{\dagger}\right)$ and $\hat b_{\zeta q s}^{\dagger}=\frac{1}{\sqrt{2}}\left(\hat b_{\zeta q \uparrow}^{\dagger}- \hat b_{\zeta q \downarrow}^{\dagger}\right)$ with $c$ ($s$) denoting charge (spin),  and then diagonalize the Hamiltonian by a Bogoliubov transformation as
\begin{equation}
H=\sum_{q>0, \zeta=\pm 1}\left[v_c \left(\alpha_{\zeta q c}^{\dagger} \alpha_{\zeta q c}+\frac{1}{2}\right)+ v_s \alpha_{\zeta q s}^{\dagger} \alpha_{\zeta q s}\right].
\end{equation}
with the velocities $v_c = q \sqrt{\left(v_{F}+\frac{g_{4}}{2 \pi}\right)^{2}-\left(\frac{g_{2}}{2 \pi}\right)^{2}}$ and $v_s = q v_{F}$.
This clearly shows that the spin and charge density has different velocity near the Fermi surface, as predicted by the theory of Luttinger liquid.
This observation has been theoretically and numerically investigated~\cite{vsmakov2007binding,arute2020observation,Cazalilla.RevModPhys.83.1405}. Arute \textit{et al.} initially reported the experimental simulation using a programmable superconducting quantum processor with high gate accuracy~\cite{arute2020observation}. 

To simulate the dynamics of the interacting fermions carrying spins on a quantum computer, we can use the Jordan-Wigner transformation to map the fermionic operators $\hat c_j$ on each site to the qubit Pauli operators as 
\begin{equation}
    \hat c_{j} \mapsto \frac{1}{2}\left(\hat{\sigma}^x_{j}+i \hat{\sigma}^y_{j}\right) \bigotimes_{i=1}^{j-1} \hat{\sigma}^z_{i}.
\end{equation}
with Pauli operators $\hat{\sigma}_j^{\alpha}$, $\alpha = (x,y,z)$ acting on the $j$th site. 
It is straightforward to have 
$\hat n_{j} \mapsto \frac{1}{2}\left(\hat{\mathcal{I}}-\hat{\sigma}^z_{j}\right)$ and  $\hat n_{j} \hat n_{k}  \mapsto \frac{1}{4}\left( \hat{\mathcal{I}}+\hat{\sigma}^z_{j} \hat{\sigma}^z_{k}-\hat{\sigma}^z_{j}-\hat{\sigma}^z_{k}\right)$.
The experimental setting after the qubit mapping has the $2$D topology. 
% The interactions for this $1$D interacting ferimions have two parts: (1) kinetic terms due to nearest hopping $t$ (2) on-site spin interaction $U$.
We consider an $8$-site interacting $1$D Fermi-Hubbard model, which requires $N = 16$ qubits to encode the spin up and spin down at each site. 
According to the topology of interactions, we have two partitioning strategies,  by regarding either the nearest hopping or on-site Coulomb interactions as the $V^{\rm int}$. Therefore, depending on the relative strength of $t$ and $U$, we can cut the full interacting systems along either longitudinal or transverse directions. 
We discuss how to implement our algorithms for this topology setting by using two partitioning strategies in Fig.~\ref{fig:fh2D_SM} in the following. We will then show how to apply our perturbative quantum simulation method to use $8+1$ qubits to simulate the dynamics of the $16$ qubit system. 

% Now, we demonstrate this observation of spin and charge separation of a $8$-site interacting $1$D Fermi-Hubbard model using our algorithm. 
We first prepare the initial state as the ground state of a non-interacting Hamiltonian. 
In the non-interacting limit, Hamiltonian commutes with the total number operators $[H, \sum_j \hat n_{j,\sigma}] = 0$. For a one-dimensional chain, one finds that the Hamiltonian in one-particle sector moves the occupied site to the left or right, and thus can be expanded on the one-particle basis as a tridiagonal matrix.
The interacting Hamiltonian  has the elements $H_{ij} = \braket{i|H|j}$ and $\ket{j} = \hat{c}_j^{\dagger} \ket{\mathbf{0}}$ with $\ket{\mathbf{0}}$ representing the vacuum. We can use unitary transformation $U  \equiv [u_{ij}]_{ij}$ to diagnalise the interacting Hamiltonian and get the eigenstates and eigenenergies in terms of the non-interacting fermionic operators $\hat a_j$ and $\hat{a}_j^{\dagger}$,  which we refer as the rotated basis. The rotated basis is related to the original basis by the unitary transformation as 
\begin{equation}
    \hat{a}_j^{\dagger}  = \sum_j u_{ij} \hat c_j^{\dagger}.
\end{equation}
In the two-particle sector, there are $\tbinom{N}{2}$ occupation number basis states, and we can similarly diagonalise the matrix of Hamiltonian to obtain the eigenstates and eigenenergies. 

For the ground state with general occupied numbers $N_f$ (relatively small $N_f$), we can get the subspace of $N_f$ particle numbers and use the transformation from the original basis $\hat{c}_j^{\dagger}$ to the rotated basis $\hat{a}_j^{\dagger}$.
Refs.~\cite{PhysRevApplied.9.044036,KivLinearDepth} discussed the algorithm that we can use linear depth circuit to prepare the ground state of a non-interacting Hamiltonian. We briefly review the procedure to prepare the initial state using the linear depth circuit.
Denote the operators in the rotated basis that diagonalize the non-interacting Hamiltonian as $\hat a$ and $\hat a^{\dagger}$. 
We can apply a particle-conserving rotation ${U}$ of the single particle basis to the rotated basis as ${U}\hat c_j^{\dagger}{U}^{\dagger} = \hat a_j^{\dagger}$.
Then we obtain the ground state of the non-interacting Hamiltonian from the easy-to-prepare state as
\begin{equation}
    \ket{\phi} = {U} \hat c_1^{\dagger} \cdots \hat c_{N_f}^{\dagger}\ket{\textbf{0}},
\end{equation}
where $\ket{\textbf{0}}$ is the vacuum.
The two bases are related by a unitary transformation that transforms the original operators $\hat c$ ($\hat c^{\dagger}$) of the interacting Hamiltonian to the new operators $\hat a$ ($\hat a^{\dagger}$) of non-interacting Hamiltonian  $\hat a_i^{\dagger} = \sum_j u_{ij} \hat c_j^{\dagger}$ where $u$ is a $N\times N$ matrix. The basis change unitary is given by $U(u) = \exp \left(\sum_{ij} [\log u]_{ij} (\hat c_i^{\dagger} \hat c_j - \hat c_j^{\dagger} \hat c_i) \right)$ which can be implemented by $\mathcal{O}(N)$ depth circuits using Given rotations in parallel~\cite{KivLinearDepth}.
In the numerical simulation, we set the hopping strength $J = 0.5$, and the on-site interaction $U = 0.5J$ or $U = J$. We set the local potential for spin up in a Gaussian distribution $h_{j, \uparrow} = -\lambda_{\uparrow} \exp\left(-\frac{(j-(L+1)/{2})^2}{2\nu^2} \right)$ with $L = 8$, $\lambda_{\uparrow} = 4$ and $\nu = 1$ while $h_{j, \downarrow} = 0$ for spin down, the same as in Ref.~\cite{arute2020observation} for comparison.
The state is initialised with quarter filling $N_{\uparrow} = N_{\uparrow} = 2$, in which the charge and spin density are generated in the middle of the chain at $t = 0$ in Fig.~\ref{fig:fh2D_SM}.

% The basic idea is that an arbitrary state with general occupied numbers $N_f$ can be prepared by applying a particle-conserving rotation of the single particle basis to the original basis  

\begin{figure}
\includegraphics[width =1.0\textwidth]{ 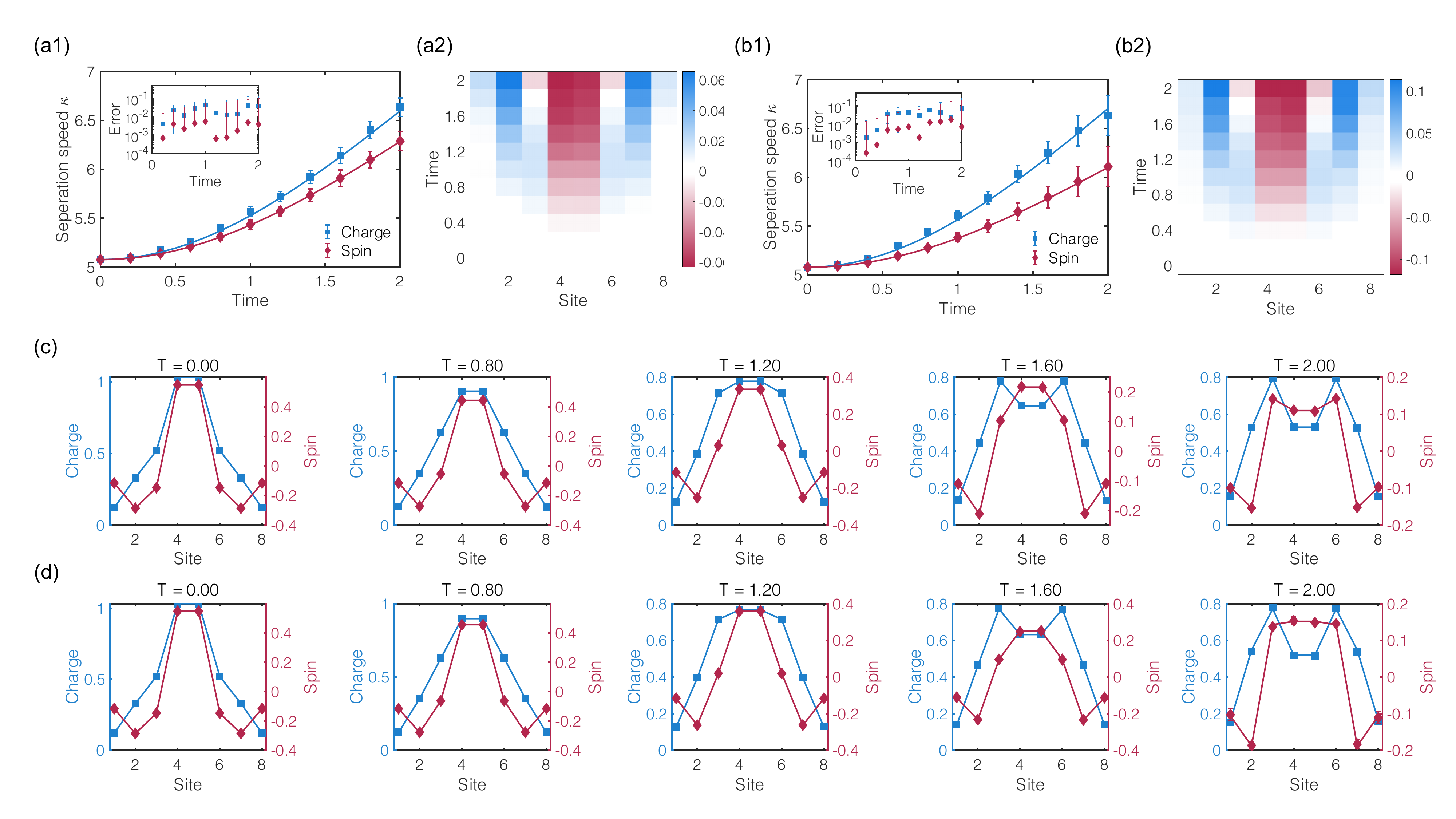}
 \caption{Separation of charge and spin density under the one-dimensional Fermi Hubbard model.
 The quantum state is initialised as the ground state of the non-interacting Hamiltonian with local potential $h_{j,\sigma}$, as specified in \autoref{appendix:fh}. We consider the dynamics with two-particle excitations, i.e. $N_{\uparrow} = N_{\uparrow} = 2$,  which are generated at the middle of the chain at $t = 0$.
 (a1) and (a3). The time evolution of separation speed $\kappa$ for charge (square and blue) and spin (diamond and red) with the interaction $U = J/2$ (a1) and $U = J$ (a3), respectively.
Solid lines represent the exact results for comparison. The figure inset shows the errors of $\kappa$ compared with the exact results over time. 
 (a2) and (a4). The difference of charge and spin densities $\rho_j^{c}(t) - \rho_j^{s}(t) - \textrm{const}$ at each site with the interaction $U = J/2$ (a2) and $U = J$ (a4), respectively. The separation are offset by a constant as $0$ at $t=0$, i.e. $\textrm{const} = \rho_j^{c}(0) - \rho_j^{s}(0)$.
 (c) and (d) show the time evolution density spreading of both charge and spin at different $T$ for $U = J$ to $T = 2.0$. We set the sampling number as $5 \times 10^5$.
 }
 \label{fig:fh2D_SM}
\end{figure}

Next, we evolve the two-particle system under the Fermi-Hubbard Hamiltonian with different strengths of on-site interaction $U$. 
The charge and spin densities characterise the collective excitations, which are defined as the sum and difference of the spin-up and -down particle densities over all sites, respectively,
\begin{equation}
\rho_{j}^{\eta}=\left\langle \hat n_{j, \uparrow}\right\rangle \pm\left\langle  \hat n_{j, \downarrow}\right\rangle
\end{equation}
where $\eta = c$ or $s$ represents charge or spin degrees of freedom.
We show the density spreading of both charge and spin in Fig.~\ref{fig:fh2D_SM} (c) and (d) at different $t$. We plot the difference of charge and spin density in Fig.~\ref{fig:fh2D_SM} (a2) and (a4) for $U = J/2$ and $U = J$, respectively. Here, the separation of charge density and spin density is offset as $0$  at $t=0$ to make the difference comparable. 

The excitations spreading from the middle can be quantitatively distinguished by introducing the separation speed
\begin{equation}
    \kappa = \sum_{j=1}^L \left| j - (L+1)/2\right| \rho_{j}^{\eta}.
\end{equation}
Under time evolution, we observe a clear separation of spin density and charge density as shown in Fig.~\ref{fig:fh2D_SM}(a1) and (a3).  As $U$ increases to $U = J$, the separation of spin density and charge density becomes much faster. 
The error for the separation speed $\kappa_\eta$ ($\eta = c/s$) are shown in the figure inset.
In the large interaction regime, the initial state considered here is a mixture of excited states, and therefore the effective physics can not be well captured by the Luttinger liquid theory~\cite{fradkin2013field}.

While this effective model can only capture the low-energy excitation in the weakly coupled regime, our method can simulate the dynamics in the highly excited regime with medium or large interaction.
The interactions for this $1$D interacting fermions have two parts: 
(1) kinetic terms due to nearest hopping $t$ 
(2) on-site spin interaction $U$.
According to the topology of interactions, we have two strategies,  by regarding either the nearest hopping or on-site spin interactions as the $V^{\rm int}$. Therefore, depending on the relative strength of $t$ and $U$, we can cut the full interacting systems along either longitudinal or transverse directions. This enables the simulation in both regimes. We note that to prepare the general entangled state, we can decompose it into a linear combination of local states, which might introduce an additional sampling cost for the state preparation. 

We remark that this partitioning strategy enables the quantum simulation for the two opposite regimes, which aligns with the view from the perturbation theory which applies to the weakly-interacting and strongly-interacting limit. 
Our method could be used to simulate the dynamics of interacting phenomena with quasi-$1$D or $2$D geometry.
In the case of the Fermi-Hubbard model considered above, the explicit decompositions for both {geometries} are optimal with respect to the resource cost for the simulation of non-local correlations.

\subsection{Quantum spin systems}
In this section, we consider to apply our perturbative approach with the explicit decomposition to simulate several emergent quantum phenomena in quantum spin systems.

\subsubsection{Dynamical quantum phase transitions}

% Dynamical paramagnetic phase and the dynamical ferromagnetic phase of LMG model are investigated. 

% Quantum spin models have been proposed to capture some typical emergent quantum phenomena in the condensed matter, such as phase transitions and collective transitions. 
Quantum spin models have been investigated to capture or predict some typical emergent quantum phenomena in the condensed matter, such as phase transitions and collective transitions. 
While many theoretical and numerical methods have been proposed to solve the effective spin models in exact or approximate solutions, a long-range spin chain with general interaction strength could be hard to solve classically. 
In this section, We consider a long-range  spin chain, which is described by
\begin{equation}
    H = \sum_{ij} J_{ij} \hat \sigma_i^z \hat \sigma_j^z + h\sum_j \hat \sigma_j^x 
    \label{eq:hamil_spin_power}
\end{equation}
with the interactions obeying the power law decay rule
 $J_{ij}={|i-j|^{-\alpha}}$. We study the dynamical quantum phase transitions (DQPT) in the long-range spin chains.
 
We show the dynamical criticality of many-qubit spin chains with fully connected topology using the local order parameters and the Loschmidt amplitude.
The state is first initialised as the eigenstate $\ket{1}^{\otimes n}$ of the non-interacting Hamiltonian with $h = 0$, and the system is quenched by suddenly adding the transverse field $h$ along $x$ direction. In the limit of $\alpha = 0$, the Hamiltonian in  \autoref{eq:hamil_spin_power} reduces to the Lipkin-Meshkov-Glick (LMG) model. LMG model has the analytical solution as it can be regarded as a classical model, of which the dynamical behaviour can be predicted by the semiclassical limit. 
The Hamiltonian $H$ preserves the magnitude of the total spin and has the spin flip symmetry, i.e., $[H, \mathbf{S}^2] = 0$ and $[H, \prod_i \hat{\sigma}_i^x] = 0$.
We can write the Hamiltonian as $H = \frac{J}{N} (\Sigma^z)^2 +h \Sigma_x  $ using collective spin operators $\Sigma^{\alpha} = \sum_i \sigma_i^{\alpha}$ with ${\alpha}=x,y,z$.
We can use the mean-field approach to represent the spin as a classical spin vector $\left(\Sigma^{x}, \Sigma^{y}, \Sigma^{z}\right)=N(\cos \theta, \sin \theta \sin \phi, \sin \theta \cos \phi)$ that can be determined by the equation of motion. 
In Ref.~\cite{zhang2017observation}, the authors considered the spin Hamiltonian with the external field along the $z$ direction and showed analytically that the spatially averaged two-point correlation shows a DQPT when $B_z/J_0$ crosses unity.  One can similarly use the analytical method to analyse the dynamical behaviour of Eq.~\ref{eq:hamil_spin_power} with small $\alpha$ near to zero.

% The results obtain from the PQS method  in weakly-coupled regime can be compared.

Refs.~\cite{zhang2017observation,xu2020probing} experimentally demonstrated the DQPT and various dynamical results for the long-range spin model with $\alpha$ close to zero with a trapped ion platform~\cite{zhang2017observation} and a superconducting processor~\cite{xu2020probing}.
Here, we focus on the weakly coupled regime, i.e. large $\alpha$ for comparison. 
In the numerical simulation, we set $J_0 = 1$, and the decay rate $\alpha = 3$.
We partition the full systems into $2$ or $3$ subsystems with each subsystem consisting of at most $8$ qubits.
We use the explicit decomposition to simulate the large system. Note that  the explicit decomposition for this example might be not optimal as it involves too many Pauli terms at each site. Other decomposition methods within the framework of generalised quantum operations could be numerically searched to obtain the minimal resource cost.
We first show the evolution of order parameters of $16$-site quantum spin chain.
In Fig.~\ref{fig:DPP_SM} (c), we show the magnetisation $M_z(t)$ and $M_x(t)$ rapidly oscillate across $0$ when the external field is large,  while the magnetisation decays slowly in the low field. The motion of spin can be illustrated in a Bloch sphere in Fig.~\ref{fig:DPP_SM}. 
These order parameters provide an evidence for two phases: ferromagentic phase and paramagnetic phase.

The dynamical quantum phase transitions in the out-of-equilibrium phase could be observed using the Loschmidt amplitude as
\begin{equation}
\mathcal{G}(t) = \left| \braket{\psi_0|e^{-iHt} |\psi_0} \right|^2
\end{equation}
as an indicator to characterise the dynamical echo back to the initial state in the out-of-equilibrium phases. A DQPT occurs with a nonanlytical behaviour of rate function 
\begin{equation}
    \gamma(t) = -N^{-1}\log\left(\mathcal{G}(t)\right)
\end{equation}
in the thermodynamic limit $N \rightarrow \infty$, which can be regarded as a dynamic counterpart to a free energy density up to a normalisation $N$.  
In the LMG model, the system undergoes the DQPT in the thermodynamic limit $N \rightarrow \infty$. 
We consider the weakly-coupled regime and
present the dynamical behaviour of Loschmidt amplitude $\mathcal{G}(t) $ for different external field $h$ in Fig.~\ref{fig:DPP_SM} (a).  We clearly see that the Loschmidt amplitude decays rapidly to zero when the external field  $h$ is above the critical field.
The nonanlytical behaviour of  rate function $\gamma(t) $ for a large external field $h$ reveals a dynamical phase transitions to the paramagnetic phases.
The minimal of Loschmidt amplitude is above zero for small $h$, which indicates the system persists a ferromagnetic phase under evolution. 
Fig.~\ref{fig:DPP_SM} (b) shows the system size dependence of minimal Loschmidt amplitude for various $h_x$. We can see that the minimal Loschmidt amplitude appears much earlier with the increasing system size. 
We note that the decay rate $\alpha$ of the trapped ions quantum simulator  can be tuned in the region of $0\leq \alpha \leq 3$  due to the physical interaction, while PQS method could be leveraged to compensate these limitations.
 
 \begin{figure}[t]
\includegraphics[width =1\textwidth]{ 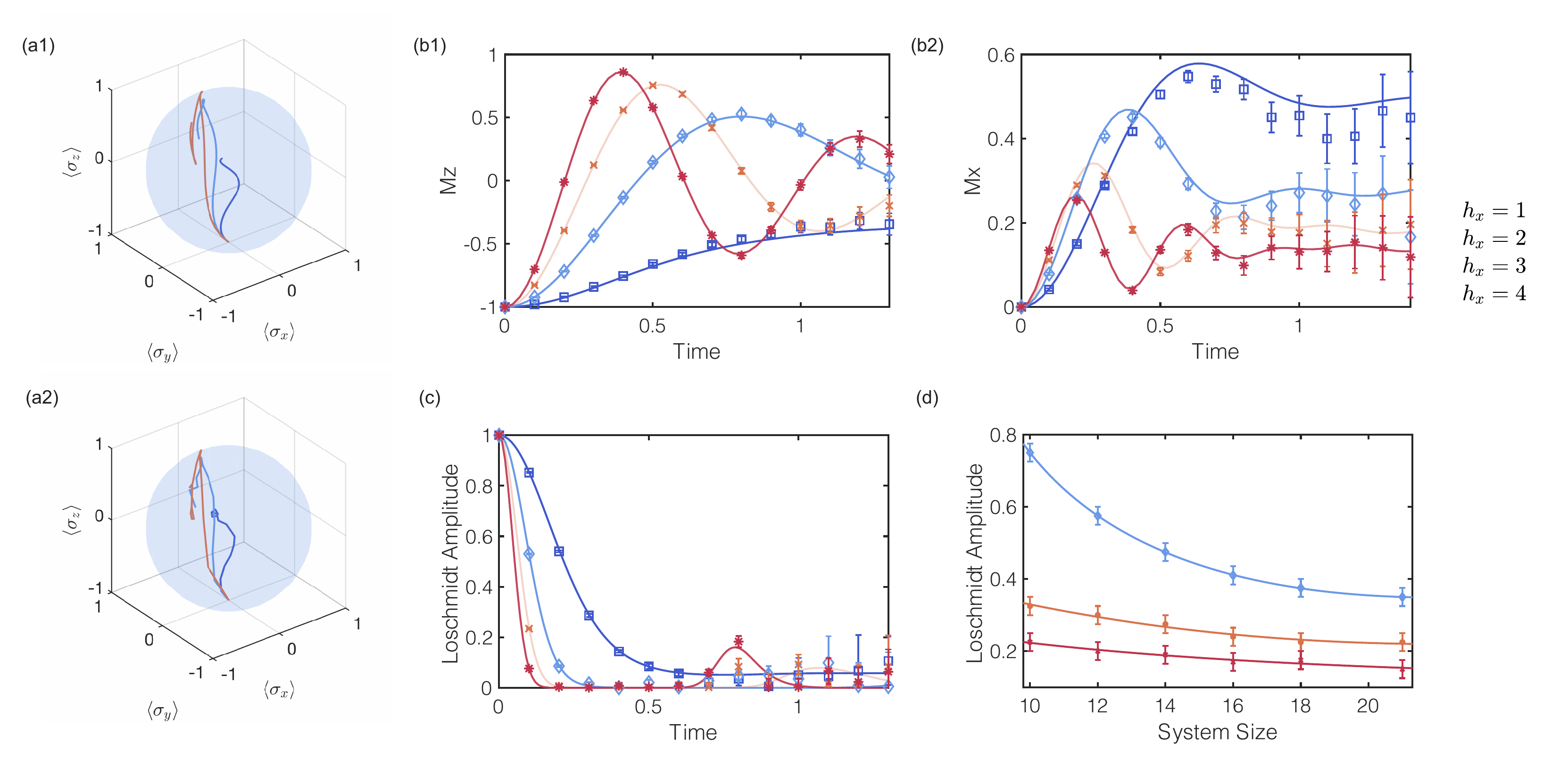}
 \caption{Dynamical quantum phase transition of the  long-range spin chain with full connectivity. 
 The system is initialised to the eigenstate of the Hamiltonian with zero field as $\ket{\psi_0} = \ket{1}^{\otimes N}$, and then the external field along $x$ axis is suddenly turned on at time $t\geq 0$.
 (a) The numerical (a1) and ideal (a2) time evolution of the average spin magnetisation shown in the Bloch sphere for different strengths of the transverse fields $h_x =1,2,3$.
 (b) Time evolution of the averaged magnetisation $M_z = \frac{1}{N}\sum_j \braket{\sigma_j^z(t)}$ (b1) and $M_x = \frac{1}{N} \sum_j \braket{\sigma_j^x(t)}$ (b2) for  different strengths of the transverse fields $h_x =1,2,3,4$.
 The magnetisation $M_z$ and $M_x$ decays rapidly at large field. The external field drives the system from the dynamical ferromagnetic phase to the dynamical paramagnetic phase.
 (c) The Loschmidt amplitude $\mathcal{G}(t) = \left| \braket{\psi_0 |e^{-iHt}|\psi_0}\right|^2$, as an indicator to characterise the dynamical echo back to the initial state for different transverse field strengths $h_x$. (d) System size dependence of the Loschmidt amplitude. The   phase transition appears earlier with larger system size~\cite{zhang2017observation,xu2020probing}.
  Solid lines represents the exact results.
 }
 \label{fig:DPP_SM}
\end{figure}

\begin{figure}[t]
\includegraphics[width =1.0\textwidth]{ 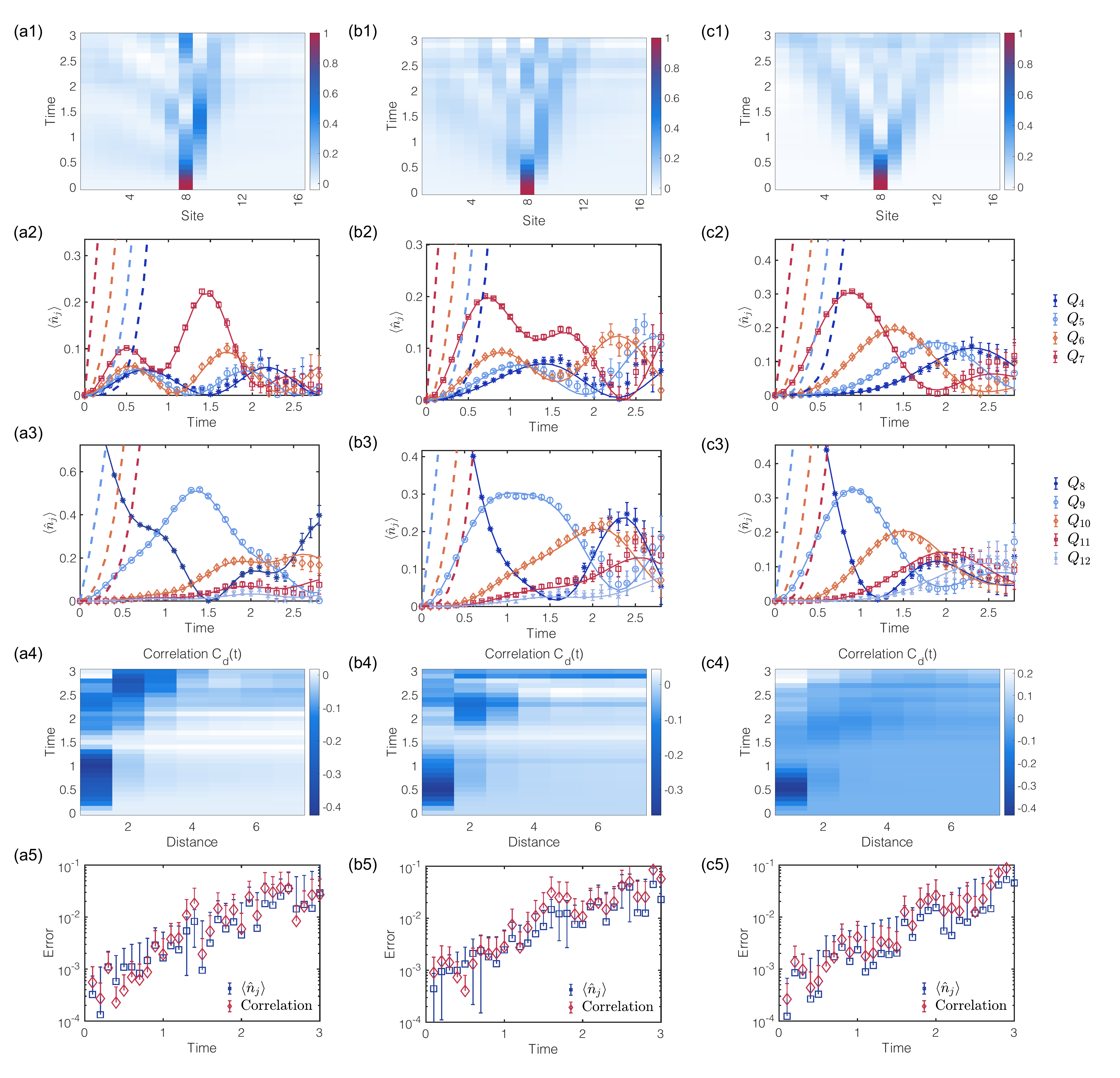}
 \caption{Correlated spin clusters with power law decay interactions $\alpha$ in the subsystems and interactions on the boundary.
 We perturb the systems at the $8$th site as $\ket{\psi_0} = \hat{\sigma}^x_{8} \ket{\psi_0}$ at $t = 0$, and suddenly turn on the interaction $J_{ij}$. Here, we set $J_0 = 1$. 
 (a), (b) and (c) show $\alpha = 0.5,~1,~2$, respectively.
 (a1), (b1) and (c1).
 Dynamics of  magnon quasiparticle excitations $\langle \hat{n}_j \rangle$, related to the local magnetic moment by $\langle \hat{n}_j \rangle = (1-\langle \hat{\sigma_j} \rangle)/2$. (a2)-(c2) and (a3)-(c3) shows the signal of the magnetisation distribution at $4$th-$7$th sites and $8$th-$12$th sites respectively. 
 The nearest-neighbour Lieb-Robinson bounds (dashed) do not capture all the signals for this propagation. 
 (a4)-c(4) shows the averaged two-body correlation functions $ C_d$ from the $8$th site.
 (a5)-(c5) shows the errors for magnetisation and the correlation functions. 
 }
 \label{fig:lr_SM}
\end{figure}

\subsubsection{Propagation of correlations}

Collective behaviour, such as magnetic excitations, emerges from interactions.
% quantum systems is essential to understand the many-body physics and is a core theme of physics. 
These elementary excitations can be described in the quasiparticles picture.
In quantum systems with finite range interactions, the quantum dynamics exhibit a light-cone-like information propagation, and the speed of  information propagation is governed by the interactions of the systems. 
% This collective behaviour can be described in the quasiparticles picture.
% Understanding the quantum dynamics with different elementary excitations would be an interesting direction. 
In the nearest-neighbour interactions, like the results presented in \autoref{appendix:qw}, the propagation of information has a finite maximal velocity $v_g$, the so-called Lieb-Robinson velocity.
If the interactions exponentially decays with increasing distance, from  \autoref{eq:Lieb_bound}, one observe that the change of the expectation of observables under the time evolution is exponentially decreased with the distance $d$, which indicates that information propagates faster than $v_g$ is exponentially suppressed with the distance $d$. This statement exhibits a light-cone-like information propagation analogous the relativistic theory.
The speed of information propagation for power law decay interactions has been experimentally investigated in  Ref.~\cite{jurcevic2014quasiparticle}, which is beyond the light-cone picture. Understanding the effective model to describe quasiparticle exciatations and the propagation of information for general interactions is an interesting direction. In this section, we show the quasiparticle excitations of correlated spin clusters with various interaction strengths using our algorithms.

% We study the propagation two-body correlation in a long-range Ising chain after quantum quench.
We consider the out-of-equilibrium dynamics of a one-dimensional interacting spins with the Hamiltonian $H=H^{\rm loc} + V^{\rm int}$ with the local Hamiltonian and interactions on the boundary as
\begin{equation}
    H_l^{\rm loc} = \sum_{ij} J_{ij} \hat\sigma_{l,i}^x \hat\sigma_{l,j}^x +h \sum_j \hat \sigma_{l,j}^z   ,~~ V^{\rm int} =  J_{0} \hat \sigma_{1,N}^x \hat\sigma_{2,1}^x.
\end{equation} 
Here, $\hat \sigma_{l,i}$ represents Pauli operators acting on the $i$th site of $l$th subsystem, and the interactions obey the power law decay rule as
 $J_{ij}= J_0 |i-j|^{-\alpha}$. 
In the regime of sufficiently large field $h \gg \max(|J_{ij}|)$, the Hamiltonian conserves the total magnetisation and thus could be mapped to the XX model $H = \sum_{ij} J_{ij} \left(\hat \sigma_i^+ \hat\sigma_j^- + \textrm{h.c.}\right)$, similar to the hard-core bosons which conserves the particle numbers with the effective Hamiltonian as $H = \sum_{ij} J_{ij} (\hat a_i^{\dagger} \hat a_j +\textrm{h.c.})$.
For the system with (continuous) transitional symmetry, we can Fourier transform the real-space operator into the operators that are diagonal in momentum space, written as 
$H = \sum_k \omega_k \hat  a_k^{\dagger} \hat  a_{-k}$ where the modes with energies $\omega_k$ has well-defined quasimomentum $k$.
Here, the operator $\hat  a_k^{\dagger}$ creates an excitation with momentum $k$ in the momentum space, and it is related to the original operator by $\hat  a_k^{\dagger} = \sum_{i,k}\hat a_i^{\dagger}$. 
In our simulation, we consider a spin-cluster system and first excite the system by local perturbation, which creates a magnon quasiparticle.
For the system with nearest-neighbour interactions, the energy spectrum has a well-known quadratic dispersion $\omega_k \propto k^2$ in the low energy excitation regime.
While for the spin cluster system, the mode does not have a well-defined momentum, one can determine the energy dispersion $\omega_k$ provided the boundary condition and the interaction $J_{ij}$.
% Note that the constant field $h$ can be neglected due to the magnetisation conservation.

In our numerical simulation, we consider an intermediate regime where the external field is much larger than the maximum interaction strength $J_0$ while it is comparable to the total interaction strength $\tilde J = \sum_{i<j} J_{ij}$.
In this case, 
% the total magnetisation is nearly conserved, but
the field effect cannot be fully negligible when mapped to the XX model and it drives the system to an excited regime.
% This out-of-equilibrium phase could be hard to simulate classically.
In the numerical simulation, we set $J_0 = 1$, while the external field is set as $h = 2 N J_0$ with $N$ being the total sites in the full system. From the simulation results shown below, the total magnetisation is found to be nearly a constant.
Here, we mainly focus on the magnetisation conserved regime, while we can similarly simulate the highly excited regime using the same method. 
We note that in the highly excited regime, i.e., $h \sim J_0$, the quasiparticle picture does not hold and the collective excitations could be very different. The investigation of the interacting physics of the spin clusters is an interesting direction.

Now, we study the excitations and dynamics of quantum information.
We first initialised the state as the eigenstate $\ket{\psi_0} = \ket{0}^{\otimes N}$ of non-interacting $H$ with $J_{ij} = 0$. 
We perturb the systems at the centre ($8$th site) of the spin chain, i.e. $\ket{\psi_0} = \hat{\sigma}^x_{8} \ket{\psi_0}$ at $t = 0$, and suddenly quench the system by turning on the interaction $J_{ij}$.
We show the information propagation with different decay rate $\alpha = 0.5$, $\alpha = 1$ and $\alpha = 2$ in Fig.~\ref{fig:lr_SM} (a), (b) and (c), respectively.
In our numerical experiments, we consider the magnetic moment of each spin and set the sampling number as $2\times 10^5$.
We show the magnetic moment distribution over each site in Fig.~\ref{fig:lr_SM} (a1)-(c1), and the distribution of several neighbour sites $Q_4$ to $Q_{7}$ in (a2)-(c2) with three interaction strengths. The effective model to describe the long-range physics and short-range physics  was discussed in Refs.~\cite{jurcevic2014quasiparticle,monroe2021programmable}. The maximum group velocity is predicted to show a divergent behaviour for the power law decay interactions.
We clearly see that the quasiparticle excitations in the first subsystem propagate much faster as the interaction strength increases ($\alpha$ decreases). 
The  quasiparticle excitations for small $\alpha$ (strongly coupled) appear to be much localised compared with the weakly coupled regime. 
Also, the propagation speed violates the Lieb-Robinson bounds, when considering the nearest-neighbour interaction $\max J_{ij}$ or renormalised interaction $\sum_{ij} J_{ij}$, which indicates that long-range physics cannot be well described by the light-cone propagation with finite group velocity.
In contrast, for the other subsystem unperturbed at the beginning, we observe a different propagation under time evolution, as shown in Fig.~\ref{fig:lr_SM} (a3), (b3) and (c3). This shows an intermediate behaviour of short- and long-range physics of the spin cluster system, which might be captured by the model of nearest-neighbour interactions $ \max J_{ij}$.
We can compare the maximum group velocity in Fig.~\ref{fig:lr_SM} with the divergent behaviour as predicted in Ref.~\cite{jurcevic2014quasiparticle}.
We also note that we can study the dynamical phase transition from the quasiparticle distribution, which can be inferred from the line of $Q_8$, provided the conservation of magnetisation.
We next present the two-body correlation functions $ C_d$ with the spin at the centre, which is expressed as
\begin{equation}
     C_d =   \braket{\hat{\sigma}^z_j \hat{\sigma}^z_{j+d}} - \braket{\hat{\sigma}^z_j} \braket{\hat{\sigma}^z_{j+d} }
\end{equation}
with $j = 8$ at the centre in Fig.~\ref{fig:lr_SM} (a4), (b4) and (c4), showing a quasiparticle picture explained above. 
We leave detailed discussions on the quasiparticle propagation to future work.

\begin{figure}
\includegraphics[width =1.0\textwidth]{ 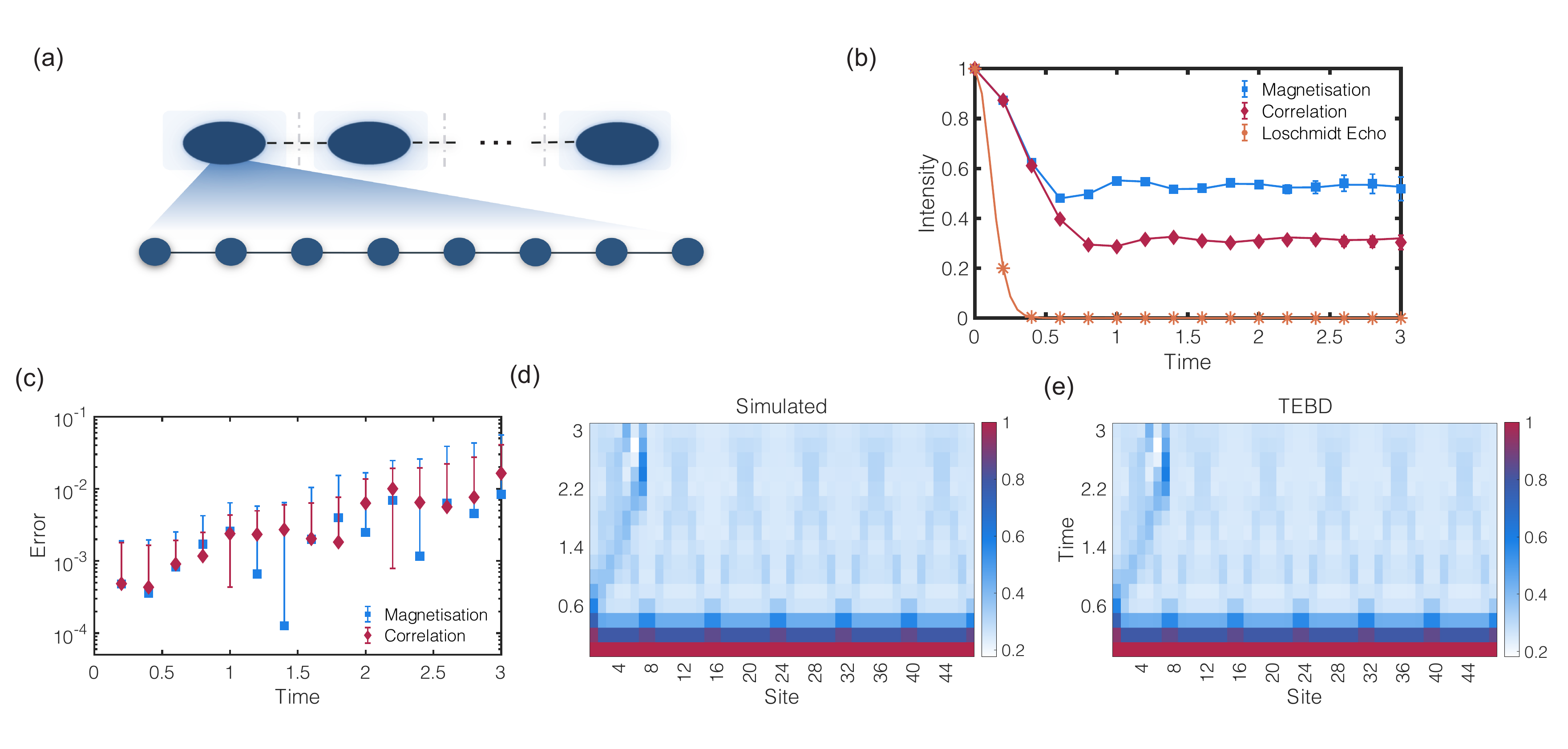}
 \caption{Numerical simulation of dynamics of $1$D $48$-site spin chains with nearest-neighbour correlations. 
  The correlations within each cluster are identical as $J = 1$ while interactions on the boundary are randomly generated from $[0,J/2]$.
   (a) The partitioning sketch. We group 8 adjacent qubits as subsystems. 
 (b) The averaged magnetisation $\frac{1}{N} \sum_{i} \hat{\sigma}_i^z  $, nearest-neighbour correlation functions $\frac{1}{N-1} \sum_{i} \braket{\hat{\sigma}_i^z \hat{\sigma}_{i+1}^z}$ and Loschmidt echo $\mathcal{G}(t)$, compared with TEBD based on the matrix product state representation as a benchmark. (c) shows the errors for the averaged magnetisation and correlation.
 (d) and (e) shows the simulated and exact results for the long-range correlation functions $C_i = \braket{\hat{\sigma}_1^z\hat{\sigma}_{i}^z}$ from $2$th-$48$th site, respectively.
 }
 \label{fig:corr_SM}
\end{figure}

\subsection{Multiple subsystems} 

Finally, we show that our method could be extended to simulate systems consisting of multiple clusters. We consider the out-of-equilibrium dynamics of one-dimensional interacting spin clusters with the Hamiltonian $H=H^{\rm loc} + V^{\rm int}$ with the local Hamiltonian and interactions on the boundary as
\begin{equation}
    H_l^{\rm loc} = J_l \sum_{i}  \hat\sigma_{l,i}^x \hat\sigma_{l,i+1}^x +h \sum_i \hat\sigma_{l,i}^z   ,~~ V_l^{\rm int} =  f_l \hat \sigma_{l,N}^x \hat\sigma_{l+1,1}^x.
\label{eq:multi_spin}
\end{equation} 
Here, $\hat \sigma_{l,i}$ represents Pauli operators acting on the $i$th site of $l$th subsystem.
The interactions in each subsystem are identical $J_l=J_0 = 1$, while interactions between subsystems $f_l$ are generated randomly from $[0,0.5]$. The external field is set as $h = 1$.
This Hamiltonian could be interpreted as a toy model representing certain features of holographic bulk in the $2+ 1$ dimension. 
% The model captures
% the properties of charged black holes. For instance, the
% Penrose diagram of the Reissner-Nordstrstrom or the Kerr black hole in four dimensions is a chain of black or white holes, old or new universes, while quantum information is propagating among different patches of the spacetime.

In the numerical simulation, we consider the spin cluster model consisting of $6$ clusters. 
We simulate up to $48$ qubits with operations only on $8+1$ qubits.
We show the averaged time-evolved magnetisation $M_z = \sum_i \braket{ \hat{\sigma}_i^z}$ and the nearest-neighbour correlation function  $\bar{C_1} = \frac{1}{N-1}\sum_i \braket{ \hat{\sigma}_i^z\hat{\sigma}_{i+1}^z}$ and long-range correlations with the first site
$\bar{C_2} = \frac{1}{N-1}\sum_i \braket{ \hat{\sigma}_1^z\hat{\sigma}_{i}^z}$ and the Loschmidt echo $\mathcal{G}(t)$ in Fig.~\ref{fig:corr_SM} (b).

To benchmark our algorithms, we compare our results with the time-evolving block decimation (TEBD) method, which is a commonly used numerical method to simulate the dynamics of quantum many-body systems based on the matrix product states formalism.
Fig.~\ref{fig:corr_SM} (c) shows that the simulation error can be achieved below $10^{-2}$ at an intermediate time scale.
We remark that according to Corollary~\ref{corollary:one_term}, the explicit decomposition for the example of Eq.~\ref{eq:multi_spin} is optimal.

% We next consider a spin cluster model, which comprises $6$ clusters. 
% The Hamiltonians are 
% \begin{equation}
% H_l^{{\rm{loc}}} = {J}_l \sum\limits_{i = 1}  { {{\hat{\sigma}^x_{l,i}}{\hat{\sigma}^x_{l,i+1}}} }  +\sum\limits_{i = 1}   h{\hat{\sigma}^z_{l,i}}~~
%     V^{{\rm{int}}}=\sum_{j = 1}^{k-1} \sum_{i=1}^4{  {f_{j,i}}{\hat{\sigma}^x_{j,i}}{\hat{\sigma}^x_{j+1,i}}}
% \end{equation}
% with $X_{j,i}$ and $Z_{j,i}$ acting on the cite $(j,i)$.

% We benchmark our method by comparing with the results using TEBD.

% \subsection{Other potential applications}
% In chemistry, the molecules can be described by the Hamiltonian.

\begin{figure*}[t]
\centering
\includegraphics[width =1\linewidth]{ 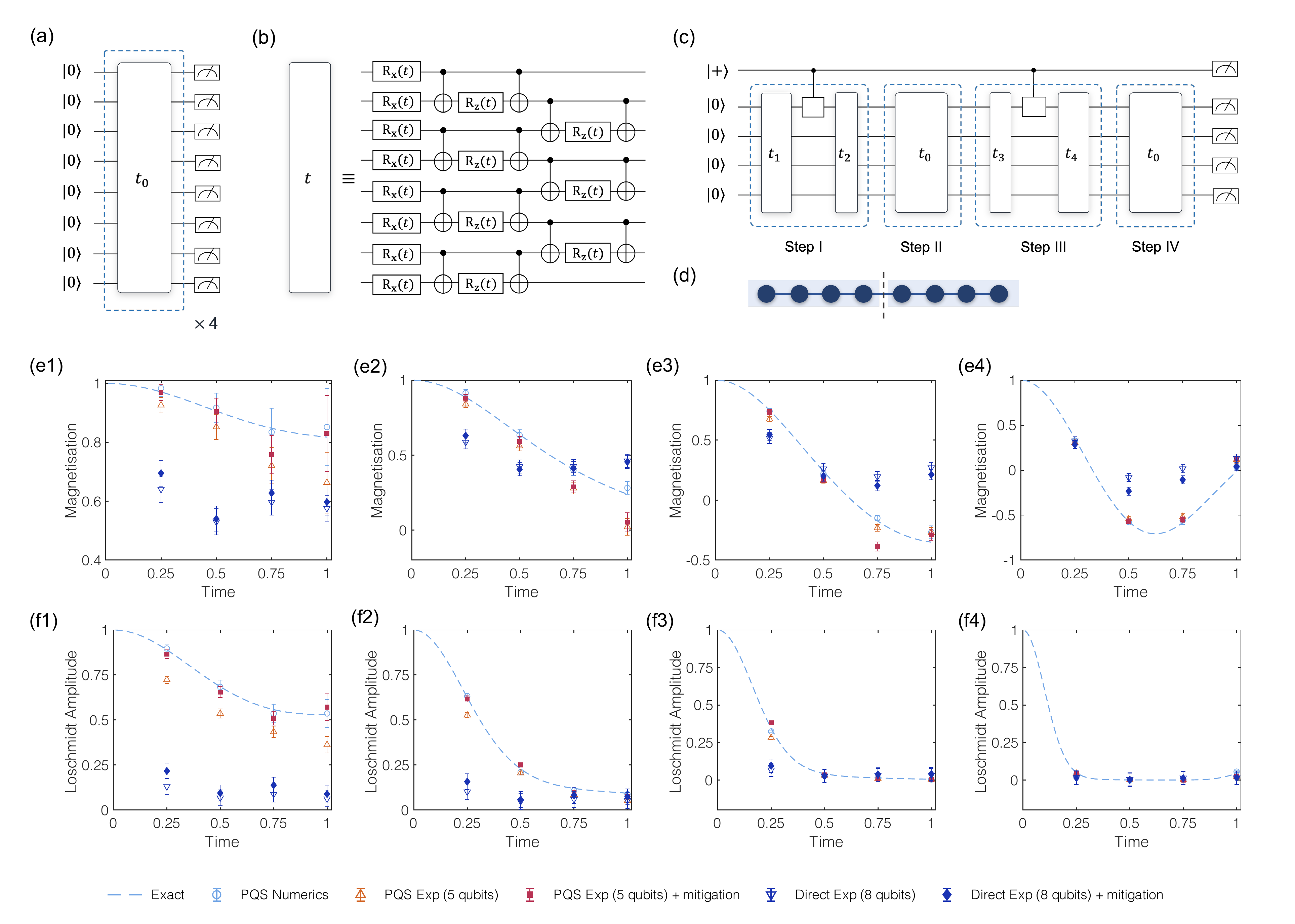}
 \caption{Implementing perturbative quantum simulation on the IBM quantum cloud. 
% IBM quantum devices AntXXX simulation result. 
We consider the DQPT of 8 interacting spins with nearest-neightbour interactions. The initial state $\ket{0}^{\otimes 8}$ is evolved under the Hamiltonian $H = \sum_j \hat{\sigma}_j^z\hat{\sigma}_{j+1}^z + h \sum_j \hat{\sigma}_j^x$. (a)~Quantum circuit implementation for $8$-qubit simulation based on Trotterisation.  (b)~An example for the implementation of PQS to simulate $8$-qubit system with operations on $4+1$-qubits. (c) The circuit block for single-step evolution. (d) The topological geometry for the spin system and the partitioning strategy. 
(e1-e4) The magnetisation along the $z$ direction with $h=0.5,1,1.5,2.5$.
(f1-f4) The Loschmidt amplitude with $h=0.5,1,1.5,2.5$.
We compare the results of exact simulation (dashed line), PQS (numerics, circle), PQS using IBMQ (5 qubits in (c), upper triangle) and the direct simulation using IBMQ (8 qubits in (a), lower triangle). We also show the results using error mitigation for measurement both for PQS (solid square) and direct simulation (solid diamond).
We run $10^3$ samples for PQS and {collect 8192 counts} each samples.
 }
 \label{fig:main_IBM_SM}
\end{figure*}

\section{Experimental Implementation on the IBM quantum devices}
\label{appendix:IBM_implementation}

\subsection{Experimental results}
We implement our perturbative quantum simulation algorithm on the IBM quantum cloud. We consider the $8$-qubit one-dimensional Ising Hamiltonians
\begin{equation}
    H = \sum_{i=1}^7 \hat \sigma_i^z \hat \sigma_{i+1}^z + h\sum_{j=1}^8 \hat \sigma_j^x, 
\end{equation}
with nearest-neighbour interaction and a transverse magnetic field with different strength $h$. Starting from an eigenstate of $H$ with $h=0$, $\ket{\psi(0)}=\ket{0}^{\otimes 8}$ , we evolve the state from time $T = 0$ to $1$ and observe the dynamical quantum phase transition. At time $t\in[0,1]$, we focus on the expectation value of the spin operator $M_z = \sum_{j=1}^8 \hat \sigma_j^z/8$ and the Loschmidt amplitude $\mathcal G(t)=|\braket{\psi(0)|e^{-iHt}|\psi(0)}|^2$, which is equivalent to evaluating the state overlap between $\ket{\psi(0)}$ and $\ket{\psi(t)}=e^{-iHt}\ket{\psi(0)}$. 

To get the exact time-evolved state, we consider the Trotterisation product formula with four timesteps. Specifically, we have
\begin{equation}
\ket{\psi(t)} = \left( \prod_{j=1}^8e^{-ih\delta t\hat \sigma_j^x }\prod_{j=1}^7e^{-i\delta t\hat \sigma_j^z \hat \sigma_{j+1}^z}\right)^{t/\delta t},
\end{equation}
with $t \in \{0.25, 0.5, 0.75, 1\}$ and $\delta t = 0.25$. Each term $e^{-i\delta t\hat \sigma_j^z \hat \sigma_{j+1}^z} = CX_{j, j+1}\textrm{R}_z(2\delta t, j+1)CX_{j, j+1}$ could be realised with a single qubit rotation gate $\textrm{R}_z(2\delta t, j+1) = e^{-i\delta t\hat \sigma_{j+1}^z}$ sandwiched by two controlled-X gates $CX_{j, j+1}$ and each $e^{-ih\delta t\hat \sigma_j^x }=\textrm{R}_z(2h\delta t, j)$  is a single qubit gate. As shown in Fig.~\ref{fig:main_IBM_SM} (c), for each step, all the single qubits gates are implemented in parallel and the two qubit gates are realised with depth $d=2$. 
% Therefore, to directly simulate the time evolution of the 8-qubit Ising Hamiltonian to time 1, we need $7\times 2\times 4=56$ controlled-X gates and $(7+8)\times 4 = 60$ single qubit gates. 
We note that the Trotter error is negligible {(much less than $10^{-2}$
% , see Fig.~\ref{fig:main_IBM_SM} and  Fig.~\ref{fig:main_IBM} in the main text
)}. 

With our perturbative quantum simulation method, we only need to apply operations on $4+1$ qubits. We truncate the maximal number of decay events to four and we can see that the truncation error is small. {When a decay event} happens at time $t$, say $t=0.1$, we further divide  the Trotter step from 0 to 0.25 into two steps, i.e., $[0,0.1]$ and $[0.1,0.25]$. Then we insert a controlled-Z operation with the control qubit being the ancilla and the target being the first (last) qubit. As shown in Fig.~\ref{fig:main_IBM_SM} (b), we design the circuit in a similar way if we have multiple decay events. 
% Therefore, in the worst case with four decay events, we need $3\times 2\times 8+4=52$ two qubit gates and $(3+4)*8+2=58$ single qubit gates.
While the quantum circuit could be further optimised with fewer gates, it is sufficient for demonstrating the power of our PQS method.

\begin{figure*}[t]
\centering
\includegraphics[width =1\linewidth]{ 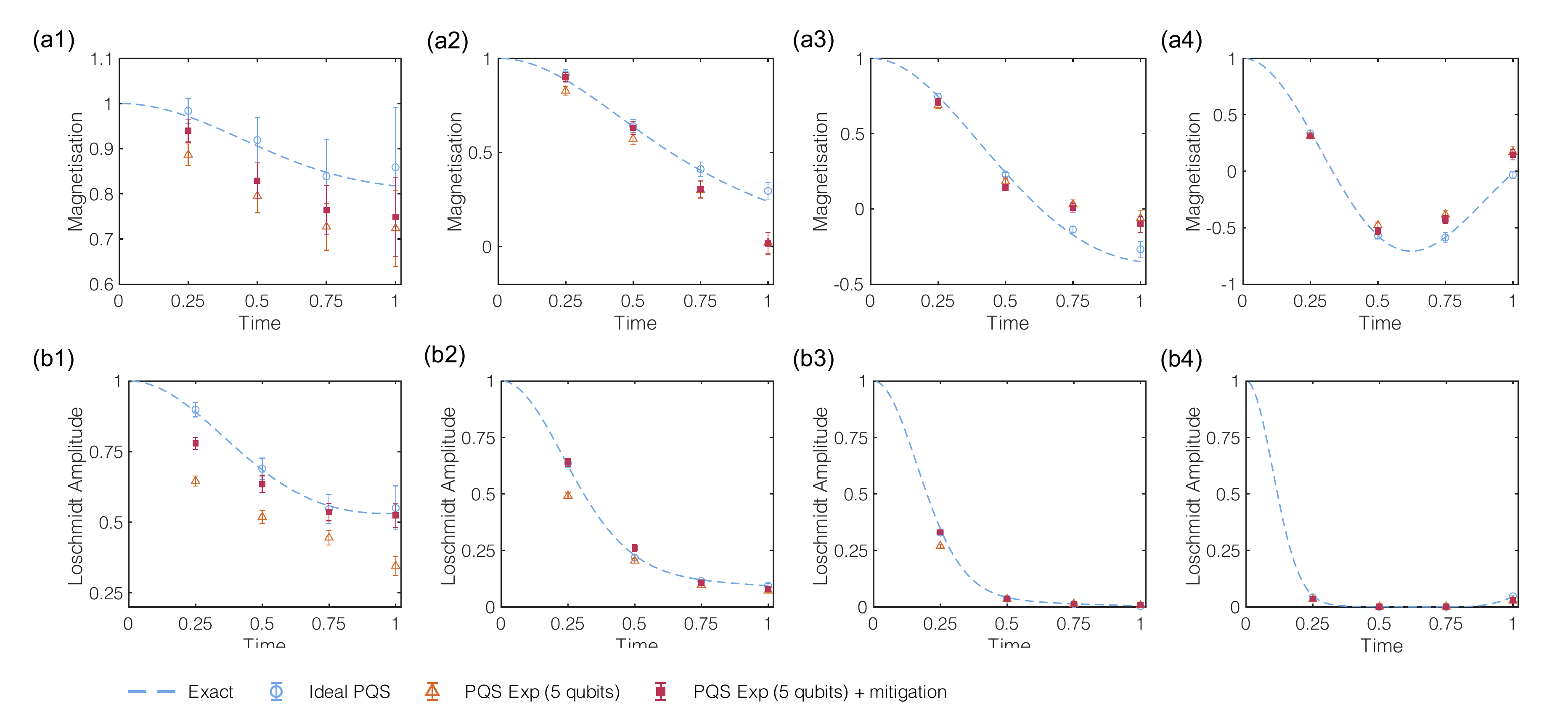}
 \caption{
%  [XXX to be added]
Implementing perturbative quantum simulation on the IBM quantum cloud with fewer samples and less optimised quantum circuit. We run $10^3$ samples for PQS and {collect 128 counts} each samples. The quantum circuit for each time evolution block is different from Fig.~\ref{fig:main_IBM_SM} (c), where we apply the two-qubit gates sequentially with depth $7$. 
 (a1-a4) The magnetisation along the $z$ direction with $h=0.5,1,1.5,2.5$.
(b1-b4) The Loschmidt amplitude with $h=0.5,1,1.5,2.5$.
We compare the results of exact simulation (dashed line), PQS (numerics, circle), PQS using IBMQ (5 qubits in (c), upper triangle). We also show the results using error mitigation for measurement both for PQS (solid square).
 }
 \label{fig:IBM_PQS_SM2}
\end{figure*}

We implement the direct 8-qubit simulation and our  5-qubit PQS method using the IBM Q Experience. 
The processor employed to conduct the direct 8-qubit simulation is `ibmq\_16\_melbourne', which has 16 qubits with $T_2$ time ranging from $18\sim105\mu s$, CNOT gate error $3.3 \times 10^{-2}$ and read-out error $4.7\times 10^{-2}$. 
The processor employed to conduct the 5-qubit PQS method is `ibmq\_santiago', which has 5 qubits with $T_2$ time ranging from $66.9\sim143\mu s$, CNOT gate error $7.1 \times 10^{-3}$ and read-out error $1.7\times 10^{-2}$. 
The circuits are implemented through {Qiskit~\cite{aleksandrowicz2019qiskit}}, a python-based software development kit for working with OpenQASM and the IBM Q  processors. 
The IBM cloud admits multiple job submissions with each job consisting a maximal of 72 circuits, where each circuit is fixed and allows 8192 single-shot measurements.

We show the experimental results in Fig.~\ref{fig:main_IBM_SM}(e, f). We consider the external field along the $x$ direction with $h=0.5,1,1.5,2.5$, and compare the results of exact simulation (dashed line), PQS (numerics, circle), PQS using IBMQ (5 qubits in Fig.~\ref{fig:main_IBM_SM} (b), upper triangle) and the direct simulation using IBMQ (8 qubits in Fig.~\ref{fig:main_IBM_SM} (a), lower triangle). 
For each data point of the direct simulation, we run 16 circuits with $8192*16$ single-shot measurements. For the PQS method, we consider 1024 trajectories with each trajectory corresponding to a circuit measured $8192$ times. We note that even though the number of samples of the PQS method is much larger than the number of samples for the direction simulation method, the shot noise is much smaller than the error caused by device imperfections. We could also use a smaller number of samples ($128$ samples) for each trajectory of the PQS method, and we observe similar results as shown in Fig.~\ref{fig:IBM_PQS_SM2}. We note that the simulation results are not the same because we run a less optimised circuit of  the IBM processor at a different time.
We also apply error mitigation for measurements, which increases the simulation accuracy. The measurement error mitigation is implemented by running a set of circuits with different computational basis input states and computational basis measurement. Then we obtain the calibration matrix and apply its inverse to correct measurement errors~\cite{bravyi2021mitigating,sun2021mitigating}. 
From our simulation result, we observe that the PQS method outperforms the direct simulation.
This is because the five-qubit `ibmq\_santiago' processor has more accurate operations than the `ibmq\_16\_melbourne' processor. 
Since our method requires running on a small quantum computer with a relatively low circuit depth, it could be applied to benchmarking the large-scale quantum devices, which may have more errors.

\subsection{Analysis of noise robustness}

We briefly discuss why PQS is more robust to noise for simulating general systems. 
Suppose we aim to simulate the time evolution of the Ising Hamiltonian with $Ln$ qubits. Conventional approaches require $2Ln-2$ two qubit gates for each Trotter step, whereas PQS only needs $2n-2$ two qubit gates if we consider $L$ number of $n$-qubit clusters. 
Suppose the fidelity of each two qubit gate is $1-\varepsilon$, then the infidelities of the conventional method and PQS is $1-(1-\varepsilon)^{2Ln-2}$ and $1-(1-\varepsilon)^{2n-2}$, respectively. In the regime of small $\varepsilon$ and  $n\varepsilon$ and relatively large $n$, the state infidelity using PQS is approximately $n\varepsilon$, which is $L$ times smaller than $Ln\varepsilon$ using conventional quantum simulation methods. For example, when $L=2$, the infidelity will be half of that obtained from conventional quantum simulation. Therefore, our method not only allows the simulation of larger systems, it also effectively increases the simulation accuracy.

% \bibliography{bib_simulation}
% \bibliographystyle{SciAdv}

\end{document}